\documentclass[preprint]{elsarticle} 
\usepackage[applemac]{inputenc}
\usepackage[english]{babel}
\usepackage{pdfsync,hyperref}
\usepackage{verbatim}
\usepackage{amsmath,amsthm,amssymb,bussproofs,tikz,stmaryrd,mathtools}
\usepackage[emu]{cmll}
\usepackage{fouriernc}
\usepackage{hyperref}
\usepackage{subfigure}
\usepackage{rotating}

\usetikzlibrary{patterns}

\usepackage{macros-revised}

\title{Interaction Graphs: Additives}

\author[ts]{Thomas Seiller\corref{cor1}\fnref{fn1}} \ead{thomas.seiller@inria.fr}
\cortext[cor1]{Corresponding author}
\fntext[fn1]{This work was partially supported by the ANR-10-BLAN-0213 LOGOI.}
\address[ts]{INRIA. Mailing address: Université de Savoie, LAMA UMR 5127 CNRS, Bâtiment Chablais, Campus Scientifique, 73376 Le Bourget-du-Lac Cedex}

\begin{document}

\begin{abstract}
Geometry of Interaction (GoI) is a research program initiated by Jean-Yves Girard which aims at defining a semantics of linear logic proofs accounting for the dynamical aspects of cut elimination. We present here a parametrised construction of a Geometry of Interaction for Multiplicative Additive Linear Logic (MALL) in which proofs are represented by families of directed weighted graphs. Contrarily to former constructions dealing with additive connectives \cite{goi3,goi5}, we are able to solve the known issue of obtaining a denotational semantics for MALL by introducing a notion of observational equivalence. Moreover, our setting has the advantage of being the first construction dealing with additives where proofs of MALL are interpreted by finite objects. The fact that we obtain a denotational model of MALL relies on a single geometric property, which we call the \emph{trefoil property}, from which we obtain, for each value of the parameter, \emph{adjunctions}. We then proceed to show how this setting is related to Girard's various constructions: particular choices of the parameter respectively give a combinatorial version of his latest GoI \cite{goi5}, a refined version of older Geometries of Interaction \cite{goi1,goi2,goi3}, and even a generalisation of his \emph{multiplicatives} \cite{multiplicatives} construction. This shows the importance of the \emph{trefoil property} underlying our constructions since all known GoI construction to this day rely on particular cases of it.
\end{abstract}
\maketitle

\section{Introduction}

\paragraph{\textbf{The Geometry of Interaction program \cite{towards}}} It was introduced by Girard a couple of years after his discovery of Linear Logic \cite{ll}. It aims at giving a semantics of linear logic proofs that would account for the dynamical aspects of cut-elimination, hence of computation through the proofs-as-program correspondence. Informally, a Geometry of Interaction (GoI) consists in:
\begin{itemize}
\item a set of mathematical objects — paraproofs — that will contain, among other things, the interpretations of proofs (or $\lambda$-terms);
\item a notion of execution that will represent the dynamics of cut-elimination (or $\beta$-reduction \cite{danos-phd,malacaria-regnier}).
\end{itemize}
Then, from these basic notions, one should be able to "reconstruct" the logic from the way the paraproofs interact:
\begin{itemize}
\item From the notion of execution, one defines a notion of orthogonality between the paraproofs that will allow to define formulas — types — as sets of paraproofs closed under bi-orthogonality (a usual construction in realisability). The notion of orthogonality should be thought of as a way of defining negation based on its computational effect. 
\item The connectives on formulas are defined from "low-level" operations on the paraproofs, following the idea that the rules governing the use of a connective should be defined by the way this connective acts at the level of proofs, i.e. by its computational effect.
\end{itemize}

Throughout the years, Girard defined several such semantics, mainly based on the interpretation of a proof as an operator on an infinite-dimensional Hilbert space. In particular, two such constructions offer a treatment of additive connectives of linear logic \cite{goi3,goi5}. It is also worth noting that the first version of GoI \cite{goi1} was used to analyse lambda-calculus' $\beta$-reduction \cite{AbadiGonthierLevy92b,danosregnierlocalasynchronous}, elucidating Lamping's optimal reduction \cite{Lamping90}.

The latest version of GoI \cite{goi5}, from which this work is greatly inspired, is related to quantum coherent spaces \cite{qcs}, which suggest future applications to quantum computing. Moreover, the great generality and flexibility of the definition of exponentials also seem promising when it comes to the study of complexity. Some results in this direction were already obtained: using a new technique proposed by Girard \cite{normativity}, the author obtained, in a joint work with Clément Aubert, new characterisations of the computational complexity classes \textbf{co-NL} \cite{seiller-conl} and \textbf{L} \cite{seiller-lsp} as sets of operators in the hyperfinite type $\text{II}_{1}$ factor.

\paragraph{\textbf{Interaction Graphs}} Departing from the realm of infinite-dimensional vector spaces and linear maps between them, we propose a graph-theoretical GoI where proofs are interpreted by finite objects\footnote{Even though the graphs we consider can have an infinite set of edges, linear logic proofs are represented by finite graphs (disjoint unions of transpositions).}. In this framework, it is possible to define the multiplicative and additive connectives of Linear Logic. Although not the first such work proposing a combinatorial formulation of GoI constructions \cite{danosregnierlocalasynchronous,goilca,catgoi}, Interaction Graphs is the first work providing such an approach to Girard's hyperfinite GoI \cite{goi5}. Another novelty lies in the fact that the construction is parametrised by a map from the interval $]0,1]$ to $\mathbb{R}_{\geqslant 0}\cup\{\infty\}$, and therefore yields not just one but a whole family of models. %

We will show how, from any of these models, one can obtain a $\ast$-autonomous category with $\parr\not\cong\otimes$ and $1\not\cong\bot$, i.e. a non-degenerate denotational semantics for Multiplicative Linear Logic (MLL). However, as in all the versions of GoI dealing with additive connectives, our construction of additives does not define a categorical product. We solve this issue by introducing a notion of \emph{observational equivalence} within the model. We are then able to define a categorical product from our additive connectives when considering classes of observationally equivalent objects, obtaining a denotational semantics for Multiplicative Additive Linear Logic (MALL).

One important point in this work is the fact that all results rely on a single geometric property we call the \emph{trefoil property}. This property ensures the fours following facts:
\begin{itemize}
\item we obtain a $\ast$-autonomous category; this is a consequence of the \emph{three-term adjunction} obtained as a corollary of the trefoil property;
\item the observational equivalence is a congruence on this category;
\item the quotiented category inherits the $\ast$-autonomous structure;
\item the quotiented category has a full subcategory with products.
\end{itemize}

We then proceed to show how our framework is related to Girard's versions of GoI by looking at models obtained for particular choices of the parameter. Indeed, a first choice of map gives us a model that can be embedded in Girard's GoI5 framework \cite{goi5}. It can be shown that it is a combinatorial version of (the multiplicative additive) fragment of GoI5, offering insights on its notion of orthogonality and constructions. On the other hand, a second choice of map defines a model where orthogonality is defined by nilpotency: our construction thus defines in this case a (refined) version of older Geometries of Interaction \cite{goi1,goi2,goi3}. We also show how a special case of our construction can be related formally to Girard's \emph{Multiplicatives} construction, the very first Geometry of Interaction construction. These results show that the framework of Interaction Graphs captures the different constructions studied by Girard, exposes the one geometric identity underlying them (the trefoil property), gives new insights about those and shows how they relate to each other.

In the last section, we show how the intuitions gained by Interaction Graphs can be used to show results in the hyperfinite Geometry of Interaction of Girard. For this, we show that a proof of a technical property (Proposition \ref{counterexgraphs}) obtained on graphs can be adapted in order to obtain an equivalent property (Propositions \ref{contreexgdi52} and \ref{counterexgdi51}) in Girard's setting. This property, as it turns out, reveals a small mistake in Girard's paper. Our constructions of additives on graphs can however be used to correct this mistake which turns out to be of small importance.

\paragraph{\textbf{About terminology}}

We will use in this paper some unusual terminology, extending the terminology of the author's previous paper \cite{seiller-goim} and inspired by Girard's \cite{goi5}. It is the author's wish to keep, throughout the Interaction Graphs papers' series, a coherent terminology which does not convey incorrect intuitions. 

Let us illustrate how usual terminology may lead to incorrect intuitions. Although some well-known syntactic notions might be related to the specific objects introduced in the paper, the latter are generalisations of the former, and it may convey the wrong ideas to collapse the terminologies. For instance the notion of \emph{project}, introduced later in this paper, generalises the notion of proof in that some specific projects will be interpretations of proofs. However, a project is not in general a representation of a proof.  Other usual terminology from a more semantic tradition relates to the notion of project such as game semantics' \emph{strategies} or classical realisability's \emph{$\lambda_{c}$-terms}, which are more general than \emph{winning strategies} and \emph{quasi-proofs} respectively and characterise (and sometimes capture) those objects interpreting proofs. Although one may have borrowed terminology from those semantic traditions, it might once again have conveyed some incorrect intuitions. For instance, $\lambda_{c}$ terms may contain cuts ($\beta$-redexes), while projects should be understood as some sort of \emph{normal form}. Concerning game semantics, the notion of \emph{conducts}, which are interpretations of proofs, is quite different from the notion of \emph{game}: while a strategy is defined on a given game, a project will belong to several conducts. This latter difference is styled by Girard as an opposition between \emph{\enquote{a posteriori} typing} (conducts) and \emph{\enquote{a priori} typing} (games), in some ways reminiscent of the distinction between \emph{Curry-style} and \emph{Church-style} typing. Notice that it makes subtyping a natural feature of Interaction Graphs; in this particular aspect the GoI and realisability approaches are again quite similar.

We show (\autoref{terminologyfig}) a table proposing translations of the terminology used in this paper and more usual terminology from both syntactic and semantic traditions. We write in boldface the \enquote{usual} words whose meaning is closest to the meaning of the interaction graph term on the left; note that this idea of closeness is quite subjective as it reflects the author's understanding of the notions used in the paper. We stress that the boldface terms are borrowed from different lines of work, which is another argument for a specific terminology: the combined use of notions from game semantics, proof nets and realisability would be more confusing than using brand new terminology. The symbol \enquote{---} denotes the lack of adequate notion.

\begin{figure}
\centering
\subfigure[Correspondence with syntactic traditions]{
\begin{tabular}{|c||c|c|}
\hline
Interaction Graphs & Sequent Calculus  & Proof Nets\\
\hline\hline
project (\autoref{projectdef}) & --- & proof structure \\\hline
wager (\autoref{wagerdef})& --- & --- \\
carrier (\autoref{projectdef}) & --- & --- \\\hline\hline
execution/cut (\autoref{executiondef}) & normalisation & \textbf{normalisation}\\\hline
$\de{Fax}$ (\autoref{faxdef}) & axiom rule & \textbf{axiom link} \\\hline
successful project (\autoref{successdef}) & proof  & proof net \\\hline
conduct (\autoref{conductdef}) & formula & formula \\\hline
behaviour (\autoref{behaviourdef}) & \textbf{\enquote{linear} formula} & \textbf{\enquote{linear} formula}\\\hline
orthogonality (\autoref{orthogonalitydef}) & --- & \textbf{correctness criterion} \\\hline
observational equivalence (\autoref{equivalencedef}) & --- & --- \\\hline
\end{tabular}
}
\subfigure[Correspondence with semantics traditions]{
\begin{tabular}{|c||c|c|}
\hline
Interaction Graphs & Game Semantics & Realisability\\
\hline\hline
project (\autoref{projectdef}) & \textbf{strategy} & $\lambda_{c}$-term\\\hline
wager (\autoref{wagerdef})& --- & --- \\
carrier (\autoref{projectdef}) & \textbf{arena} & --- \\\hline\hline
execution/cut (\autoref{executiondef}) & \textbf{composition} & $\beta$-reduction\\\hline
$\de{Fax}$ (\autoref{faxdef}) & \textbf{copy-cat} & identity\\\hline
successful project (\autoref{successdef}) & winning strategy & \textbf{quasi-proof} \\\hline
conduct (\autoref{conductdef}) & game & \textbf{type} \\\hline
behaviour (\autoref{behaviourdef}) &  --- & --- \\\hline
orthogonality (\autoref{orthogonalitydef}) & --- & \textbf{orthogonality} \\\hline
observational equivalence (\autoref{equivalencedef}) & --- & \textbf{contextual equivalence} \\\hline
\end{tabular}
}
\caption{Tables showing terminologies' (approximate) correspondence}\label{terminologyfig}
\end{figure}

\section{Graphs and Cycles}\label{basicssection}

\subsection{Basic Definitions}

We first recall some definitions and notations of our earlier paper \cite{seiller-goim}. We include the proofs of Propositions \ref{countingeqclas} and \ref{assoc} to keep this section self-contained.

\begin{definition}
A \emph{directed weighted graph} is a tuple $G=(V^{G},E^{G},s^{G},t^{G},\omega^{G})$, where $V^{G}$ is the set of vertices, $E^{G}$ is the set of edges, $s^{G}$ and $t^{G}$ are two functions from $E^{G}$ to $V^{G}$, the \emph{source} and \emph{target} functions, and $\omega^{G}$ is a function\footnote{We chose here to work with the set $]0,1]$ as the set of possible weights for the edges of our graphs. The fact that we chose this particular set is used in Sections \ref{truthsection} and \ref{embedsection}. However, the constructions and results of Sections \ref{basicssection}, \ref{additivesection}, \ref{gdisection} and \ref{denotsection} are independent of this choice and could be performed with graphs with edges weighted in any semi-group (associative magma) $(\Omega,\times)$ (the multiplication being necessary to define the weight of paths, see Definition \ref{execdef}).} $E^{G} \rightarrow ]0,1]$.

In this paper, we will work with directed weighted graphs where the set of vertices is \emph{finite}, and the set of edges is \emph{finite or countably infinite}.
\end{definition}

We will write $E^{G}(v,w)$ the set of $e\in E^{G}$ satisfying $s^{G}(e)=v$ and $t^{G}(e)=w$. Moreover, we will sometimes forget the exponents when the context is clear. In the next definition, and throughout the paper, we will use the notation $\cdot\disjun\cdot$ to denote both the disjoint union of sets or the co-pairing of functions, i.e. if $f:E\rightarrow S$ and $g:F\rightarrow S$, the function $f\disjun g$ has domain $E\disjun F$ and codomain $S$.

\begin{definition}[Plugging]
Given two graphs $G$ and $H$, we define the graph $G\bicol H$ as the union graph $(V^{G}\cup V^{H},E^{G}\disjun E^{H},s^{G}\disjun s^{H},t^{G}\disjun t^{H},\omega^{G}\disjun\omega^{H})$ of $G$ and $H$, together with a coloring function $\delta$ from $E^{G}\disjun E^{H}$ to $\{0,1\}$ such that 
\begin{equation*}
\left\{\begin{array}{l}
\delta(x)=0\text{ if }x\in E^{G}\\
\delta(x)=1\text{ if }x\in E^{H}
\end{array}\right.
\end{equation*}
We refer to $G\bicol H$ as the \emph{plugging of $G$ and $H$.}
\end{definition}

\begin{example}
Let us consider the graphs $F$, $G$ and $H$ illustrated in \autoref{examplegraphs}. The plugging $F\bicol G$ and the plugging $F\bicol H$ are represented in \autoref{exampleplug}, where edges $e$ such that $s(e)=0$ are shown above the vertices and edges such that $s(e)=1$ are shown below them.
\end{example}

\begin{figure}
\centering
\subfigure[The graph $F$]{
\begin{tikzpicture}

\node (1) at (0,0) {1};
\node (2) at (1,0) {2};
\node (3) at (2,0) {3};
\node (4) at (3,0) {4};

\draw[->] (1) -- (2) {};
\draw[->] (2) .. controls (1,1) and (3,1) .. (4) {};
\draw[->] (4) -- (3) {};
\draw[->] (3) .. controls (2,1) and (0,1) .. (1) {};
\end{tikzpicture}
}
\subfigure[The graph $G$]{
\begin{tikzpicture}

\node (5) at (6,0) {1};
\node (6) at (7,0) {2};

\draw[->] (5) .. controls (5,1) and (7,1) .. (5) {};
\draw[->] (6) .. controls (6,1) and (8,1) .. (6) {};

\end{tikzpicture}
}
\subfigure[The graph $H$]{
\begin{tikzpicture}

\node (5) at (6,0) {1};
\node (6) at (7,0) {2};
\node (a) at (5.5,0) {};
\node (b) at (7.5,0) {};

\draw[->] (5) .. controls (6,0.8) and (7,0.8) .. (6) {};
\draw[->] (6) .. controls (7,1.2) and (6,1.2) .. (5) {};

\end{tikzpicture}
}
\caption{Three weighted graphs $F$, $G$ and $H$.\label{examplegraphs}}
\end{figure}
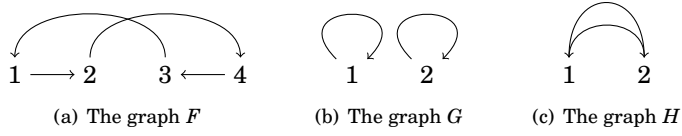

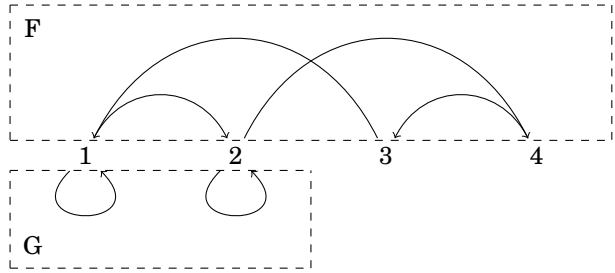
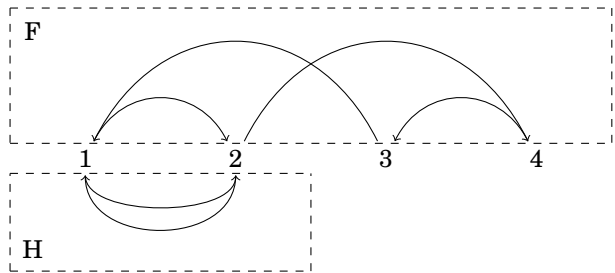
\begin{figure}
\centering
\subfigure[The graph $F\bicol G$]{
\begin{tikzpicture}
\node (1) at (0,0) {1};
\node (2) at (2,0) {2};
\node (3) at (4,0) {3};
\node (4) at (6,0) {4};

\draw[->] (1) ..controls (0.5,1) and (1.5,1) .. (2) {};
\draw[->] (2) .. controls (3,2) and (5,2) .. (4) {};
\draw[->] (4) .. controls (5.5,1) and (4.5,1) .. (3) {};
\draw[->] (3) .. controls (3,2) and (1,2) .. (1) {};

\draw[->] (1) .. controls (-1,-1) and (1,-1) .. (1) {};
\draw[->] (2) .. controls (1,-1) and (3,-1) .. (2) {};

\draw[-,dashed] (-1,2) -- (7,2) {};
\draw[-,dashed] (-1,2) -- (-1,0.2) {};
\draw[-,dashed] (-1,0.2) -- (7,0.2) {};
\draw[-,dashed] (7,2) -- (7,0.2) {};
\node (f) at (-0.7,1.7) {F};

\draw[-,dashed] (-1,-1.5) -- (-1,-0.2) {};
\draw[-,dashed] (-1,-1.5) -- (3,-1.5) {};
\draw[-,dashed] (3,-1.5) -- (3,-0.2) {};
\draw[-,dashed] (-1,-0.2) -- (3,-0.2) {};
\node (g) at (-0.7,-1.2) {G};
\end{tikzpicture}
}
\subfigure[The graph $F\bicol H$]{
\begin{tikzpicture}
\node (1) at (0,0) {1};
\node (2) at (2,0) {2};
\node (3) at (4,0) {3};
\node (4) at (6,0) {4};

\draw[->] (1) ..controls (0.5,1) and (1.5,1) .. (2) {};
\draw[->] (2) .. controls (3,2) and (5,2) .. (4) {};
\draw[->] (4) .. controls (5.5,1) and (4.5,1) .. (3) {};
\draw[->] (3) .. controls (3,2) and (1,2) .. (1) {};

\draw[->] (1) .. controls (0,-0.8) and (2,-0.8) .. (2) {};
\draw[->] (2) .. controls (2,-1.2) and (0,-1.2) .. (1) {};

\draw[-,dashed] (-1,2) -- (7,2) {};
\draw[-,dashed] (-1,2) -- (-1,0.2) {};
\draw[-,dashed] (-1,0.2) -- (7,0.2) {};
\draw[-,dashed] (7,2) -- (7,0.2) {};
\node (f) at (-0.7,1.7) {F};

\draw[-,dashed] (-1,-1.5) -- (-1,-0.2) {};
\draw[-,dashed] (-1,-1.5) -- (3,-1.5) {};
\draw[-,dashed] (3,-1.5) -- (3,-0.2) {};
\draw[-,dashed] (-1,-0.2) -- (3,-0.2) {};
\node (g) at (-0.7,-1.2) {H};
\end{tikzpicture}
}
\caption{Examples of plugging.}\label{exampleplug}
\end{figure}

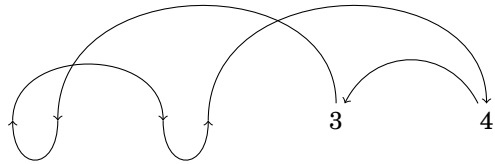
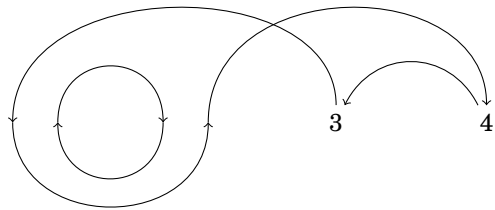
\begin{figure}
\centering
\subfigure[Alternating paths in $F\bicol G$\label{FbicolG}]{
\begin{tikzpicture}
\node (1) at (0,0) {};
	\node (1G) at (-0.3,0) {};
	\node (1D) at (0.3,0) {};
\node (2) at (2,0) {};
	\node (2G) at (1.7,0) {};
	\node (2D) at (2.3,0) {};
\node (3) at (4,0) {3};
\node (4) at (6,0) {4};

\draw[->] (4) .. controls (5.5,1) and (4.5,1) .. (3) {};
\draw[->] (3) .. controls (4,2) and (0.3,2) .. (0.3,0) {};
\draw[->] (0.3,0) .. controls (0.3,-0.7) and (-0.3,-0.7) .. (-0.3,0) {};
\draw[->] (-0.3,0) .. controls (-0.3,1) and (1.7,1) .. (1.7,0) {};
\draw[->] (1.7,0) .. controls (1.7,-0.7) and (2.3,-0.7) .. (2.3,0) {};
\draw[->] (2.3,0) .. controls (2.3,2) and (6,2) .. (4) {};
\end{tikzpicture}
}
\subfigure[Alternating paths in $F\bicol H$\label{FbicolH}]{
\begin{tikzpicture}
\node (1) at (0,0) {};
	\node (1G) at (-0.3,0) {};
	\node (1D) at (0.3,0) {};
\node (2) at (2,0) {};
	\node (2G) at (1.7,0) {};
	\node (2D) at (2.3,0) {};
\node (3) at (4,0) {3};
\node (4) at (6,0) {4};

\draw[->] (4) .. controls (5.5,1) and (4.5,1) .. (3) {};
\draw[->] (3) .. controls (4,2) and (-0.3,2) .. (-0.3,0) {};
\draw[<-] (0.3,0) .. controls (0.3,-1) and (1.7,-1) .. (1.7,0) {};
\draw[->] (0.3,0) .. controls (0.3,1) and (1.7,1) .. (1.7,0) {};
\draw[<-] (2.3,0) .. controls (2.3,-1.5) and (-0.3,-1.5) .. (-0.3,0) {};
\draw[->] (2.3,0) .. controls (2.3,2) and (6,2) .. (4) {};
\end{tikzpicture}
}
\caption{Alternating paths.}\label{examplepaths}
\end{figure}

\begin{definition}[Paths, cycles and $k$-cycles]\label{pathdef}
A \emph{path} in a graph $G$ is a finite sequence of edges $(e_{i})_{0\leqslant i\leqslant n}\text{ }(n\in\mathbf{N})$ in $E^{G}$ such that $s(e_{i+1})=t(e_{i})$ for all $0\leqslant i\leqslant n-1$. We will call the vertices $s(\pi)=s(e_{0})$ and $t(\pi)=t(e_{n})$ the beginning and the end of the path.

We will also call a \emph{cycle} a path $\pi=(e_{i})_{0\leqslant i\leqslant n}$ such that $s(e_{0})=t(e_{n})$. If $\pi$ is a cycle, and $k$ is the greatest integer such that there exists a cycle $\rho$ with\footnote{Here, we denote by $\rho^{k}$ the concatenation of $k$ copies of $\rho$.} $\pi=\rho^{k}$, we will say that $\pi$ is a \emph{$k$-cycle}.
\end{definition}

\begin{definition}[Alternating paths]
Let $G$ and $H$ be two graphs. We define the \emph{alternating paths} between $G$ and $H$ as the paths $(e_{i})$ in $G\bicol H$ which satisfy
\begin{equation*}
\delta(e_{i})\neq\delta(e_{i+1})~~~~(i=0,\dots,n-1)
\end{equation*} 
We will call an \emph{alternating cycle} in $G\bicol H$ a cycle $(e_{i})_{0\leqslant i\leqslant n}$ in $G\bicol H$ which is an alternating path and such that $\delta(e_{n})\not=\delta(e_{0})$.

The set of alternating paths in $G\bicol H$ will be denoted by $\enpaths{G,H}$, while $\enpaths{G,H}_{V}$ will mean the subset of alternating paths in $G\bicol H$ with source and target in a given set of vertices $V$.
\end{definition}

\begin{proposition}\label{countingeqclas}\label{eqclasskcycl}
Let $\rho=(e_{i})_{0\leqslant i\leqslant n-1}$ be a cycle, and let $\sigma$ be the permutation taking $i$ to $i+1$ ($i=0,\dots,n-2$) and $n-1$ to $0$. We define the set
\begin{equation*}
\bar{\rho}=\{(e_{\sigma^{k}(i)})_{0\leqslant i\leqslant n-1}~|~ 0\leqslant k\leqslant n-1\}
\end{equation*}
Then $\rho$ is a $k$-cycle if and only if the cardinality of $\bar{\rho}$ is equal to $n/k$. In the following, we will refer to such an equivalence class modulo cyclic permutations as a \emph{$k$-circuit}.
\end{proposition}

\begin{proof}
We use classical cyclic groups techniques here. We will abusively denote by $\sigma^{p}(\rho)$ the path $(e_{\sigma^{p}(i)})_{0\leqslant i\leqslant n-1}$.

First, notice that if $\rho$ is a $k$-cycle, then $\sigma^{n/k}(\rho)=\rho$. Now, if $s$ is the smallest integer such that $\sigma^{s}(\rho)=\rho$, we have that $e_{i+s}=e_{i}$. Hence, writing $m=n/s$, we have $\rho=\pi^{m}$ where $\pi=(e_{i})_{0\leqslant i\leqslant s-1}$. This implies that $k=n/s$ from the maximality of $k$. Hence $\rho$ is a $k$-cycle if and only if the smallest integer $s$ such that $\sigma^{s}(\rho)=\rho$ is equal to $n/k$.

Let $s$ be the smallest integer such that $\sigma^{s}(\rho)=\rho$. We have that for any integers $p,q$ such that $0\leqslant q< s$, $\sigma^{ps+q}(\rho)=\sigma^{q}(\rho)$. Indeed, it is a direct consequence of the fact that $\sigma^{ps}(\rho)=\rho$ for any integer $p$. Moreover, since $\sigma^{n}(\rho)=\rho$, we have that $s$ divides $n$. Hence, we have that the cardinality of $\bar{\rho}$ is at most $s$. To show that the cardinality of $\bar{\rho}$ is exactly $s$, we only need to show that $\sigma^{i}(\rho)\neq\sigma^{j}(\rho)$ for $i< j$ between $0$ and $s-1$. But if it were the case, we would have, since $\sigma$ is a bijection, $\rho=\sigma^{j-i}(\rho)$, an equality contradicting the minimality of $s$.
\end{proof}

\begin{definition}[The set of $1$-circuits]
We will denote by $\circuits{G,H}$ the set of alternating $1$-circuits in $G\bicol H$, i.e. the quotient of the set of alternating $1$-cycles by cyclic permutations.
\end{definition}

\begin{example}
\autoref{examplepaths} shows the alternating paths between $F$ and respectively $G$ and $H$. Notice that in \autoref{FbicolH} the circle represent an infinity of paths, since there are two paths of length $k$ for every integer $k$ (of sources $1$ and $2$ respectively). Let us denote by $f$ and $h$ the edges of $F$ and $H$ used to define thoses paths. The alternating cycles between $F$ and $H$ are also infinite in number: every such path of even length defines a cycle, i.e. cycles are paths $(ef)^{k}$ or $(fe)^{k}$ where $\pi^{k}$ denotes the concatenation of $k$ copies of $\pi$. Since circuits are equivalence classes of cycles modulo cyclic permutations, there are then one circuit for each even integer, and therefore the number of circuits is still infinite. However, there is only one $1$-circuit alternating between $F$ and $H$, namely the equivalence class $\rho$ of cycles of length $2$. Indeed, every other circuit $\pi$ is obtained as the concatenation $\rho^{k}$ for an integer $k$. 
\end{example}

\begin{definition}\label{execdef}
Let $F$ and $G$ be two graphs. We define the \emph{execution} of $F$ and $G$ as the graph $F\plug G$ defined by:
\begin{eqnarray*}
V^{F\plug G}&=&V^{F}\Delta V^{G}=(V^{F}\cup V^{G})-(V^{F}\cap V^{G})\\
E^{F\plug G}&=&\enpaths{F,G}_{V^{F}\Delta V^{G}}\\
s^{F\plug G}&=&\pi\mapsto s(\pi)\\
t^{F\plug G}&=&\pi\mapsto t(\pi)\\
\omega^{F\plug G}&=& \pi=\{e_{i}\}_{i=0}^{n}\mapsto \prod_{i=0}^{n}\omega^{G\bicol H}(e_{i})
\end{eqnarray*}

When $V^{F}\cap V^{G}=\emptyset$, we will write $F\cup G$ instead of $F\plug G$.
\end{definition}

\begin{proposition}[Associativity]\label{assoc}\label{associativity}
Let $G_{0},G_{1},G_{2}$ be three graphs with $V^{G_{0}}\cap V^{G_{1}}\cap V^{G_{2}}=\emptyset$. We have:
\begin{equation*}
G_{0}\plug (G_{1}\plug G_{2})=(G_{0}\plug G_{1})\plug G_{2}
\end{equation*}
\end{proposition}

\begin{proof}
Let us first  define the 3-colored graph $G_{0}\square G_{1}\square G_{2}$ as the union graph $(\bigcup V^{i},\biguplus E^{i},\biguplus s^{i}, \biguplus t^{i})$ together with the coloring function $\delta$ from $\biguplus E^{i}$ into $\{0,1,2\}$ which associates to each edge the number $i$ of the graph $G_{i}$ it comes from. We consider the 3-alternating paths between $G_{0},G_{1},G_{2}$, that is the paths $(e_{i})$ in $G_{0}\square G_{1}\square G_{2}$ satisfying:
\begin{equation*}
\delta(e_{i})\neq\delta(e_{i+1})
\end{equation*}
Then, we can define the simultaneous reduction of $G_{0}$, $G_{1}$, and $G_{2}$ as the graph $\mathop{\dblcolon}_{i}G_{i}=(V^{0}\Delta V^{1}\Delta V^{2},F,s^{F},t^{F})$, where $F$ is the set of 3-alternating paths between the graphs $G_{0}$, $G_{1}$, and $G_{2}$, $s^{F}(e)$ is the beginning of the path $e$ and $t^{F}(e)$ is its end.

We then show that this induced graph $\mathop{\dblcolon}_{i}G_{i}$ is equal to $(G_{0}\plug G_{1})\plug G_{2}$ and $G_{0}\plug (G_{1}\plug G_{2})$. This is a simple verification. Indeed, to prove for instance that $\mathop{\dblcolon}_{i}G_{i}$ is equal to $(G_{0}\plug G_{1})\plug G_{2}$, we just write the 3-alternating paths in $G_{0},G_{1},G_{2}$ as an alternating sequence of alternating paths in $G_{0}\bicol G_{1}$ (with\footnote{\label{footnote}This is where the hypothesis $V^{0}\cap V^{1}\cap V^{2}=\emptyset$ is important. If this is not satisfied, one gets some 3-alternating paths of the form $\rho x$, where $x$ is an edge in $G_{2}$ and $\rho$ is an alternating path in $G_{0}\bicol G_{1}$, but such that $\rho$ does not correspond to an edge in $G_{0}\plug G_{1}$.} source and target in $V^{0}\Delta V^{1}$, i.e. an edge of $G_{0}\plug G_{1}$) and edges in $G_{2}$.
\end{proof}

\begin{definition}\label{defmeasurement}
For any function\footnote{When one is working with a monoid of weights $\Omega$, the function $m$ should be \emph{tracial}, i.e. satisfy for every $a,b\in\Omega$, the equation $m(ab)=m(ba)$. Even though satisfied by every map $m:\Omega\rightarrow\mathbb{R}_{\geqslant 0}\cup\{\infty\}$ when $\Omega$ is a commutative monoid, the traciality is not ensured in general. This requirement is necessary in the general case from the very definition of the measurement, since the quantity $m(\omega(\pi))$ for a circuit $\pi$ would not be well-defined without it, as the value $m(\rho_{1})$ and $m(\rho_{2})$ would not be equal for two choices of representatives $\rho_{1},\rho_{2}$ of the circuit $\pi$.} $m: ]0,1]\rightarrow \mathbb{R}_{\geqslant 0}\cup\{\infty\}$, let us define a measure on graphs by $\meas{F,G}=\sum_{\pi\in\circuits{F,G}} m(\omega^{F\bicol G}(\pi))$.
\end{definition}

\subsection{The Trefoil Property}

In our first paper \cite{seiller-goim}, we obtained the following three-terms adjunction in the case $m(x)=-\log(1-x)$:
$$\meas{F,G\cup H}=\meas{F,G}+\meas{F\plug G,H}$$
As it turns out, the adjunction is independent from the chosen measure, and comes from a more general geometrical identity that describes how the sets of $1$-circuits interact with the execution. This more general identity, which we call \emph{the trefoil property}, turns out to be of great importance in our construction. On the one hand it gives the adjunction which ensures that the construction of multiplicative connectives is adequate (i.e. it implies that $\cond{A\otimes B}=\cond{(A\multimap B^{\pol})^{\pol}}$), i.e. that one can obtain a $\ast$-autonomous category \catmll{}. On the other hand, it will allow us to show that:
\begin{itemize}
\item the $\with$ connective we define is a product \emph{up to observational equivalence};
\item this observational equivalence is a congruence on the category \catmll{};
\item the quotient category inherits the $\ast$-autonomous structure of \catmll{}.
\end{itemize}
As a consequence of this sole geometrical property, we thus obtain a family of GoI constructions for MALL which all define a categorical model.

\subsubsection{The trefoil property explained}

We will use Figure \ref{figtrefoil} to explain the trefoil property. In this figure, we consider three graphs $F$, $G$, and $H$ such that a given vertex (in any of the graphs) cannot be a vertex in the three graphs simultaneously, i.e. the intersection $V^{F}\cap V^{G}\cap V^{H}$ is empty. The double arrows in Figure \ref{figtrefoila} represent the sets of edges (in any of the graphs) that one can go through to go from one graph to the other: for instance the double arrow between $V^{F}$ and $V^{G}$ stands for the set of edges of $F$ whose target is an element of $V^{G}$. We also represent the sets of cycles formed from edges of $F$ and $G$ only (respectively edges of $F$ and $H$ only, respectively edges of $G$ and $H$ only) by a dotted cycle (respectively plain, respectively dashed) in Figures \ref{figtrefoilb}, \ref{figtrefoilc}, \ref{figtrefoild}, and \ref{figtrefoile}. Finally, the set of cycles that contain at least one edge from each graph is represented by a double-line cycle in these figures. The rectangles in Figures \ref{figtrefoilc}, \ref{figtrefoild}, and \ref{figtrefoile} represent the computation of the execution between two graphs, respectively $F$ and $G$ (Figure \ref{figtrefoilc}), $G$ and $H$ (Figure \ref{figtrefoild}) and $F$ and $H$ (Figure \ref{figtrefoile}). We then notice that during the execution of two graphs, one hides and forget about the alternating cycles contained in the rectangle. For instance when the execution between $F$ ans $G$ has been computed, we can no longer \enquote{see} the cycles composed of edges of $F$ and $G$ only. As a consequence, in order to account for all cycles, one should consider both the set of cycles between $F\plug G$ and $H$ and the set of cycles between $F$ and $G$. The trefoil property then states that, if $\cyclesfig{A,B}$ denotes the set of alternating cycles between the two graphs $A$ and $B$:
\begin{equation*}
\cyclesfig{F\plug G,H}\cup\cyclesfig{F,G}\cong\cyclesfig{G\plug H,F}\cup\cyclesfig{G,H}\cong\cyclesfig{H\plug F,G}\cup\cyclesfig{H,F}
\end{equation*}

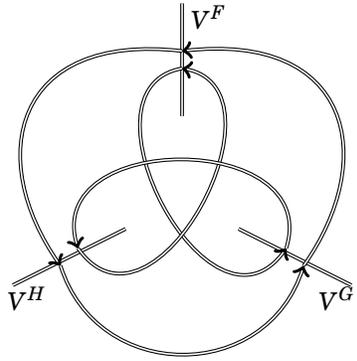
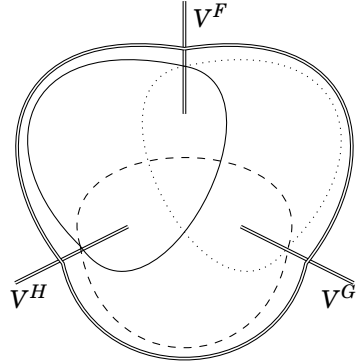
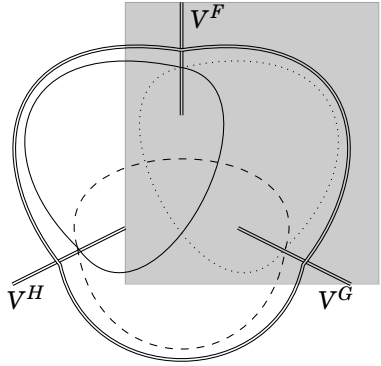
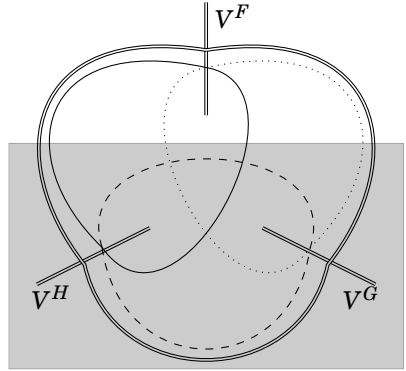
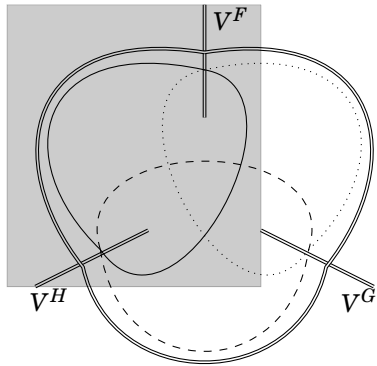
\begin{figure}
\centering
\subfigure[Representation of $F,G,H$\label{figtrefoila}]{
\begin{tikzpicture}[x=0.75cm,y=0.75cm]
     \draw[-,double] (0,0) -- (0,-2) node [very near start, right] {$V^{F}$};
     \draw[-,double] (-1,-4) -- (-3,-5) node [very near end,below] {$V^{H}$};
     \draw[-,double] (1,-4) -- (3,-5) node [very near end,below] {$V^{G}$};
     
     \node (F) at (0,-1) {};
     \node (H) at (-2,-4.5) {};          
     \node (G) at (2,-4.5) {};
     
     \node (CFH1) at (-3,-0.5) {};
     \node (CFH2) at (-3.5,-2.75) {};
     \node (CHF1) at (-0.5,-6.25) {};
     \node (CHF2) at (2,-1.5) {};
     \node (CGH1) at (2.5,-2.25) {};
     \node (CGH2) at (-2.5,-2.25) {};
     \node (CHG1) at (-1.5,-6.75) {};
     \node (CHG2) at (1.5,-6.75) {};
     \node (CGF1) at (3.5,-2.75) {};
     \node (CGF2) at (3,-0.5) {};
     \node (CFG1) at (-2,-1.5) {};
     \node (CFG2) at (0.5,-6.25) {};
    
     \draw[->,double] (F.north) .. controls (CFH1) and (CFH2) .. (H.south west) {};
     \draw[->,double] (G.south east) .. controls (CGF1) and (CGF2) .. (F.north) {};
     \draw[->,double] (H.south west) .. controls (CHG1) and (CHG2) .. (G.south east) {};	

     \draw[->,double] (F.south) .. controls (CFG1) and (CFG2) .. (G.north west) {};	
     \draw[->,double] (H.north east) .. controls (CHF1) and (CHF2) .. (F.south) {};	
     \draw[->,double] (G.north west) .. controls (CGH1) and (CGH2) .. (H.north east) {};
\end{tikzpicture}
}
\subfigure[Alternating cycles between $F$, $G$ and $H$\label{figtrefoilb}]{
\begin{tikzpicture}[x=0.75cm,y=0.75cm]
     \definecolor{darkgreen}{rgb}{0.15,0.6,0.15};
     
     \draw[-,double] (0,0) -- (0,-2) node [very near start, right] {$V^{F}$};
     \draw[-,double] (-1,-4) -- (-3,-5) node [very near end,below] {$V^{H}$};
     \draw[-,double] (1,-4) -- (3,-5) node [very near end,below] {$V^{G}$};

     \node (F) at (0,-1) {};
     \node (H) at (-2,-4.5) {};          
     \node (G) at (2,-4.5) {};
     
     \node (CFH1) at (-3,-0.5) {};
     \node (CFH2) at (-3.5,-2.75) {};
     \node (CHF1) at (-0.5,-6.25) {};
     \node (CHF2) at (2,-1.5) {};
     \node (CGH1) at (2.5,-2.25) {};
     \node (CGH2) at (-2.5,-2.25) {};
     \node (CHG1) at (-1.5,-6.75) {};
     \node (CHG2) at (1.5,-6.75) {};
     \node (CGF1) at (3.5,-2.75) {};
     \node (CGF2) at (3,-0.5) {};
     
     \draw[-] (F.south) .. controls (CFH1) and (CFH2) .. (H.north east) .. controls (CHF1) and (CHF2) .. (F.south) {};
     \draw[dashed,-] (H.north east) .. controls (CHG1) and (CHG2) .. (G.north west) .. controls (CGH1) and (CGH2) .. (H.north east) {};
     \draw[double,-] (F.north) .. controls (CFH1.north west) and (CFH2.west) .. (H.south west) .. controls (CHG1.south west) and (CHG2.south east) .. (G.south east) .. controls (CGF1.east) and (CGF2.north east) .. (F.north) {};
     \draw[dotted,-] (F.south) .. controls (CFG1) and (CFG2) .. (G.north west) .. controls (CGF1) and (CGF2) .. (F.south) {};
\end{tikzpicture}}
\subfigure[Alternating cycles between $F\plug{}G$ and $H$\label{figtrefoilc}]{
\begin{tikzpicture}[x=0.75cm,y=0.75cm]
     \definecolor{darkgreen}{rgb}{0.15,0.6,0.15};
     
     \draw[-,double] (0,0) -- (0,-2) node [very near start, right] {$V^{F}$};
     \draw[-,double] (-1,-4) -- (-3,-5) node [very near end,below] {$V^{H}$};
     \draw[-,double] (1,-4) -- (3,-5) node [very near end,below] {$V^{G}$};

     \node (F) at (0,-1) {};
     \node (H) at (-2,-4.5) {};          
     \node (G) at (2,-4.5) {};
     
     \node (CFH1) at (-3,-0.5) {};
     \node (CFH2) at (-3.5,-2.75) {};
     \node (CHF1) at (-0.5,-6.25) {};
     \node (CHF2) at (2,-1.5) {};
     \node (CGH1) at (2.5,-2.25) {};
     \node (CGH2) at (-2.5,-2.25) {};
     \node (CHG1) at (-1.5,-6.75) {};
     \node (CHG2) at (1.5,-6.75) {};
     \node (CGF1) at (3.5,-2.75) {};
     \node (CGF2) at (3,-0.5) {};
     
     \draw[-] (F.south) .. controls (CFH1) and (CFH2) .. (H.north east) .. controls (CHF1) and (CHF2) .. (F.south) {};
     \draw[dashed,-] (H.north east) .. controls (CHG1) and (CHG2) .. (G.north west) .. controls (CGH1) and (CGH2) .. (H.north east) {};
     \draw[double,-] (F.north) .. controls (CFH1.north west) and (CFH2.west) .. (H.south west) .. controls (CHG1.south west) and (CHG2.south east) .. (G.south east) .. controls (CGF1.east) and (CGF2.north east) .. (F.north) {};
     \draw[dotted,-] (F.south) .. controls (CFG1) and (CFG2) .. (G.north west) .. controls (CGF1) and (CGF2) .. (F.south) {};

     \draw[-,fill=black,opacity=0.2] (-1,0) -- 
     				(-1,-5) --
					(3.5,-5) --
						(3.5,0) --
							(-1,0) {};
\end{tikzpicture}}
\subfigure[Alternating cycles between $G\plug{}H$ and $F$\label{figtrefoild}]{
\begin{tikzpicture}[x=0.75cm,y=0.75cm]
     \definecolor{darkgreen}{rgb}{0.15,0.6,0.15};

     \draw[-,double] (0,0) -- (0,-2) node [very near start, right] {$V^{F}$};
     \draw[-,double] (-1,-4) -- (-3,-5) node [very near end,below] {$V^{H}$};
     \draw[-,double] (1,-4) -- (3,-5) node [very near end,below] {$V^{G}$};

     \node (F) at (0,-1) {};
     \node (H) at (-2,-4.5) {};          
     \node (G) at (2,-4.5) {};
     
     \node (CFH1) at (-3,-0.5) {};
     \node (CFH2) at (-3.5,-2.75) {};
     \node (CHF1) at (-0.5,-6.25) {};
     \node (CHF2) at (2,-1.5) {};
     \node (CGH1) at (2.5,-2.25) {};
     \node (CGH2) at (-2.5,-2.25) {};
     \node (CHG1) at (-1.5,-6.75) {};
     \node (CHG2) at (1.5,-6.75) {};
     \node (CGF1) at (3.5,-2.75) {};
     \node (CGF2) at (3,-0.5) {};
     
     \draw[-] (F.south) .. controls (CFH1) and (CFH2) .. (H.north east) .. controls (CHF1) and (CHF2) .. (F.south) {};
     \draw[-,dashed] (H.north east) .. controls (CHG1) and (CHG2) .. (G.north west) .. controls (CGH1) and (CGH2) .. (H.north east) {};
     \draw[double,-] (F.north) .. controls (CFH1.north west) and (CFH2.west) .. (H.south west) .. controls (CHG1.south west) and (CHG2.south east) .. (G.south east) .. controls (CGF1.east) and (CGF2.north east) .. (F.north) {};
     \draw[dotted,-] (F.south) .. controls (CFG1) and (CFG2) .. (G.north west) .. controls (CGF1) and (CGF2) .. (F.south) {};

     \draw[-,fill=black,opacity=0.2] (-3.5,-6.5) -- 
     				(3.5,-6.5) --
					(3.5,-2.5) --
						(-3.5,-2.5) --
							(-3.5,-6.5) {};
\end{tikzpicture}}
\subfigure[Alternating cycles between $F\plug{}G$ and $H$\label{figtrefoile}]{
\begin{tikzpicture}[x=0.75cm,y=0.75cm]
     \definecolor{darkgreen}{rgb}{0.15,0.6,0.15};

     \draw[-,double] (0,0) -- (0,-2) node [very near start, right] {$V^{F}$};
     \draw[-,double] (-1,-4) -- (-3,-5) node [very near end,below] {$V^{H}$};
     \draw[-,double] (1,-4) -- (3,-5) node [very near end,below] {$V^{G}$};

     \node (F) at (0,-1) {};
     \node (H) at (-2,-4.5) {};          
     \node (G) at (2,-4.5) {};
     
     \node (CFH1) at (-3,-0.5) {};
     \node (CFH2) at (-3.5,-2.75) {};
     \node (CHF1) at (-0.5,-6.25) {};
     \node (CHF2) at (2,-1.5) {};
     \node (CGH1) at (2.5,-2.25) {};
     \node (CGH2) at (-2.5,-2.25) {};
     \node (CHG1) at (-1.5,-6.75) {};
     \node (CHG2) at (1.5,-6.75) {};
     \node (CGF1) at (3.5,-2.75) {};
     \node (CGF2) at (3,-0.5) {};
     
     \draw[-] (F.south) .. controls (CFH1) and (CFH2) .. (H.north east) .. controls (CHF1) and (CHF2) .. (F.south) {};
     \draw[dashed,-] (H.north east) .. controls (CHG1) and (CHG2) .. (G.north west) .. controls (CGH1) and (CGH2) .. (H.north east) {};
     \draw[double,-] (F.north) .. controls (CFH1.north west) and (CFH2.west) .. (H.south west) .. controls (CHG1.south west) and (CHG2.south east) .. (G.south east) .. controls (CGF1.east) and (CGF2.north east) .. (F.north) {};
     \draw[dotted,-] (F.south) .. controls (CFG1) and (CFG2) .. (G.north west) .. controls (CGF1) and (CGF2) .. (F.south) {};
     \draw[-,fill=black,opacity=0.2] (1,0) -- 
     				(1,-5) --
					(-3.5,-5) --
						(-3.5,0) --
							(1,0) {};
\end{tikzpicture}}
\caption{Graphical representation of the trefoil property}\label{figtrefoil}
\end{figure}

\subsubsection{The trefoil property: formally}

\begin{proposition}[Geometric trefoil property]\label{cyclpptycycles}
Let $F,G,H$ be three graphs such that $V^{F}\cap V^{G}\cap V^{H}=\emptyset$. Then:
\begin{eqnarray*}
\circuits{F,G}\cup\circuits{F\plug G,H}&\cong&\circuits{G,H}\cup\circuits{G\plug H,F}\\
&\cong&\circuits{H,F}\cup\circuits{H\plug F,G}
\end{eqnarray*}
Moreover, these bijections preserve weights.
\end{proposition}

\begin{proof}
We prove this following the idea of the proof of associativity.

Let $A=F\bicol G\bicol H$ the three-colored graph obtained from $F,G,H$ together with its $3$-coloring function $\rho$, and consider a $3$-alternating $1$-cycle in $F\bicol G\bicol H$, i.e. a $1$-cycle $(e_{i})_{i=0}^{n}$ which is $3$-alternating and that satisfies $\rho(e_{0})\neq \rho(e_{n})$. Denote by $\cythree(F,G,H)$ the set of $3$-alternating $1$-cycles in $F\bicol G\bicol H$.

Then $\pi\in\cythree(F,G,H)$ is either  an alternating $1$-cycle in $F\bicol G$ (if it does not contain any edge from $H$) or it is a $1$-cycle composed of alternations between alternating paths (with source and target in $V^{F}\Delta V^{G}$, see Footnote \ref{footnote}) in $F\bicol G$ and edges in $H$. Hence, $\pi$ is either in $\circuits{F,G}$ or in $\circuits{F\plug G,H}$. Conversely, any alternating $1$-cycle in $F\bicol G$ corresponds to a unique $3$-alternating cycle in $F\bicol G\bicol H$, and any alternating $1$-cycle in $(F\plug G)\bicol H$ corresponds to a unique $3$-alternating $1$-cycle in $F\bicol G\bicol H$. Hence we have a bijection $$\cythree(F,G,H)\cong\circuits{F,G}\cup\circuits{F\plug G,H}$$

A similar argument shows that $\cythree(F,G,H)\cong\circuits{G,H}\cup\circuits{G\plug H,F}$ and that $\cythree(F,G,H)\cong\circuits{H,F}\cup\circuits{H\plug F,G}$.
\end{proof}

\begin{corollary}[Scalar trefoil property]\label{scalarcycl}
Let $F,G,H$ be such as in the preceding proposition. Then:
\begin{eqnarray*}
\meas{F,G}+\meas{F\plug G,H}&=&\meas{G,H}+\meas{G\plug H,F}\\
&=&\meas{H,F}+\meas{H\plug F,G}
\end{eqnarray*}
\end{corollary}

\begin{corollary}[Geometric three-terms adjunction]
Let $F$, $G$, and $H$ be weighted graphs such that $V^{G}\cap V^{H}=\emptyset$. We have 
\begin{equation*}
\circuits{F,G\cup H}\cong \circuits{F,G}\cup\circuits{F\plug G, H}
\end{equation*}
\end{corollary}

\begin{corollary}[Scalar three-terms adjunction]
With the hypotheses of the last corollary:
\begin{equation*}
\meas{F,G\cup H}=\meas{F,G}+\meas{F\plug G,H}
\end{equation*}
\end{corollary}

\section{The Additive Construction}\label{additivesection}

\subsection{Sliced Graphs}

The additive construction consists in considering finite weighted families of graphs $G_{i}$ on the same set of vertices. This can be related to the way additives are dealt with in Girard's proof nets \cite{proofnets}: a sum $G+G'$ intuitively corresponds to the superposition of two proof structures $G$ and $G'$ using a boolean weight $p$. The main interest of this approach is that it allows us to juxtapose two graphs while making sure they cannot interact: such an operation will be the interpretation of the $\with$ connective in the GoI construction defined later on. In this section, we will define \emph{sliced graphs} and extend in a natural way the operations of execution and measurement; we then show that the scalar trefoil property can be extended to this setting.

\begin{definition}
A \emph{sliced graph} $F$ with carrier $V^{F}$ is a family $F=\sum_{i\in I^{F}}\alpha^{F}_{i}F_{i}$ where $I^{F}$ is a finite set, and, for all $i\in I^{F}$, $\alpha^{F}_{i}$ is a real number and $F_{i}$ is a graph on the set of vertices $V^{F}$.

Given a sliced graph $F$, we define the real number $\unit{F}=\sum_{i\in I^{F}}\alpha^{F}_{i}$. 
\end{definition}

Let us emphasise that the summation sign is merely a notation for indexed families. We will suppose here that these families are identified modulo isomorphisms of the indexing set; although not essential -- the corresponding \emph{sliced graphs} would be \enquote{identified}\footnote{We use quotes here as the equivalence is defined on projects, not sliced graphs. However two projects $(a,A)$ and $(a,B)$, where $B$ is obtained from $A$ through an isomorphism of the indexing set, will always be equivalent. This is how the graphs $A,B$ are \enquote{identified} by the equivalence on projects.} modulo observational equivalence  (\autoref{equivalencedef}) -- this supposition makes the manipulation of indexed families easier. No other properties are supposed: for instance $A+A$ and $2A$ are different sliced graphs. Let us notice however that these sliced graphs are \enquote{identified} as well when considering the quotient modulo observational equivalence.

We can now easily extend execution and measurement on sliced graphs.
\begin{definition}
Let $F$ and $G$ be two sliced graphs. We define their \emph{execution} as:
\begin{equation*}
\left(\sum_{i\in I^{F}} \alpha^{F}_{i}F_{i}\right)\plug\left(\sum_{i\in I^{G}} \alpha^{G}_{i}G_{i}\right)=\!\!\!\!\sum_{(i,j)\in I^{F}\times I^{G}}\!\!\!\!\alpha^{F}_{i}\alpha^{G}_{j} F_{i}\plug G_{j}
\end{equation*}

When $V^{F}\cap V^{G}=\emptyset$, we will denote the execution by $F\cup G$.
\end{definition}

\begin{definition}
Let $F$ and $G$ be two sliced graphs. We define the \emph{measurement}:
\begin{equation*}
\meas{\sum_{i\in I^{F}} \alpha^{F}_{i}F_{i},\sum_{i\in I^{G}} \alpha^{G}_{i}G_{i}}=\!\!\!\!\sum_{(i,j)\in I^{F}\times I^{G}}\!\!\!\!\alpha^{F}_{i}\alpha^{G}_{j} \meas{F_{i},G_{j}}
\end{equation*}
In the case where some of the $\meas{F_{i},G_{j}}$ are equal to $\infty$, we set $\meas{F,G}$ to be equal to $\infty$.
\end{definition}

A simple computation gives us the trefoil property for sliced graphs as a direct consequence of the trefoil property for graphs.
\begin{proposition}\label{cyclicpptydia}
Let $F,G,H$ be sliced graphs such that $V^{F}\cap V^{G}\cap V^{H}=\emptyset$. Then
\begin{eqnarray*}
\lefteqn{\unit{H}\meas{F,G}+\meas{F\plug G,H}}\\&=&\unit{F}\meas{G,H}+\meas{G\plug H,F}\\
&=&\unit{G}\meas{H,F}+\meas{H\plug F,G}
\end{eqnarray*}
\end{proposition}

\begin{proof}
This is a straightforward computation, consequence of Corollary \ref{scalarcycl}:
\begin{eqnarray*}
\lefteqn{\meas{F,G\plug H}+\unit{F}\meas{G,H}}\\
&=&\meas{F,G\plug H}+ \left(\sum_{i\in I^{F}}d^{F}(i)\right)\sum_{j\in I^{G}}\sum_{k\in I^{H}} \meas{G_{j},H_{k}}\\
&=&\sum_{i\in I^{F}}\sum_{j\in I^{G}}\sum_{k\in I^{H}} d^{F}(i)d^{G}(j)d^{H}(k) \meas{F_{i}, G_{j}\plug H_{k}}\\
&&~~~~+\sum_{i\in I^{F}}\sum_{j\in I^{G}}\sum_{k\in I^{H}} d^{F}(i)d^{G}(j)d^{H}(k) \meas{G_{j}, H_{k}}\\
&=&\sum_{i\in I^{F}}\sum_{j\in I^{G}}\sum_{k\in I^{H}} d^{F}(i)d^{G}(j)d^{H}(k) (\meas{F_{i}, G_{j}\plug H_{k}}+\meas{G_{j}, H_{k}})\\
&=&\sum_{i\in I^{F}}\sum_{j\in I^{G}}\sum_{k\in I^{H}} d^{F}(i)d^{G}(j)d^{H}(k) (\meas{H_{k},F_{i}\plug G_{j}}+\meas{F_{i}, G_{j}})\\
&=&\sum_{i\in I^{F}}\sum_{j\in I^{G}}\sum_{k\in I^{H}} d^{F}(i)d^{G}(j)d^{H}(k) \meas{F_{i}\exec G_{j},H_{k}}\\
&&~~~~+\sum_{i\in I^{F}}\sum_{j\in I^{G}}\sum_{k\in I^{H}} d^{F}(i)d^{G}(j)d^{H}(k) \meas{F_{i}, G_{j}}\\
&=&\meas{H,F\plug G}+ \left(\sum_{k\in I^{H}}d^{H}(k)\right)\sum_{i\in I^{F}}\sum_{j\in I^{G}} \meas{F_{i},G_{j}}\\
&=&\meas{H, F\plug G}+\unit{H}\meas{F,G}
\end{eqnarray*}
The other equality is obtained in a similar way.
\end{proof}

\begin{corollary}\label{adjunctdia}
Let $F$, $G$, and $H$ be sliced graphs such that $V^{G}\cap V^{H}=\emptyset$. Then:
\begin{equation*}
\meas{F,G\cup H}=\unit{H}\meas{F,G}+\meas{F\plug G,H}
\end{equation*}
\end{corollary}

\subsection{Wagers}\label{wagerintro}\label{wagerdef}

In order to get an adjunction in the usual sense, i.e. relating only the terms $\meas{F,G\cup H}$ and $\meas{F\plug G, H}$, we need to get rid of the additional term $\unit{H}\meas{F,G}$. A good way of doing so is to capture this term in a real number that will be associated to our graphs, the \emph{wager}. Hence, the objects we will be working with will be couples of a wager and a sliced graph.

\begin{definition}\label{projectdef}
A \emph{project} is a couple $\de{a}=(a,A)$ where $A$ is a sliced graph and $a\in\mathbb{R}\cup\{\infty\}$. We will call $V^{A}$ the \emph{carrier} of $\de{a}$ and denote it by $\supp{\de{a}}$.
\end{definition}

Since we will work in $\mathbb{R}\cup\{\infty\}$, we need to explain how sums and products are defined on $\infty$. We will follow a very simple rule: any sum and any product containing $\infty$ will be equal to $\infty$.

\begin{definition}
Let $\de{a},\de{b}$ be projects. We define the \emph{measurement}:
\begin{equation*}
\sca{a}{b}=a\unit{B}+\unit{A}b+\meas{A,B}
\end{equation*}
\end{definition}

\begin{definition}\label{executiondef}
Let $\de{f}$ and $\de{g}$ be projects. We define the \emph{cut} between $\de{f}$ and $\de{g}$ by:
\begin{equation*}
\de{f\deplug g}=(\sca{f}{g},F\plug G)
\end{equation*}
When $\de{a},\de{b}$ are with disjoint carriers, i.e. $V^{F}\cap V^{G}=\emptyset$, the execution $\de{a\deplug b}$ defines the \emph{tensor product} of $\de{a}$ and $\de{b}$ as $(a\unit{B}+b\unit{A},A\cup B)$.
\end{definition}

The following theorem is a straightforward application of Proposition \ref{cyclicpptydia}.

\begin{theorem}[Trefoil Property]\label{cyclicppty}
Let $\de{f,g,h}$ be projects such that $V^{F}\cap V^{G}\cap V^{H}=\emptyset$. Then:
\begin{equation*}
\sca{f\deplug g}{h}=\sca{f\deplug h}{g}=\sca{g\deplug h}{f}
\end{equation*}
\end{theorem}

\begin{proof}
We only show one equality, namely $\sca{f\deplug g}{h}=\sca{f\deplug h}{g}$. Using the definitions and Proposition \ref{cyclicpptydia}:
\begin{eqnarray*}
\sca{f\deplug g}{h}&=&(f\unit{G}+g\unit{F})\unit{H}+h\unit{F}\unit{G}+\meas{F\plug G,H}\\
&=&(f\unit{H}+h\unit{F})\unit{G}+g\unit{F}\unit{H}+\meas{F\plug H,G}\\
&=&\sca{f\deplug h}{g}
\end{eqnarray*}
The second equality is obtained by a similar computation.
\end{proof}

\begin{corollary}[Adjunction]\label{adjunction}
Let $\de{f,a,b}$ be projects such that $V^{A}\cap V^{B}=\emptyset$. Then:
\begin{equation*}
\sca{f}{a\otimes b}=\sca{f\deplug a}{b}
\end{equation*}
\end{corollary}

\section{Localised Connectives}\label{gdisection}

We can now proceed with the construction of connectives. We first define a notion of orthogonality from which we derive the notion of \emph{conduct}: a set which is closed under bi-orthogonality. Interpreting the orthogonality $\de{a}\poll\de{b}$ as the fact that $\de{a}$ passes the test $\de{b}$, the definition of conducts can be understood as a way of typing the projects: $\de{a}$ is an element of a conduct $\cond{A}$ if and only if it passes all the tests in $\cond{A}^{\pol}$. The constructions of the connectives is then defined on projects and extends to a definition on conducts by applying the constructions on the elements and closing under bi-orthogonality. Connectives are thus defined by their computational meaning.

\subsection{Multiplicatives}

\begin{definition}[Orthogonality]\label{orthogonalitydef}
Two projects $\de{a,b}$ on the same carrier are \emph{orthogonal}, denoted by $\de{a}\poll[m]\de{b}$, if $\sca{a}{b}\neq 0,\infty$.

If $A$ is a set of projects, the orthogonal $A^{\pol[m]}$ of $A$ is defined as the set: 
$$\{\de{b}~|~\forall \de{a}\in A, \de{a}\poll\de{b}\}$$

We will in general forget about the subscript $m$ in order to simplify notations.
\end{definition}

\begin{definition}[Conducts]\label{conductdef}
A \emph{conduct} is a set $\cond{A}$ of projects which is equal to its bi-orthogonal, i.e. $\cond{A}=\cond{A}^{\pol\pol}$, together with a set $V^{A}$ such that  $\de{a}\in \cond{A}\Rightarrow \supp{\de{a}}=V^{A}$. The set $V^{A}$ will be called the \emph{carrier} of the conduct, and denoted by $\supp{\cond{A}}$.
\end{definition}

\begin{proposition}\label{miseinfcondpleine}
Let $\cond{A}$ be a conduct, and $\de{a}=(\infty,A)\in\cond{A}$ be a project with an infinite wager. Then $\cond{A}^{\pol}=\emptyset$. We will denote by $\cond{0}_{V^{A}}=\cond{A}^{\pol}$ the empty conduct with carrier $V^{A}$ and by $\cond{T}_{V^{A}}=\cond{A}$ the conduct that contains all projects with carrier $V^{A}$.
\end{proposition}

\begin{proof}
From the definition of operations on $\infty$, any algebraic expression containing $\infty$ is equal to $\infty$. It follows that no project is orthogonal to the project $\de{a}=(\infty,A)$. Hence $\cond{A}^{\pol}=\emptyset$.
\end{proof}

\begin{remark}
If there exists only one conduct (up to the choice of a carrier) containing projects with infinite wager, then why introduce them ? In fact, the introduction of infinite wagers ensures that the application $\de{f}\deplug \de{a}$ is always defined. Technically, this allows us to have the equality between $\cond{0}_{V}\multimap \cond{A}$ and $\cond{A}^{\pol}\multimap \cond{T}_{V}$, which wouldn't be the case if application were not always defined. Indeed, by definition of the linear implication (Definition \ref{linear} below), the conduct $\cond{0}_{V}\multimap \cond{A}$ would be equal to $\cond{T}_{V\cup V^{A}}$, while the conduct $\cond{A}^{\pol}\multimap \cond{T}_{V}$ would contain only the projects $\de{f}$ with $\supp{\de{f}}=V\cup V^{A}$ such that $\de{f}\deplug\de{a}$ is defined for all $\de{a}\in\cond{A}^{\pol}$.
\end{remark}

\begin{definition}[Tensor on Conducts]
Let $\cond{A,B}$ be conducts of disjoint carrier. We can form the conduct $\cond{A}\otimes\cond{B}$
\begin{equation*}
\cond{A\otimes B}=\{\de{a}\otimes\de{b}~|~\de{a}\in\cond{A},\de{b}\in\cond{B}\}^{\pol\pol}
\end{equation*}
\end{definition}

\begin{proposition}\label{ethtenscondtens}
We denote by $\cond{A\odot B}$ the set $\{\de{a}\otimes\de{b}~|~\de{a}\in\cond{A},\de{b}\in\cond{B}\}$. Let $E,F$ be non-empty sets of projects of respective carriers $V,W$ with $V\cap W=\emptyset$. Then 
\begin{equation*}
(E\odot F)^{\pol\pol}=(E^{\pol\pol}\odot F^{\pol\pol})^{\pol\pol}
\end{equation*}
\end{proposition}

\begin{proof}
Obviously, we have $E\subset E^{\pol\pol}$ and $F\subset F^{\pol\pol}$, hence $E\odot F\subset E^{\pol\pol}\odot F^{\pol\pol}$ and finally we get a first inclusion $(E\odot F)\subset (E^{\pol\pol}\odot F^{\pol\pol})^{\pol\pol}$.

For the other inclusion, we prove that $(E\odot F)^{\pol}\subset (E^{\pol\pol}\odot F^{\pol\pol})^{\pol}$. Let $\de{a}$ be a project in $(E\odot F)^{\pol}$. Then for all $\de{e}\in E$ and $\de{f}\in F$ we have $\sca{a}{e\otimes f}\neq 0,\infty$. By the adjunction this means that $\sca{a\deplug e}{f}\neq 0,\infty$, i.e. $\de{a}\deplug\de{e}\in F^{\pol}$. Thus $\sca{a\deplug e}{f'}\neq 0,\infty$ for all $\de{e}\in E$ and $\de{f'}\in F^{\pol\pol}$. Since $\sca{a\deplug e}{f'}=\sca{a\deplug f'}{e}$, we deduce that $\de{a\deplug f'}\in E^{\pol}$, which means that $\de{a\deplug f'}\poll \de{e'}$ for all $\de{f'}\in F^{\pol\pol}$ and $\de{e'}\in E^{\pol\pol}$. To conclude, just notice that this is equivalent to the fact that for all $\de{e'}\in E^{\pol\pol}$ and for all $\de{f'}\in F^{\pol\pol}$, $\sca{a}{e'\otimes f'}\neq 0,\infty$. This implies that $\de{a}\in(E^{\pol\pol}\odot F^{\pol\pol})^{\pol}$ which gives us the second inclusion.
\end{proof}

\begin{definition}
We will write $\de{0}_{V}$ the project $(0,(V,\emptyset))$ where $(V,\emptyset)$ is the empty graph on the set of vertices $V$ considered as a one-sliced graph. Then we will denote by $\cond{1}_{V}$ the conduct $\{\de{0}_{V}\}^{\pol\pol}$.
\end{definition}

\begin{proposition}[Properties of the Tensor]
The tensor product is commutative and associative. Moreover it has a neutral element, namely $\cond{1}_{\emptyset}$.
\end{proposition}

\begin{proof}
These properties are consequences (using Proposition \ref{ethtenscondtens}) of the fact that $\otimes$ on projects is commutative, associative and has a neutral element: the project $\de{0}_{\emptyset}$. Indeed, the tensor product is a special case of execution which is associative and commutative (we did not prove commutativity but it is a direct consequence of locativity). Moreover, one can easily check that $\de{a}\otimes\de{0}_{\emptyset}=\de{a}$.
\end{proof}

\begin{definition}[Linear Implication]\label{linear}
Let $\cond{A,B}$ be conducts with disjoint carriers $V^{A}$ and $V^{B}$.
\begin{equation*}
\cond{A\multimap B}=\{\de{f}~|~\supp{f}=V^{A}\cup V^{B} \wedge \forall \de{a}\in \cond{A}, \de{f}\deplug\de{a}\in\cond{B}\}
\end{equation*}
\end{definition}

The fact that this defines a conduct is justified by the following proposition, which is a simple corollary of the adjunction (Corollary \ref{adjunction}).

\begin{proposition}[Duality]\label{duality}
For any conducts $\cond{A,B}$ with disjoint carriers:
\begin{equation*}
\cond{A\multimap B}=(\cond{A}\otimes \cond{B}^{\pol})^{\pol}
\end{equation*}
\end{proposition}

\begin{proof}
Choose $\de{a}$ in $\cond{A}$ and $\de{b'}\in\cond{B}^{\pol}$. Then, for any project $\de{f}$ with carrier $V^{A}\cup V^{B}$ the adjunction \ref{adjunction} is written:
\begin{equation*}
\sca{f}{a\otimes b'}=\sca{f\deplug a}{b'}
\end{equation*}
This shows that $\de{f}\in\cond{(A\otimes B^{\pol})^{\pol}}$ if and only if $\de{f}\in\cond{A\multimap B}$.
\end{proof}

We now introduce \emph{delocations}, which are renaming of projects' (and conducts') carriers through a chosen bijection. This leads to the introduction of specific projects, named \emph{faxes} following Ludic's terminology \cite{locussolum}, which intuitively correspond to axiom links or copy-cat strategies: these are the projects that interpret axiom rules.

\begin{definition}[Delocations]\label{deloc}
Let $\de{a}$ be a project with carrier $V^{A}$, and $\phi: V^{A} \rightarrow V^{B}$ a bijection. We define the \emph{delocation} of a graph $G$ as the graph $\phi(G)=(V^{B},E^{A},\phi\circ s^{A},\phi\circ t^{A},\omega^{A})$. This extends to projects: the delocation of $\de{a}=(a,\sum_{i\in I^{A}}\alpha^{A}_{i}A_{i})$ is defined as $\phi(\de{a})=(a,\sum_{i\in I^{A}}\alpha^{A}_{i}\phi(A_{i}))$.

Similarly, the delocation of a conduct $\cond{A}$ with carrier $V^{A}$ is defined as the conduct $\phi(\cond{A})=\{\phi(\de{a})~|~\de{a}\in\cond{A}\}$ with carrier $V^{B}$.
\end{definition}

\begin{proposition}\label{faxs}\label{faxdef}
Keeping the notations of Definition \ref{deloc} and supposing $V^{A}\cap V^{B}=\emptyset$, we define the project $\de{Fax}_{\phi}=(0,\{\Phi\})$ whose slice has weight $1$ and where:
\begin{eqnarray*}
E^{\Phi}&=&\{(a,\phi(a))~|~ a\in V^{A}\}\cup\{(\phi(a),a)~|~ a \in V^{A}\}\\
\Phi&=&(V^{A}\cup V^{B},E^{\Phi},\omega^{\Phi}(e)=1)
\end{eqnarray*}
Then $\de{Fax}_{\phi}\in\cond{A}\multimap\phi(\cond{A})$.
\end{proposition}

\subsection{Additives}

To define additive connectives, we will use the formal sum introduced on graphs. In a way, this formal sum corresponds to the notion of \emph{slice} traditionally introduced to define additive proof nets without boxes.

\subsubsection{Definitions}

\begin{definition}
We extend the sum and product by a scalar to projects as follows:
\begin{equation*}
\de{a}+\lambda\de{b}=(a+\lambda b,A+\lambda B)
\end{equation*}
where $A+\lambda B$ is a notation for $\sum_{i\in I^{A}}\alpha^{A}_{i}A_{i}+\sum_{i\in I^{B}}\lambda\alpha^{B}_{i}B_{i}$.
\end{definition}

\begin{proposition}
Let $\de{a},\de{b}$ be projects, and $\lambda\in\mathbb{R}^{\ast}$. Then, for any project $\de{c}$, we have:
\begin{equation*}
\sca{a+\lambda b}{c}=\sca{a}{c}+\lambda\sca{b}{c}
\end{equation*}
\end{proposition}

\begin{corollary}[Homothety Lemma]\label{homothetie}
Conducts are closed under homothety: for all $\de{a}\in\cond{A}$ and all $\lambda\in\mathbf{R}$ with $\lambda\neq 0$, $\lambda\de{a}\in\cond{A}$.
\end{corollary}

In order to define the $\oplus$ connective, we need to be able to associate to a project $\de{a}$ in a conduct $\cond{A}$ a project $\extde{a}{B}$ in the conduct $\cond{A\oplus B}$. There is only one natural way of defining such a project $\extde{a}{V^{B}}$.

\begin{definition}
Let $\de{a}=(a,A)$ be a project with carrier $V^{A}$, and $V$ a finite set such that $V\cap V^{A}=\emptyset$. We will write $\extde{a}{V}$ the project $\de{a}\otimes \de{0}_{V}$.

If $\cond{A}$ is a conduct with support $V^{A}$, then $\extend{A}{V}$ will denote the set $\{\extde{a}{V}~|~\de{a}\in\cond{A}\}$. Usually, this construction will be used in the particular case where $V=V^{B}$ is the support of a conduct $\cond{B}$; when in this case we will use the notations $\extend{A}{B}$ and $\extde{a}{B}$ instead of $\extend{A}{V^{B}}$ and $\extde{a}{V^{B}}$.
\end{definition}

Let $\cond{A}$ be a conduct with carrier $V^{A}$ and let $V$ be a finite set such that $V\cap V^{A}=\emptyset$, then we want to define from $\cond{A}$ a conduct with carrier $V^{A}\cup V$. A conduct $\cond{A}$ can be considered as defined by its elements — the projects $\de{a}\in\cond{A}$ — or its test — the projects $\de{a'}\in\cond{A}^{\pol}$. This means that an extension of the conduct $\cond{A}$ should satisfy two properties: contain the extended elements of $\cond{A}$ and pass the extended tests of $\cond{A}$.

\begin{definition}
A \emph{conservative extension} of a conduct $\cond{A}$ along a set $V$ such that $V\cap V^{A}=\emptyset$ is a conduct $\cond{B}$ with carrier $V\cup V^{A}$ such that:
\begin{itemize}
\item $\cond{B}$ contains the extensions $\extde{a}{V}$ of the elements of $\de{a}\in\cond{A}$;
\item $\cond{B^{\pol}}$ contains the extensions $\extde{a'}{V}$ of the elements of $\de{a'}\in\cond{A^{\pol}}$;
\end{itemize}
\end{definition}

\begin{proposition}\label{dualextensions}
If $\cond{B}$ is a conservative extension of a conduct $\cond{A}$ on a set $V$, then $\cond{B}^{\pol}$ is a conservative extension of $\cond{A}^{\pol}$ on the set $V$.
\end{proposition}

\begin{proof}
By definition of conservative extensions, $\cond{B}^{\pol}$ contains the set $\extend{(A^{\pol})}{V}$ and $\cond{B}$ contains the set $\extend{A}{V}$. Since $\cond{A}^{\pol\pol}=\cond{A}$ and $\cond{B}^{\pol\pol}=\cond{B}$, the second property is equivalent to the fact that $\cond{B}^{\pol\pol}$ contains the set $\extend{(A^{\pol\pol})}{V}$. Hence $\cond{B}^{\pol}$ satisfies the properties of conservative extensions of $\cond{A}^{\pol}$ along $V$. 
\end{proof}

Notice that if $\extde{a}{V}$ and $\extde{a'}{V}$ are elements of $\extend{A}{V}$ and $\extend{(A^{\pol})}{V}$ respectively, we have $\sca{\extde{a}{V}}{\extde{a'}{V}}=\sca{a}{b}\neq 0,\infty$. This implies that the sets $\extend{A}{V}$ and $\extend{(A^{\pol})}{V}$ are contained in the orthogonal of the other:
\begin{eqnarray*}
\extend{A}{V}&\subset&(\extend{(A^{\pol})}{V})^{\pol}\\
\extend{(A^{\pol})}{V}&\subset&(\extend{A}{V})^{\pol}
\end{eqnarray*}

We will now consider two particular conservative extensions. The first is to extend the elements of $\cond{A}$ and take the bi-orthogonal closure of the set thus obtained, i.e. consider the conduct $(\extend{A}{V})^{\pol\pol}$. This is a conservative extension since it contains the set $\extend{A}{V}$ by definition, and we have $\extend{(A^{\pol})}{V}\subset (\extend{A}{V})^{\pol}$. This conservative extension intuitively corresponds to the set of projects that act as an element of $\cond{A}$ on $V^{A}$ and act as the project $\de{0}_{V}$.

The second possibility is to extend the set of tests of $\cond{A}$, i.e. the elements of $\cond{A}^{\pol}$, and take the orthogonal of this set, i.e. consider the conduct $(\extend{(A^{\pol})}{V})^{\pol}$. This is a conservative extension of $\cond{A}$ along $V$ since it is a conduct (as the orthogonal of a set, it is equal to its bi-orthogonal closure), which satisfies $\extend{A}{V}\subset(\extend{(A^{\pol})}{V})^{\pol}$ and such that its set of tests contains the set $\extend{(A^{\pol})}{V}$ by definition. This conservative extension intuitively corresponds to the set of projects that act as an element of $\cond{A}$ on $V^{A}$ without any restriction on what they can do on $V$ since the set of tests they have to pass act as $\de{0}_{V}$.

These two different "completions", which are the core of the definition of additives, intuitively seem to be the \emph{minimal} and \emph{maximal} conservative extensions of $\cond{A}$ along $V$. As it turns out, one can actually show that these intuitions are true.

\begin{proposition}
Let $\cond{B}$ be a conservative extension of $\cond{A}$ along $V$. Then:
$$(\extend{A}{V})^{\pol\pol}\subset\cond{B}\subset(\extend{(A^{\pol})}{V})^{\pol}$$
\end{proposition}

\begin{proof}
Since $\cond{B}$ is a conservative extension of $\cond{A}$ along $V$, it is a conduct containing $\extend{A}{V}$. This implies that $(\extend{A}{B})^{\pol\pol}\subset\cond{B}^{\pol\pol}$. Since $\cond{B}$ is a conduct, we have shown $(\extend{A}{B})^{\pol\pol}\subset\cond{B}$.

Since $\cond{B}$ is a conservative extension of $\cond{A}$ along $V$, its orthogonal is a conduct containing $\extend{(A^{\pol})}{V}$. We therefore obtain that $(\extend{(A^{\pol})}{V})^{\pol\pol}\subset\cond{B}^{\pol}$, i.e. $\cond{B}\subset (\extend{(A^{\pol})}{V})^{\pol}$.
\end{proof}

We will now proceed to define the additives on conducts. The underlying idea is simple. On the one hand, the set $\cond{A\oplus B}$ should contain projects that either act as an element of $\cond{A}$ on $V^{A}$ and as $\de{0}_{V^{B}}$ on $V^{B}$, or act as an element of $\cond{B}$ on $V^{B}$ and as $\de{0}_{V^{A}}$ on $V^{A}$.

On the other hand, the set $\cond{A\with B}$ should contain elements that act like an element of $\cond{A}$ on $V^{A}$ and act as an element of $\cond{B}$ on $V^{B}$. The natural way to define this set is therefore to consider the intersection of the maximal conservative extensions of $\cond{A}$ and of $\cond{B}$. 

\begin{definition}[Additive Connectives]
Let $\cond{A},\cond{B}$ be conducts. We define the conducts $\cond{A\with B}$ and $\cond{A\oplus B}$ as follows:
\begin{eqnarray*}
\cond{A\oplus B}&=&((\extend{A}{B})^{\pol\pol}\cup(\extend{B}{A})^{\pol\pol})^{\pol\pol}\\
\cond{A\with B}&=&(\extend{(A^{\pol})}{B})^{\pol}\cap(\extend{(B^{\pol})}{A})^{\pol}
\end{eqnarray*}
\end{definition}

One can check that these two definitions are dual.

\begin{proposition}\label{dualityadditives}
\begin{equation*}
\cond{A^{\pol}\oplus B^{\pol}}=(\cond{A\with B})^{\pol}
\end{equation*}
\end{proposition}

\begin{proof}
The orthogonal of a union is the intersection of the orthogonals, and thus:
\begin{eqnarray*}
\cond{A^{\pol}\oplus B^{\pol}}&=&((\extend{(A^{\pol})}{B})^{\pol\pol}\cup((\extend{(B^{\pol})}{A})^{\pol\pol})^{\pol\pol}\\
&=&((\extend{(A^{\pol})}{B})^{\pol\pol\pol}\cap(\extend{(B^{\pol})}{A})^{\pol\pol\pol})^{\pol}\\
&=&((\extend{(A^{\pol})}{B})^{\pol}\cap(\extend{(B^{\pol})}{A})^{\pol})^{\pol\pol}\\
&=&\cond{(A\with B)^{\pol}}
\end{eqnarray*}
\end{proof}

\begin{proposition}\label{compintoplus}
Let $\cond{A,B}$ be conducts. Then:
$$(\{\de{a}\otimes \de{0}_{\mathnormal{B}}~|~\de{a}\in\cond{A}\}\cup\{\de{0}_{\mathnormal{A}}\otimes\de{b}~|~\de{b}\in\cond{B}\})^{\pol\pol}=\cond{A\oplus B}$$
\end{proposition}

\begin{proof}
For any sets $A,B$ such that $\cond{A}=A^{\pol\pol}$ and $\cond{B}=B^{\pol\pol}$ one has:
\begin{eqnarray*}
(A\cup B)^{\pol}&=&A^{\pol}\cap B^{\pol}\\
(\cond{A}\cap \cond{B})^{\pol\pol}&=&(\cond{A}^{\pol}\cup\cond{B}^{\pol})^{\pol}
\end{eqnarray*}
The first equality — satisfied by any sets $A,B$ — is a consequence of the fact that $\de{a}\in(A\cup B)^{\pol}$ if and only if $\de{a}\in A^{\pol}$ and $\de{a}\in B^{\pol}$ if and only if $\de{a}\in A^{\pol}\cap B^{\pol}$. The second — which is true only when $\cond{A,B}$ are conducts — comes from the fact that $(\cond{A}\cap\cond{B})^{\pol\pol}=(\cond{A}^{\pol\pol}\cap\cond{B}^{\pol\pol})^{\pol\pol}=(\cond{A}^{\pol}\cup\cond{B}^{\pol})^{\pol\pol\pol}=(\cond{A}^{\pol}\cup\cond{B}^{\pol})^{\pol}$.\\
Denoting by $\extend{A}{B}$ (resp. $\extend{B}{A}$) the set $\{\de{a}\otimes \de{0}_{\mathnormal{B}}~|~\de{a}\in\cond{A}\}$ (resp. the set $\{\de{0}_{\mathnormal{A}}\otimes\de{b}~|~\de{b}\in\cond{B}\}$), we can use these equalities to obtain:
\begin{eqnarray*}
(\extend{A}{B}\cup \extend{B}{A})^{\pol\pol}&=&((\extend{A}{B})^{\pol}\cap (\extend{B}{A})^{\pol})^{\pol}\\
&=&((\extend{A}{B})^{\pol}\cap (\extend{B}{A})^{\pol})^{\pol\pol\pol}\\
&=&((\extend{A}{B})^{\pol\pol}\cup (\extend{B}{A})^{\pol\pol})^{\pol\pol}
\end{eqnarray*}
As a consequence, $\{\extde{a}{B}~|~\de{a}\in\cond{A}\}\cup\{\extde{b}{A}~|~\de{b}\in\cond{B}\}$ generates the behaviour $\cond{A\oplus B}$.
\end{proof}

It is quite natural to wonder if these additive connectives can be characterised in some way by means of conservative extensions. The answer to this question will turn out to be negative, as shown in \autoref{negativeanswer}, even though the next proposition seems to hint at a positive answer!

\begin{proposition}\label{commonconsext}
Let $\cond{A,B}$ be conducts with disjoint carriers. If there exists a conduct $\cond{C}$ which is both a conservative extension of $\cond{A}$ along $V^{B}$ and a conservative extension of $\cond{B}$ along $V^{A}$, then $\cond{A\oplus B}\subset\cond{C}\subset\cond{A\with B}$.
\end{proposition}

\begin{proof}
Since $\cond{C}$ is both a conservative extension of $\cond{A}$ and $\cond{B}$, it contains $\extend{A}{B}$ and $\extend{B}{A}$. It therefore contains $\extend{A}{B}\cup\extend{B}{A}$, which implies that $(\extend{A}{B}\cup \extend{B}{A})^{\pol\pol}\subset \cond{B}$.

On the other hand, $\cond{C^{\pol}}$ is both a conservative extension of $\cond{A^{\pol}}$ along $V^{B}$ and a conservative extension of $\cond{B^{\pol}}$ along $V^{A}$ by Proposition \ref{dualextensions}. As a consequence, $\cond{C}^{\pol}$ contains $\cond{A^{\pol}\oplus B^{\pol}}$. Taking the orthogonal reverses the inclusion and we obtain that $\cond{C}\subset(\cond{A^{\pol}\oplus B^{\pol}})$. We then conclude using Proposition \ref{dualityadditives}.
\end{proof}

This proposition, however, is true for all the wrong reasons, and conveys the wrong intuitions about the additive connectives. In fact, the inclusion $\cond{A\oplus B}\subset\cond{A\with B}$ is not satisfied in general, i.e. the proposition holds because there exists no common conservative extensions of $\cond{A}$ and $\cond{B}$. Let us try to describe what the fact that the inclusion $\cond{A\oplus B}\subset\cond{A\with B}$ is satisfied implies for $\cond{A}$ and $\cond{B}$.

\begin{lemma}
Let $\cond{C}$ be a conservative extension of a conduct $\cond{A}$ along a set $V^{B}$ disjoint from $V^{A}$. For every project $\de{c}\in\cond{C}$, $\de{c}\plug\de{0}_{V^{B}}\in\cond{A}$.
\end{lemma}

\begin{proof}
Let $\de{c}$ be an element of $\cond{C}$ and $\de{a'}$ an element of $\cond{A^{\pol}}$. By definition of conservative extensions, $\cond{C}^{\pol}$ contains $\extend{(A^{\pol})}{B}$, which means that $\sca{c}{a'\otimes 0_{\textnormal{$V^{B}$}}}\neq0,\infty$. By the trefoil property, this implies $\sca{c\plug 0_{\textnormal{$V^{B}$}}}{a'}\neq 0,\infty$, i.e. $\de{c\plug 0}_{V^{B}}\in\cond{A}$.
\end{proof}

\begin{proposition}
Let $\cond{A}$ and $\cond{B}$ be conducts with disjoint carriers such that $\cond{A\oplus B}\subset \cond{A\with B}$. Then:
\begin{itemize}
\item if $\cond{A}\neq\cond{0}_{V^{A}}$, then $\de{0}_{V^{B}}\in\cond{B}$;
\item if $\cond{A}\neq\cond{T}_{V^{A}}$, then $\de{0}_{V^{B}}\in\cond{B^{\pol}}$;
\end{itemize}
\end{proposition}

\begin{proof}
As a consequence of Proposition \ref{dualityadditives}, the inclusion $\cond{A\oplus B}\subset \cond{A\with B}$ is equivalent to the inclusion $\cond{A^{\pol}\oplus \cond{B}^{\pol}}\subset \cond{A^{\pol}\with B^{\pol}}$.

If $\cond{A\oplus B}\subset\cond{A\with B}$, then $\extend{A}{B}$ and $\extend{B}{A}$ are contained in $\cond{A\with B}$. By definition of $\with$, this means that in particular $\extend{A}{B}$ is contained in $(\extend{(B^{\pol})}{A})^{\pol\pol}$, i.e. all elements of the form $\de{a}\otimes\de{0}_{V^{B}}$ are in $(\extend{(B^{\pol})}{A})^{\pol\pol}$. Supposing that $\cond{A}$ is not empty it implies, using the preceding lemma, that $\de{0}_{V^{B}}\in\cond{B}$. 

Using the same reasoning on the inclusion $\cond{A^{\pol}\oplus \cond{B}^{\pol}}\subset \cond{A^{\pol}\with B^{\pol}}$, one shows that $\de{0}_{V^{A}}\in\cond{A^{\pol}}$ as long as $\cond{A^{\pol}}\neq\cond{0}_{V^{A}}$, i.e. as long as $\cond{A}\neq\cond{T}_{V^{A}}$.
\end{proof}

\begin{proposition}
Let $\cond{A}$ and $\cond{B}$ be conducts with disjoint carriers. Then $\cond{A\oplus B}\subset\cond{A\with B}$ if and only if $\cond{A,B}$ are both empty or $\cond{A,B}$ are both full.
\end{proposition}

\begin{proof}
By the preceding proposition, if $\cond{A}$ is neither equal to $\cond{0}_{V^{A}}$ nor $\cond{T}_{V^{A}}$, then $\de{0}_{V^{B}}$ is an element of both $\cond{B}$ and $\cond{B^{\pol}}$. This is of course a contradiction, since $\sca{\de{0}_{\textnormal{$V^{B}$}}}{\de{0}_{\textnormal{$V^{B}$}}}=0$. Similarly, we reach a contradiction if we suppose that $\cond{B}$ is neither equal to $\cond{0}_{V^{B}}$ nor equal to $\cond{T}_{V^{B}}$. As a consequence, $\cond{A}$ and $\cond{B}$ are either empty of full (the orthogonal of the empty conduct). 

Using this fact, one can then show that $\cond{A}=\cond{0}_{V^{A}}$ if and only if $\cond{B}=\cond{0}_{V^{B}}$ (and symmetrically, $\cond{B}=\cond{T}_{V^{B}}$ if and only if $\cond{A}=\cond{T}_{V^{A}}$). To show this, we just need to remark that $\cond{A}=\cond{0}_{V^{A}}$ implies that $\cond{A}\neq\cond{T}_{V^{A}}$. By the preceding lemma this implies that $\de{0}_{V^{B}}$ is an element in $\cond{B}^{\pol}$, hence $\cond{B}^{\pol}$ is non-empty, which implies that $\cond{B}=\cond{0}_{V^{B}}$. Conversely, if $\cond{B}=\cond{0}_{V^{B}}$, then $\cond{B}\neq\cond{T}_{V^{B}}$ hence $\de{0}_{V^{A}}\in\cond{A^{\pol}}$ by the preceding proposition, and finally $\cond{A}=\cond{0}_{V^{A}}$.
\end{proof}

\begin{corollary}
Let $\cond{A}$ and $\cond{B}$ be conducts with disjoint carriers $V^{A}, V^{B}$, and let $V=V^{A}\cup V^{B}$. Then $\cond{A\oplus B}\subset\cond{A\with B}$ if and only if $\cond{A\oplus B}=\cond{0}_{V}$ or $\cond{A\with B}=\cond{T}_{V}$.
\end{corollary}

\begin{proof}
We show that $\cond{A\oplus B}=\cond{0}_{V}$ if and only if $\cond{A}=\cond{0}_{V^{A}}$ and $\cond{B}=\cond{0}_{V^{B}}$. This is a direct consequence of the fact that $\cond{A\oplus B}=(\extend{A}{B}\cup\extend{B}{A})^{\pol\pol}$. If one of $\cond{A}$ and $\cond{B}$ is non-empty, then $\extend{A}{B}\cup\extend{B}{A}$ is non-empty, and therefore $\cond{A\with B}$ is non-empty since it contains $\extend{A}{B}\cup\extend{B}{A}$. Conversely, it is clear that $\cond{0_{V^{A}}\oplus 0_{V^{B}}}=\cond{0}_{V}$.

This implies by duality that $\cond{A\with B}=\cond{T}_{V}$ if and only if $\cond{A}=\cond{T}_{V^{A}}$ and $\cond{B}=\cond{T}_{V^{B}}$.

We then conclude by using the preceding proposition.
\end{proof}

\begin{corollary}\label{negativeanswer}
Let $\cond{A}$ and $\cond{B}$ be conducts with disjoint carriers $V^{A}, V^{B}$, and let $V=V^{A}\cup V^{B}$. If $\cond{A\oplus B}\neq \cond{0}_{V}$ and $\cond{A\with B}\neq\cond{T}_{V}$, then no conduct is both a conservative extension of $\cond{A}$ along $V^{B}$ and a conservative extension of $\cond{B}$ along $V^{A}$.
\end{corollary}

\begin{proof}
This is a simple consequence of the preceding corollary and Proposition \ref{commonconsext}.
\end{proof}

Even though, as we have just seen, the additive connectives we defined cannot be described in terms of conservative extensions, it turns out that they satisfy all the properties we expect. We already saw that the duality between them is satisfied. It is now a simple exercise to see that these connectives are commutative, associative and have neutral elements.

\begin{proposition}
The $\with$ connective is commutative, associative, and has a neutral element: the full conduct on the empty carrier $\cond{T}_{\emptyset}$.
\end{proposition}

However, these constructions leave us more or less empty-handed: how could one obtain a project in $\cond{A\with B}$ from two projects $\de{a,b}$ respectively in $\cond{A,B}$? We will need a more explicit construction of the $\with$ construction at the level of projects. In order to do this, one has to restrict to a particular class of conducts which we will call behaviours\footnote{The behaviours are the Interaction Graphs' counterpart of Girard's \emph{dichologies} in his recent "manifeste" \cite{syntran}}.

\subsubsection{behaviours}

\begin{definition}\label{behaviourdef}
A \emph{behaviour} $\cond{A}$ with carrier $V$ is a conduct $\cond{A}$ such that for all $\lambda\in\mathbf{R}$:
\begin{enumerate}
\item if $\de{a}=(a,A)\in\cond{A}$, then $\de{a}+\lambda \de{0}_{V}\in \cond{A}$.;
\item if $\de{a}=(a,A)\in\cond{A}^{\pol}$, then $\de{a}+\lambda \de{0}_{V}\in \cond{A}^{\pol}$.
\end{enumerate}
\end{definition}

\begin{remark}
The orthogonal of a behaviour is a behaviour.
\end{remark}

This definition of behaviour is however quite cumbersome and difficult to work with, as it is quite unnatural to reason about the elements of the orthogonal. We will thus first obtain a characterisation of behaviours  (Proposition \ref{charac}) which will simplify greatly the proofs of the results of this section.

\begin{proposition}\label{wagerfreeorth}
If $A$ is a non-empty set of projects with carrier $V$ such that $\de{a}\in A\Rightarrow \de{a+\lambda 0}_{V}\in A$, then any project in $A^{\pol}$ is wager-free, i.e. if $(a,A)\in A^{\pol}$ then $a=0$.
\end{proposition}

\begin{proof}
Chose $\de{a}=(a,A)\in\cond{A}$, which is possible since $\cond{A}$ is supposed to be non-empty. Then for any $\de{b}=(b,B)\in\cond{A}^{\pol}$, $\sca{a}{b}\neq 0$. But $\de{a}+\lambda\de{0}\in\cond{A}$ for any $\lambda\in\mathbb{R}$. Then $\sca{a+\lambda 0}{b}=\sca{a}{b}+b\lambda$ must be non-zero, hence $b$ must be equal to $0$.
\end{proof}

\begin{proposition}\label{fellorth}
If $A$ is a non-empty set of projects with the same carrier $V^{A}$ such that $(a,A)\in A$ implies $a=0$, then $\de{b}\in A^{\pol}$ implies $\de{b}+\lambda\de{0}_{V^{A}}\in A^{\pol}$ for all $\lambda\in\mathbb{R}$.
\end{proposition}

\begin{proof}
Chose $\de{a}=(a,A)\in\cond{A}$. For any project $\de{b}$, we have $\sca{a}{b+\lambda 0}=\sca{a}{b}+\lambda a$. Since $a=0$, we get that $\sca{a}{b+\lambda 0}=\sca{a}{b}$. Hence, if $\de{b}\in\cond{A}^{\pol}$, $\de{b+\lambda 0}\in\cond{A}^{\pol}$.
\end{proof}

\begin{corollary}\label{properbehaviour}\label{charact}
If a conduct $\cond{A}$ with carrier $V$ is such that
\begin{enumerate}
\item if $\de{a}\in\cond{A}$, then $\de{a}+\lambda\de{0}_{V}\in\cond{A}$;
\item if $\de{a}=(a,A)\in\cond{A}$ then $a=0$;
\item $\cond{A}$ is non-empty.
\end{enumerate}
Then $\cond{A}$ is a behaviour, and its orthogonal satisfies all the conditions above. We call such a behaviour \emph{proper}.
\end{corollary}

\begin{proof}
From the preceding proposition and the second and third conditions, we get that $\cond{A}^{\pol}$ satisfies that $\de{b}\in\cond{A}^{\pol}\Rightarrow \de{b+\lambda 0}\in\cond{A}^{\pol}$. Hence $\cond{A}$ is a behaviour. Moreover, if $\cond{A}^{\pol}$ were empty, any project $\de{a}$ with the same carrier as $\cond{A}$ would be in the orthogonal of $\cond{A}^{\pol}$. Hence $\cond{A}$ wouldn't be a conduct, since it contains only wager-free projects. Finally, all projects in $\cond{A}^{\pol}$ are wager-free from proposition \ref{wagerfreeorth}.
\end{proof}

\begin{proposition}\label{charac}
A behaviour is either proper, equal to $\cond{0}_{V}=\emptyset$ or equal to $\cond{T}_{V}=\cond{0}_{V}^{\pol}$.
\end{proposition}

\begin{proof}
Let $\cond{A}$ be a behaviour. If it is empty, then $\cond{A}=\cond{0}$. If $\cond{A}^{\pol}$ is empty, then $\cond{A}=\cond{T}$. In the other cases, since $\cond{A}^{\pol}$ is non-empty we get that $\cond{A}$ contains only wager-free projects from proposition \ref{wagerfreeorth}. Hence $\cond{A}$ is proper since it satisfies all the needed conditions.
\end{proof}

As a consequence of this proposition, it is easy to show that $\cond{1}$ and $\cond{\bot}$ are not behaviours. However, the class of behaviours, which is closed under taking the orthogonal, is closed under the multiplicative and additive connectives, as the following propositions state.

\begin{proposition}\label{closedmult}
If $\cond{A}$ and $\cond{B}$ are behaviours with disjoint carriers, then $\cond{A\multimap B}$ and $\cond{A\otimes B}$ are behaviours.
\end{proposition}

\begin{proof}
First suppose $\cond{A},\cond{B}$ are proper. Let $\de{f}\in\cond{A \multimap B}$, $\de{a}\in\cond{A}$ and $\de{b}\in\cond{B}^{\pol}$ be projects. Then, from the adjunction we have that $\de{f}\poll\de{a\otimes b}$ if and only if $\de{f\deplug a}\poll \de{b}$. Since $\cond{A\odot B^{\pol}}$ is non-empty and contains only wager-free projects, $\de{f}\in\cond{A\multimap B}$ implies $\de{f+\lambda 0}\in\cond{A\multimap B}$. Moreover, for all $\lambda\in\mathbb{R}$, $\de{f}\deplug\de{a+\lambda 0}\in\cond{B}$, hence $f\unit{A}+f\lambda+\scalar{F,A}= 0$. We therefore get that $f=0$. Hence $\cond{A\multimap B}$ contains only wager-free projects. Finally, either $\cond{A\multimap B}$ is empty, and therefore equal to $\cond{0}$, either it is non-empty and therefore satisfies the conditions needed to be a proper behaviour.

Now, if either $\cond{A}=\cond{0}$ or $\cond{B}=\cond{T}$, it is easy to see that $\cond{A\multimap B}=\cond{T}$.

The last case is when $\cond{A}=\cond{T}$ and $\cond{B}$ is proper or $\cond{B}=\cond{0}$ while $\cond{A}$ is proper. Then by definition of the linear implication we get that $\cond{A\multimap B}=\cond{0}$ in the second case. To prove that $\cond{T}\multimap \cond{A}=\cond{0}$ when $\cond{A}$ is proper, notice that if $\de{f}$ is with carrier $V\cup V^{A}$, then applying $\de{f}$ to a project with infinite wager yields a project with infinite wager. Hence, if $\cond{T\multimap A}$ were to be non-empty, $\cond{A}$ should contain a project with an infinite wager, i.e. $\cond{A}$ would not be proper.

From the duality of multiplicative connectives, $\cond{A\otimes B}=(\cond{A\multimap B^{\pol}})^{\pol}$. Since $\cond{A},\cond{B}^{\pol}$ are behaviours, $\cond{A\multimap B^{\pol}}$ is a behaviour. Therefore $\cond{A\otimes B}$ is the orthogonal of a behaviour, hence a behaviour.
\end{proof}

\begin{proposition}\label{closedadd}
Let $\cond{A}$, $\cond{B}$ be behaviours. Then $\cond{A\with B}$ and $\cond{A\oplus B}$ are behaviours.
\end{proposition}

\begin{proof}
Let $\cond{A}$ be a proper behaviour and let $V^{B}$ be such that $V^{B}\cap V^{A}=\emptyset$. Then if $\de{a}\otimes \de{0}\in\extend{A}{B}$, $\de{a\otimes 0+\lambda0}\in\extend{A}{B}$ since $\de{(a+\lambda 0)\otimes 0}=\de{a\otimes 0}+\lambda \de{0}$. Moreover, $\cond{A}$ contains only wager-free projects, hence $\extend{A}{B}$ contains only wager-free projects. Since $\extend{A}{B}$ is closed under expansion, non-empty and contains only wager-free projects, it generates a behaviour $(\extend{A}{B})^{\pol\pol}$. Thus if $\cond{A}$ is a proper behaviour, $\cond{A}^{\pol}$ is a proper behaviour, which implies as just showed that $((\extend{(A^{\pol})}{B})^{\pol\pol}$ is a behaviour, and finally $(\extend{(A^{\pol})}{B})^{\pol}$ is a behaviour.

If $\cond{A}=\cond{T}_{V^{A}}$, we have $\cond{T}_{V^{A}}^{\pol}=\cond{0}_{V^{A}}$, hence $((\extend{(\Topp{A}^{\pol})}{B})^{\pol}=\cond{T}_{V^{A}\cup V^{B}}$. Thus $(\extend{(A^{\pol})}{B})^{\pol}$ is a behaviour.

If $\cond{A}=\cond{0}_{V^{A}}$, then $\cond{A}^{\pol}=\cond{T}_{V^{A}}$, hence $(\extend{(A^{\pol})}{B})^{\pol}=\{\de{a}\otimes \de{0}_{V^{B}}~|~\de{a}\in\cond{T}_{V^{A}}\}^{\pol}=(\cond{T}_{V^{A}}\otimes\cond{1}_{V^{B}})^{\pol}=\cond{T}\multimap \cond{\bot}_{V^{B}}$. But, since there exists infinite wager projects in $\cond{T}$, any $\de{f}\in\cond{T}\multimap \cond{\bot}_{V^{B}}$ would yield a project with an infinite wager. Hence $\cond{\bot}_{V^{B}}$ should contain a project with an infinite wager. But this is not possible, since $\cond{\bot}_{V^{B}}\neq\cond{T}_{V^{B}}$. Hence $\cond{T}\multimap\cond{\bot}_{V^{B}}$ is necessarily empty, i.e. $(\extend{(\Zero{A}^{\pol})}{B})^{\pol}=\cond{0}_{V^{A}\cup V^{B}}$ is a behaviour.

This implies that if $\cond{A,B}$ are behaviours, the conduct $\cond{A\with B}$ is a behaviour, as the intersection of $(\extend{(A^{\pol})}{B})^{\pol}$ and $(\extend{(B^{\pol})}{A})^{\pol}$ which are behaviours.

Since the orthogonal of a behaviour is a behaviour, we get that if $\cond{A,B}$ are behaviours, then $\cond{A^{\pol},B^{\pol}}$ are behaviours, hence $\cond{A^{\pol}\with B^{\pol}}$ is a behaviour, and eventually $\cond{A\oplus B}=(\cond{A^{\pol}\with B^{\pol}})^{\pol}$ is a behaviour.
\end{proof}

\subsection{Additive Constructions on Projects}

Now, there is something more. We justified the restriction to behaviours by the fact that additive connectives on conducts were not obtained through an operation at the level of projects, and therefore did not allow us to interpret the additive rules of sequent calculus. As we show in the next results, this restriction bore its fruits since we are now able to define the interpretation of the $\with$ rule at the level of projects.

\begin{definition}
Let $\cond{A}$ and $\cond{B}$ be non-empty conducts with disjoint carriers, we define the set
\begin{equation*}
\cond{A}+\cond{B}=\{\extde{a}{B}+\extde{b}{A}~|~\de{a}\in\cond{A},\de{b}\in\cond{B}\}
\end{equation*}
\end{definition}

\begin{lemma}\label{plusdanswith}
Let $\cond{A,B}$ be non-empty behaviours. Then $\cond{A+B}\subset\cond{A\with B}$.
\end{lemma}

\begin{proof}
Choose $\de{f}=\de{a\otimes 0_{\text{$V^{B}$}}+b\otimes 0_{\text{$V^{A}$}}}\in\cond{A+B}$. By the preceding proposition, we can write:
\begin{eqnarray*}
\cond{A\with B}&=&\cond{(A^{\pol}\oplus B^{\pol})^{\pol}}\\
&=&(\{\extde{a}{B}~|~\de{a}\in\cond{A}\}\cup\{\extde{b}{A}~|~\de{b}\in\cond{B}\})^{\pol}
\end{eqnarray*}
It is thus sufficient to show that for any $\de{c}\in\extend{(A^{\pol})}{B}\cup\extend{(B^{\pol})}{A}$, one has $\de{f}\poll\de{c}$. We suppose, without loss of generality, that $\de{c}\in\extend{(A^{\pol})}{B}$, i.e. $\de{c}=\de{a'}\otimes \de{0}_{V^{B}}$ for a chosen $\de{a'}\in\cond{A}^{\pol}$. One can notice that 
$$\sca{b\otimes 0_{\text{$V^{A}$}}}{a'\otimes 0_{\text{$V^{B}$}}}=0$$ because the wagers of $\de{b}$ and $\de{a'}$ are equal to $0$. We thus obtain:
\begin{eqnarray*}
\sca{f}{c}&=&\sca{a\otimes 0_{\text{$V^{B}$}}+b\otimes 0_{\text{$V^{A}$}}}{a'\otimes 0_{\text{$V^{B}$}}}\\
&=&\sca{a\otimes 0_{\text{$V^{B}$}}}{a'\otimes 0_{\text{$V^{B}$}}}+\sca{b\otimes 0_{\text{$V^{A}$}}}{a'\otimes 0_{\text{$V^{B}$}}}\\
&=&\sca{a}{a'}
\end{eqnarray*}
Finally, $\de{f}\poll\de{c}$.
\end{proof}

However, the set $\cond{A+B}$ do not in general generate $\cond{A\with B}$, as we will show in the following proposition.

\begin{proposition}\label{counterexgraphs}
Suppose the map $m: ]0,1]\rightarrow\realposN\cup\{\infty\}$ takes a value $\mu\neq 0,\infty$ (i.e. there exists $x\in]0,1]$ such that $m(x)\neq 0,\infty$). If $\cond{A,B}$ are proper behaviours, then $\cond{(A+B)^{\pol[m]\pol[m]}}\subsetneq\cond{A\with B}$.
\end{proposition}

\begin{proof}
Let $\de{a'}\in\cond{A}^{\pol[m]}$, and define $\de{0}_{u}$ as the project $(0,U)$ where $U$ is the one-sliced graph with a unique edge from a source vertex $s_{a}\in V^{A}$ and a target vertex $s_{b}\in V^{B}$. Then for all project $\de{c}$ in $\cond{A+B}$, one has $\sca{c}{0_{\text{$U$}}}=0$ and therefore $\sca{c}{a'\otimes 0_{\text{$V^{B}$}}+0_{\text{$U$}}}=\sca{c}{a'\otimes 0_{\text{$V^{B}$}}}\neq 0,\infty$. As a consequence, 
$$\de{d}=\de{a'\otimes 0_{\text{$V^{B}$}}+0_{\text{$U$}}}\in\cond{(A+B)^{\pol}}$$

We will now show that for all $\nu$, $0<\lambda<1$ and $\de{a},\de{b}$ in $\cond{A,B}$ respectively, the project $\de{t_{\nu}}=\de{a+b+\nu 0}_{\lambda U^{\ast}}$ is an element of $\cond{(A^{\pol[m]}\oplus B^{\pol[m]})^{\pol[m]}}$ — where $U^{\ast}$ is the one-sliced graph with a unique edge whose source is $s_{b}$ and target if $s_{a}$. Indeed, for all $\de{a''\otimes 0}\in\extend{(A^{\pol[m]})}{B}$ (resp. $\de{b''\otimes 0}\in\extend{(B^{\pol[m]})}{A}$), we can compute:
\begin{equation*}
\begin{array}{c}
\sca{a''\otimes 0}{t_{\nu}}=\sca{a''\otimes 0}{a}+\sca{a''\otimes 0}{b}+\nu\sca{a''\otimes 0}{0_{\text{$\lambda U^{\ast}$}}}=\sca{a''}{a}\\
\sca{b''\otimes 0}{t_{\nu}}=\sca{b''\otimes 0}{a}+\sca{b''\otimes 0}{b}+\nu\sca{b''\otimes 0}{0_{\text{$\lambda U^{\ast}$}}}=\sca{b''}{b}
\end{array}
\end{equation*}
As a consequence,  we have that $\de{a''\otimes 0}\poll[m]\de{t_{\nu}}$ and $\de{b''\otimes 0}\poll[m]\de{t_{\nu}}$. This means that $\de{t_{\nu}}\in\cond{(A^{\pol[m]}\oplus B^{\pol[m]})^{\pol[m]}}$ (Proposition \ref{compintoplus}), i.e. $t_{\nu}\in\cond{A\with B}$.

But $\sca{t_{\nu}}{a'\otimes 0_{\text{$V^{B}$}}+0_{\text{$U$}}}=\sca{a'}{a}+\nu m(\lambda)$. Eventually changing the value of $\lambda$, one can suppose that $m(\lambda)=\mu\neq0,\infty$. Since $\sca{a'}{a}$ and $m(\lambda)$ are both different from $0$ and $\infty$, one can define $\nu=(-\sca{a'}{a})/\mu$. Then $\sca{t_{\nu}}{d}=0$, i.e. $\de{t_{\nu}}\not\in\cond{(A+B)^{\pol[m]\pol[m]}}$. 

Finally, we showed the inclusion $\cond{(A+B)^{\pol[m]\pol[m]}}\subset\cond{A\with B}$ — a consequence of Proposition \ref{plusdanswith} — is strict.
\end{proof}

To deal with this issue, we will define a notion of observational equivalence and show that even though the elements of the form $\de{a+b}$ do not generate the behaviour $\cond{A\with B}$, any element in $\cond{A\with B}$ is observationally equivalent to such a sum. We will now define the notion of equivalence and state some of its properties before showing this result.

\begin{definition}[Equivalence]\label{equivalencedef}
We define, given a conduct $\cond{A}$, an \emph{equivalence relation} on the set of projects in $\cond{A}$ as follows:
\begin{equation*}
\de{a}\cong_{\cond{A}}\de{b}\Leftrightarrow \forall \de{c}\in\cond{A}^{\pol}, \sca{a}{c}=\sca{b}{c}
\end{equation*}
We will denote the equivalence class of $\de{f}$ by $[f]_{\cond{A}}$, forgetting the subscript when it is clear from the context.
\end{definition}

\begin{proposition}\label{equivgoingup}
Let $\cond{A,B}$ be conducts such that $\cond{A\subset B}$. Then:
$$\de{a}\cong_{\cond{A}}\de{b}\Rightarrow \de{a}\cong_{\cond{B}}\de{b}$$
\end{proposition}

\begin{proof}
One only needs to notice that $\cond{A}\subset\cond{B}$ implies that $\cond{B}^{\pol}\subset\cond{A}^{\pol}$. Then:
\begin{eqnarray*}
\de{a}\cong_{\cond{A}}\de{b}&\Leftrightarrow&\forall\de{e}\in\cond{A}^{\pol},~ \sca{a}{e}=\sca{b}{e}\\
&\Rightarrow&\forall\de{e}\in\cond{B}^{\pol},~ \sca{a}{e}=\sca{b}{e}
\end{eqnarray*}
By definition of the equivalence, this shows that $\de{a}\cong_{\cond{A}}\de{b}\Rightarrow\de{a}\cong_{\cond{B}}\de{b}$.
\end{proof}

\begin{lemma}\label{compethics}
Let $E$ be a set of projects with carrier $V$, and $\de{a,b}\in E^{\pol}$. Then $\de{a}\cong_{E^{\pol}}\de{b}$ if and only if $\forall \de{e}\in E, \sca{a}{e}=\sca{b}{e}$.
\end{lemma}

\begin{proof}
By definition, if $\de{a}\cong_{E^{\pol}}\de{b}$, we have $\forall \de{f}\in E^{\pol\pol}, \sca{a}{f}=\sca{b}{f}$, which is equivalent to $\forall \lambda\in\mathbb{R}, \sca{\lambda a-\lambda b}{f}=0$. Thus the equivalence of $\de{a}$ and $\de{b}$ can be restated as $\forall \lambda\in\mathbb{R},\forall \de{c}\in E^{\pol}, \de{c+\lambda a-\lambda b}\in E^{\pol}$. This is by definition equivalent to $\forall \lambda\in\mathbb{R},\forall \de{c}\in E^{\pol},\forall \de{e}\in E, \sca{c+\lambda a-\lambda b}{e}\neq 0$, i.e. $\sca{\lambda a-\lambda b}=0$. Finally, we have showed that $\de{a}\cong_{E^{\pol}}\de{b}$ if and only if $\forall \de{e}\in E, \sca{a}{e}=\sca{b}{e}$.
\end{proof}

\begin{proposition}\label{sumwith}
Let $\cond{A},\cond{B}$ be non empty behaviours, and $\de{f}\in\cond{A\with B}$. Then there exists $\de{g}\in\cond{A}$ and $\de{h}\in\cond{B}$ such that
\begin{equation*}
\de{f}\cong_{\cond{A\with B}}\de{g+h}
\end{equation*}
\end{proposition}

\begin{proof}
Take a project $\de{f}\in\cond{A\with B}$. Since $\de{f}\in(\extend{(A^{\pol})}{B})^{\pol}\cap(\extend{(B^{\pol})}{A})^{\pol}$, we have that $\sca{f}{a'\otimes 0}\neq 0,\infty$ and $\sca{f}{0\otimes b'}\neq 0,\infty$ for all $\de{a'}\in\cond{A}^{\pol}, \de{b'}\in\cond{B}^{\pol}$. Hence $\de{g}=\de{f}\deplug \de{0}_{V^{B}}\in \cond{A}$ and $\de{h}=\de{f}\deplug \de{0}_{V^{A}}\in\cond{B}$. We are going to show that $\de{f}\cong_{\cond{A\with B}}\de{g}+\de{h}$. 

Notice first that $(\extend{(A^{\pol})}{B}\cup \extend{(B^{\pol})}{A})^{\pol}=(\extend{(A^{\pol})}{B})^{\pol}\cap (\extend{(B^{\pol})}{A})^{\pol}=\cond{A\with B}$. Hence, by lemma \ref{compethics}, we only need to show that $\sca{f}{c}=\sca{g+h}{c}$ for all $\de{c}\in \extend{(A^{\pol})}{B}\cup \extend{(B^{\pol})}{A}$ in order to prove that $\de{f}\cong_{\cond{A\with B}}\de{g+h}$.

Now, take for instance $\de{a'}\otimes \de{0}$ in $\extend{(A^{\pol})}{B}$. Then, since: 
$$\sca{a\otimes 0_{\text{V}}}{b\otimes 0_{\text{V}}}=\sca{a}{b}$$ 
We get:
\begin{eqnarray*}
\sca{g+h}{a'\otimes 0}&=&\sca{g\otimes 0}{a'\otimes 0}+\sca{0\otimes h}{a'\otimes 0}\\
&=&\sca{g\otimes 0}{a'\otimes 0}\\
&=&\sca{(f\deplug 0)\otimes 0}{a'\otimes 0}\\
&=&\sca{f\deplug 0}{a'}\\
&=&\sca{f}{a'\otimes 0}
\end{eqnarray*}
Similarly, for any $\de{b'}\otimes \de{0}$ in $\extend{(B^{\pol})}{A}$, we have:
$$\sca{g+h}{b'\otimes 0}=\sca{g}{b'\otimes 0}=\sca{f}{b'\otimes 0}$$
We have thus shown that $\de{f}\cong_{\cond{A\with B}}\de{g+h}$.
\end{proof}

\begin{proposition}\label{ethplus}
Let $\cond{A},\cond{B}$ be non empty behaviours, and $\de{f}\in\cond{A\oplus B}$. Then there exists $\de{h}\in\extend{A}{B}\cup\extend{B}{A}$ such that
\begin{equation*}
\de{f}\cong_{\cond{A\oplus B}}\de{h}
\end{equation*}
\end{proposition}

\begin{proof}
Let $\de{f}\in\cond{A\oplus B}$. Then, since $\cond{A\oplus B}\subset(\cond{A^{\pol}+B^{\pol}})^{\pol}$, we have that $\de{f}$ is orthogonal to $\de{a'+b'}$ for any $\de{a'}\in\cond{A^{\pol}}$ and $\de{b'}\in\cond{B^{\pol}}$. If $\cond{A}^{\pol},\cond{B}^{\pol}$ are non-empty, we can take $\de{a_{0}},\de{b_{0}}$ in them. Then, for any $\lambda,\mu$, the projects $\lambda\de{a_{0}}$ and $\mu\de{b_{0}}$ are in $\cond{A}^{\pol},\cond{B}^{\pol}$ respectively, hence $\sca{f}{\lambda a_{0}+\mu b_{0}}=\lambda\sca{f}{a_{0}}+\mu\sca{f}{b_{0}}\neq 0,\infty$. Since this must be true for any $\lambda,\mu\neq 0$, we have that either $\sca{f}{a_{0}}=0$ and $\sca{f}{b_{0}}\neq 0,\infty$, either $\sca{f}{b_{0}}=0$ and $\sca{f}{a_{0}}\neq 0,\infty$. 

Without loss of generality, we can suppose we are in the first case, that is: $\sca{f}{a_{0}}=0$ and $\sca{f}{b_{0}}\neq 0,\infty$. Then, for any $\de{a'}\in\cond{A^{\pol}}$, we have $\sca{f}{a'}=0$ since for all $\lambda$ and $\mu$, $\sca{f}{\lambda a'+\mu b_{0}}\neq0,\infty$. This implies that $\sca{f}{b'}\neq0,\infty$ for any $\de{b'}\in\cond{B}^{\pol}$. This gives us that $\sca{f}{b'\otimes 0_{\text{$V^{A}$}}}=\sca{f\deplug 0}{b'}\neq 0,\infty$, i.e. $\de{f\deplug 0}_{V^{A}}$ is in $\cond{B}$.

Now, we want to show that $\de{f}\cong_{\cond{A^{\pol}\with B^{\pol}}}\de{f\deplug 0}_{V^{A}}\otimes \de{0}_{V^{A}}$. For this, we chose an element $\de{g}\in\cond{A^{\pol}\with B^{\pol}}$. We want to show that $\sca{f}{g}=\sca{(f\deplug 0)\otimes 0}{g}$. Using the preceding proposition, there are two projects $\de{g_{1}}$ and $\de{g_{2}}$ in $\cond{A^{\pol}}$, $\cond{B^{\pol}}$ respectively, such that $\sca{g}{c}=\sca{g_{1}+g_{2}}{c}$ for all $\de{c}\in\cond{A\oplus B}$. Hence, we know that 
$$\sca{f}{g}=\sca{f}{g_{1}+g_{2}}~~~~~~\sca{(f\deplug 0)\otimes 0}{g}=\sca{(f\deplug 0)\otimes 0}{g_{1}+g_{2}}$$
Since $\sca{f}{a'}=0$ for all $\de{a'}\in\cond{A^{\pol}}$, we obtain:
$$\sca{f}{g_{1}+g_{2}}=\sca{f}{g_{1}}+\sca{f}{g_{2}}=\sca{f}{g_{2}}$$
On the other hand, $\sca{(f\deplug 0)\otimes 0}{g_{1}+g_{2}}=\sca{(f\deplug 0)\otimes 0}{g_{2}}+\sca{(f\deplug 0)\otimes 0}{g_{1}}$. 

\noindent Since $\sca{(f\deplug 0)\otimes 0}{g_{1}}=0$, we obtain:
\begin{eqnarray*}
\sca{(f\deplug 0)\otimes 0}{g_{1}+g_{2}}&=&\sca{(f\deplug 0)\otimes 0}{g_{2}}\\
&=&\sca{(f\deplug 0)\otimes 0}{g_{2}\otimes 0}\\
&=&\sca{f\deplug 0}{g_{2}}\\
&=&\sca{f}{g_{2}\otimes 0}\\
&=&\sca{f}{g_{2}}
\end{eqnarray*}
We finally obtained that the equality $\sca{f}{g}=\sca{(f\deplug 0)\otimes 0}{g}$ is satisfied for all $\de{g}\in\cond{A^{\pol}\with B^{\pol}}$.
\end{proof}

\subsection{Some Properties}

In order to prove distributivity, the following corollaries of Proposition \ref{ethtenscondtens} will be useful.

\begin{corollary}\label{tensor0tensor1}
Let $\cond{A}$ be a conduct. Then $(\extend{A}{V})^{\pol\pol}=\cond{A}\otimes \cond{1}_{V}$.
\end{corollary}

\begin{proof}
This is a simple application of Proposition \ref{ethtenscondtens} with $E=\cond{A}$ and $F=\{\de{0}_{V}\}$.
\end{proof}

\begin{proposition}\label{send0send1}
Let $\cond{A}$ be a conduct with carrier $W$, and $\de{f}$ a project with carrier $V$ such that $\de{f}\plug \de{0}_{V}\in\cond{A}$. Then $\de{f}\in \cond{1}_{V}\multimap \cond{A}$.
\end{proposition}

\begin{proof}
For all $\de{a'}\in\cond{A}^{\pol}$, we have $\sca{f}{a'\otimes 0_{\text{V}}}\neq 0,\infty$. Hence $\de{f}\in (\cond{A}^{\pol}\odot \{\de{0}_{V}\})^{\pol}$. By Proposition \ref{ethtenscondtens}, it implies that $\de{f}\in(\cond{A}^{\pol}\otimes\cond{1}_{V})^{\pol}$, and finally $\de{f}\in\cond{1}_{V}\multimap \cond{A}$.
\end{proof}

\begin{proposition}[Distributivity]\label{distributivity}
For any behaviours $\cond{A,B,C}$, and delocations $\phi,\psi,\theta,\rho$ of $\cond{A},\cond{A},\cond{B},\cond{C}$ respectively, there is a project $\de{distr}$ in the behaviour 
$$\cond{((\phi(A)\!\multimap\! \theta(B))\!\with\! (\psi(A)\!\multimap\! \rho(C)))\!\multimap\! (A\!\multimap\! (B\!\with\! C))}$$
\end{proposition}

\begin{proof}
Let $\de{g}$ be a project in $(\phi(A)\multimap \theta(B))\with (\psi(A)\multimap \rho(C))$. Using the definition of $\with$, and propositions \ref{duality} and \ref{tensor0tensor1}, we get:
\begin{eqnarray*}
\lefteqn{\cond{(\phi(A)\multimap \theta(B))\with (\psi(A)\multimap \rho(C))}}\\
&=&\cond{(((\phi(A)\multimap \theta(B))^{\pol}){\uparrow_{\text{$\psi(V^{A})\cup\rho(V^{C})$}}})^{\pol}}\cap \cond{(((\psi(A)\multimap \rho(C))^{\pol}){\uparrow_{\text{$\phi(V^{A})\cup\theta(V^{B})$}}})^{\pol}}\\
&=&\cond{(\phi(A)\otimes \theta(B)\otimes 1_{\text{$\phi(V^{A})\cup\rho(V^{C})$}})^{\pol}}\cond{\cap(\psi(A)\otimes \rho(C)\otimes 1_{\text{$\phi(V^{A})\cup\theta(V^{B})$}})^{\pol}}\\
&=&\cond{(\phi(A)\multimap(1_{\text{$\psi(V^{A})\cup\rho(V^{C})$}}\multimap \theta(B)))}\cond{\cap(\psi(A)\multimap(1_{\text{$\phi(V^{A})\cup\theta(V^{B})$}}\multimap \rho(C)))}
\end{eqnarray*}
Now, define the projects:
 \begin{eqnarray*}
 \de{f_{1}}&=&\de{Fax}_{\phi}\otimes\de{Fax}_{\theta}\otimes \de{0}_{\psi(V^{A})\cup\rho(V^{C})}\otimes \de{0}_{V^{C}}\\
 \de{f_{2}}&=&\de{Fax}_{\psi}\otimes\de{Fax}_{\rho}\otimes \de{0}_{\phi(V^{A})\cup\theta(V^{B})}\otimes \de{0}_{V^{B}}\\
 \de{distr}&=&\de{f_{1}+f_{2}}
\end{eqnarray*}
We have $\de{Distr\deplug g}=\de{f_{1}\deplug g}+\de{f_{2}\deplug g}$.
Let us compute $\de{(f_{1}\deplug g)\deplug a}$ for $\de{a}\in\cond{A}$:
\begin{eqnarray*}
&\de{(f_{1}\deplug g)\deplug a}\\
&=&\!\!\!\!\!\!\!\!\!\!((\de{Fax}_{\phi}\otimes\de{Fax}_{\theta}\otimes \de{0}_{\psi(V^{A})\cup\rho(V^{C})}\otimes \de{0}_{V^{C}})\deplug \de{g})\deplug \de{a}\\
&=&\!\!\!\!\!\!\!\!\!\!(\de{Fax}_{\theta}\otimes \de{0}_{\psi(V^{A})\cup\rho(V^{C})}\otimes \de{0}_{V^{C}})\deplug (g\deplug \phi(\de{a}))\\
&=&\!\!\!\!\!\!\!\!\!\!(\de{Fax}_{\theta}\otimes \de{0}_{V^{C}})\deplug((\de{g}\deplug \phi(\de{a}))\deplug \de{0}_{\psi(V^{A})\cup\rho(V^{C})})
\end{eqnarray*}
Since, as we have shown earlier, $\de{g}\in\cond{(\phi(A)\multimap(1_{\text{$\psi(V^{A})\cup\rho(V^{C})$}}\multimap \theta(B)))}$, the project $((\de{g}\deplug \phi(\de{a}))\deplug \de{0}_{\psi(V^{A})\cup\rho(V^{C})})$ is in $\cond{\theta(B)}$, hence it is equal to $\theta(\de{b})$ with $\de{b}\in\cond{B}$. This yields:
\begin{eqnarray*}
\de{(f_{1}\deplug g)\deplug a}&=&(\de{Fax}_{\theta}\otimes \de{0}_{V^{C}})\deplug\theta(\de{b})\\
&=&\de{b}\otimes \de{0}_{V^{C}}
\end{eqnarray*}
Similarly, one gets that $\de{(f_{2}\deplug g)\deplug a}=\de{c}\otimes \de{0}_{V^{B}}$ for $\de{c}\in\cond{C}$. Hence $(\de{distr}\deplug\de{g})\deplug\de{a}\in\cond{B+C}\subset\cond{B\with C}$. 

Eventually, $\de{distr}\deplug\de{g}\in\cond{A\multimap B\with C}$, which implies that the project $\de{distr}$ we defined is the one we were looking for.
\end{proof}

Distributivity is quite easy to grasp without going into details: the project $\de{Distr}$ just superimposes elements of $((\phi(A)\!\multimap\! \theta(B))$ and $(\psi(A)\!\multimap\! \rho(C)))$ over the carrier of $\cond{A}$.

\begin{proposition}\label{mix}
The mix rule is never satisfied for proper behaviours.
\end{proposition}

\begin{proof}
Let $\cond{A}$, $\cond{B}$ be behaviours, and let $\de{a}$, $\de{a'}$, $\de{b}$, and $\de{b'}$ be projects in the conducts $\cond{A}$, $\cond{A^{\pol}}$, $\cond{B}$, and $\cond{B^{\pol}}$ respectively. Then, noticing that:
\begin{eqnarray*}
\lefteqn{\scalar{A\cup B,A'\cup B'}}\\
&=&\sum_{i,j,k,l}\alpha^{A}_{i}\alpha^{B}_{j}\alpha^{A'}_{k}\alpha^{B'}_{l}\scalar{A_{i}\cup B_{j},A'_{k}\cup B'_{l}}\\
&=&\sum_{i,j,k,l}\alpha^{A}_{i}\alpha^{B}_{j}\alpha^{A'}_{k}\alpha^{B'}_{l}(\scalar{A_{i},A'_{k}}+\scalar{B_{j},B'_{l}})\\
&=&\unit{B}\unit{B'}\scalar{A,A'}+\unit{A}\unit{A'}\scalar{B,B'}
\end{eqnarray*}
we can compute $\sca{a\otimes b}{a'\otimes b'}$ as follows:
\begin{eqnarray*}
\lefteqn{\sca{a\otimes b}{a'\otimes b'}}\\
&=&\unit{A'}\unit{B'}(\unit{A}b+\unit{B}a)+\unit{A}\unit{B}(\unit{A'}b'+\unit{B'}a')\\
&&+\scalar{A\cup B,A'\cup B'}\\
&=&\unit{A}\unit{A'}(\unit{B'}b+\unit{B}b'+\scalar{B,B'})\\
&&+\unit{B}\unit{B'}(\unit{A'}a+\unit{A}a'+\scalar{A,A'})\\
&=&\unit{A}\unit{A'}\sca{b}{b'}+\unit{B}\unit{B'}\sca{a}{a'}
\end{eqnarray*}
Since $\sca{a}{a'}$ and $\sca{b}{b'}$ are different from $0$ and $\infty$, it is possible to make the last expression be equal to $0$ by changing the value of $\unit{A}$ for instance, using the fact that for $\de{c}=\de{a+\lambda 0}$, we have $\sca{c}{a'}=\sca{a}{a'}$ and that $\unit{C}=\unit{A}+\lambda$.
\end{proof}

\begin{proposition}\label{weak}
Weakening does not hold for non-empty behaviours.
\end{proposition}

\begin{proof}
Let $\cond{A,B}$ be conducts, let $\cond{C}$ be a behaviour, and let $\de{f}\in\cond{A\multimap B}$. Then $\de{f}\otimes\de{0}_{V^{C}}$ is not an element of $\cond{A\otimes C\multimap B}$. Indeed, chose $\de{a\otimes c}$ in $\cond{A\otimes C}$. Then $\sca{f\otimes 0}{a\otimes c}=\unit{C}\sca{f}{c}$. Moreover, $(F\otimes 0)\plug (A\otimes C)=\unit{C} F\plug A$. This yields that $\de{(f\otimes 0)\deplug (a\otimes c)}=\unit{C}\de{f\deplug a}$. Since $\cond{C}$ is a behaviour, it is possible to cancel $\unit{C}$ by considering $\de{c-\unit{\text{$C$}}0}\in\cond{C}$. Eventually, this gives that $\de{f\otimes 0}$ is not in $\cond{A\otimes C\multimap B}$ unless $\cond{B}=\cond{T}$ or $\cond{A}=\cond{0}$.
\end{proof}

\begin{remark}
This amounts to showing that there are no maps from $\cond{C}$ to $\cond{1}$ when $\cond{C}$ is a behaviour. 
\end{remark}

\subsection{Booleans}

One of the interests of the MALL fragment of Linear Logic is the possibility of defining booleans as proofs of $\cond{T\oplus T}$. Although it is possible to define a type of booleans in the multiplicative fragment, this additive version is more natural, especially when one considers that the proofs represented $\textnormal{True}$ and $\textnormal{False}$ are closely related to the usual lambda-calculus definitions $\textnormal{true}:=\lambda x. \lambda y. x$ and $\textnormal{false}:=\lambda x.\lambda y. y$. Our construction, however, satisfies a property that could be seen as a drawback in this perspective.

\begin{proposition}
$$\cond{T_{\emptyset}\oplus T_{\emptyset}}=\cond{T}_{\emptyset}$$
\end{proposition}

\begin{proof}
The conduct $\cond{T}_{\emptyset}$ is defined as the full behaviour, i.e. the orthogonal of the empty set, over the empty set. One can check that:
$$\cond{T}_{\emptyset}=\{(a,\emptyset)~|~ a\in \realN\}^{\pol\pol}=\{(\infty,\emptyset)\}^{\pol\pol}$$
By definition, $\cond{T_{\emptyset}\oplus T_{\emptyset}}$ is equal to the bi-orthogonal closure of the union of $\extend{T_{\emptyset}}{\emptyset}$ and itself. But $\extend{T_{\emptyset}}{\emptyset}=\{\de{t}\otimes \de{0}_{\emptyset}~|~ \de{t}\in\cond{T}_{\emptyset}\}=\{\de{t}\in\cond{T}_{\emptyset}\}$. This implies that $\cond{T_{\emptyset}\oplus T_{\emptyset}}=\cond{T}_{\emptyset}$.
\end{proof}

This shows that the classical way of defining booleans does not work in our setting. We could imagine a solution: replacing $\cond{T}_{\emptyset}$ by $\cond{T}_{V}$ where $V$ is a non-empty location. Unfortunately, this solution is not satisfying.

\begin{proposition}
Let $\{t,f\}$ be a two-element set. Then $\cond{T}_{\{t\}}\oplus \cond{T}_{\{f\}}$ is equal to $\cond{T}_{\{t,f\}}$.
\end{proposition}

\begin{proof}
The argument is quite straightforward. Since $\cond{T}_{\{t\}}$ contains a project with infinite wager, so does $\cond{T}_{\{t\}}\oplus \cond{T}_{\{f\}}$. Thus, by \autoref{miseinfcondpleine}, $\cond{T}_{\{t\}}\oplus \cond{T}_{\{f\}}$ is equal to $\cond{T}_{\{t,f\}}$.
\end{proof}

To define adequate booleans, we need a more radical change to get rid of infinite wagers. First notice that, as a consequence of the homothety lemma (\autoref{homothetie}), a conduct never contains exactly one element up to observational equivalence. We therefore need to consider conducts up to a more restrictive equivalence to obtain quotients of conducts that contain a finite number of elements.

\begin{definition}
Let $V$ be a set. A conduct $\cond{A}$ with carrier $V$ is \emph{singular} when $\cond{A}\neq\cond{T}_{V}$ and it contains exactly one element up to homothety and observational equivalence, i.e. for all $\de{a,b}\in\cond{A}$, there exists a real number $\lambda$ such that $\de{a}\equiv_{\cond{A}}\lambda\de{b}$.
\end{definition}

The following proposition shows that singular conducts can be used to define a boolean type.
\begin{proposition}
If $\cond{A,B}$ are singular conducts with respective (disjoint) carriers $T$ and $F$, the conduct $\cond{A\oplus B}$ contains exactly two elements up to homothety and observational equivalence.
\end{proposition}

\begin{proof}
Let us take an element $\de{c}$ in $\cond{A\oplus B}$. By Proposition \ref{ethplus}, we have two cases:
\begin{itemize}
\item either $\de{c}=\de{t}\otimes \de{0}_{F}$ where $\de{t}\in\cond{A}$;
\item or $\de{c}=\de{0}_{T}\otimes\de{f}$ where $\de{f}\in\cond{B}$.
\end{itemize}
By definition of singular conducts, any two projects $\de{p,q}$ in either $\cond{A}$ or $\cond{B}$ are equivalent up to scalar multiplication. Moreover, for any conduct $\cond{C}$ and set $V$ disjoint from the carrier of $\cond{C}$, $\de{p}\cong_{C}\de{q}$ implies that $\de{p}\otimes\de{0}_{V}\cong_{(\extend{C}{V})^{\pol\pol}}\de{q}\otimes\de{0}_{V}$. Choosing elements $\de{true}$ and $\de{false}$ in $\cond{A}$ and $\cond{B}$ respectively, we have that there exists a real number $\lambda$ such that $\de{c}$ is either equivalent to $\lambda.\de{true}\otimes\de{0}_{F}$ in the conduct $(\cond{A\uparrow_{F}})^{\pol\pol}$ or equivalent to $\lambda.\de{false}\otimes\de{0}_{T}$ in the conduct $(\cond{B{\uparrow_{T}}})^{\pol\pol}$. Since both $(\cond{A{\uparrow_{F}}})^{\pol\pol}$ and $(\cond{B{\uparrow_{T}}})^{\pol\pol}$ are subsets of the behaviour $\cond{A\oplus B}$, we finally obtain, by Proposition \ref{equivgoingup}:
\begin{itemize}
\item either $\de{c}\cong_{\cond{A\oplus B}}\lambda.\de{true}\otimes \de{0}_{F}$;
\item or $\de{c}\cong_{\cond{A\oplus B}}\lambda.\de{false}\otimes \de{0}_{T}$.\qedhere
\end{itemize}
\end{proof}

In order to define adequate booleans, we still need to understand if \emph{singular behaviours}, or at least \emph{singular conducts}, do exist. This is in fact depends on the chosen parameter map $m:]0,1]\rightarrow\mathbf{R}_{\geqslant 0}\cup\{\infty\}$. For instance, if one choses the map $m:x\mapsto\infty$, then \autoref{trivial} shows that all behaviours are either empty or full, thus no singular behaviours exist in this model. For this reason, we do not study further here the existence of singular behaviours and conducts.

\section{Denotational Semantics}\label{denotsection}

In this section, we will show how the GoI construction we obtained can be used to obtain a categorical model of MALL. The first part of the section is a quick overview of the definitions of the category \catmll{}. As this construction is presented in detail (though not in the sliced graphs setting) in our previous work \cite{seiller-goim}, we will not dwell on the details of the proofs which are straightforward adaptations of the corresponding proofs in the simpler setting of graphs.

\subsection{A $\ast$-autonomous category}

Let us first define the category of conducts. For this, we define $\psi_{i}: \mathbf{N}\rightarrow \mathbf{N}\times\{0,1\}$ ($i=0,1$):
\begin{equation*}
\psi_{i}:  x  \mapsto  (x,i)
\end{equation*}

\begin{definition}[Objects and morphisms of \catmll]
We define the following category:
\begin{equation*}
\left.\begin{array}{l}\mathrm{Obj}=\{\cond{A}~|~\cond{A}=\cond{A}^{\pol\pol}\text{ with carrier }X_{\cond{A}}\subset\mathbf{N}\}\\
\mathrm{Mor}[\cond{A},\cond{B}]=\{\de{f}\in \psi_{0}(\cond{A})\multimap\psi_{1}(\cond{B})\}\end{array}\right.
\end{equation*}
\end{definition}

To define the composition of morphisms, we will use three copies of $\mathbf{N}$. We thus define the following useful bijections:
\begin{equation*}
\begin{array}{rrclrrcl}
\mu: & \!\!\mathbf{N}\times\{0,1\} &\!\!\!\! \rightarrow \!\!\!\!& \mathbf{N}\times\{1,2\},&\!\!\!\!& (x,i) & \!\!\mapsto\!\! & (x,i+1)\\
\nu: & \!\!\mathbf{N}\times\{0,2\} &\!\!\!\! \rightarrow \!\!\!\!& \mathbf{N}\times\{0,1\},&\!\!\!\!& (x,i) & \!\!\mapsto\!\! & (x,i/2)\end{array}
\end{equation*}

\begin{definition}[Composition in \catmll]
Given two morphisms $\de{f}$ and $\de{g}$ in $\mathfrak{Mor}[\cond{A},\cond{B}]$ and $\mathfrak{Mor}[\cond{B},\cond{C}]$ respectively, we define 
\begin{equation*}
\de{g}\circ\de{f}=\nu(\de{f}\deplug\mu(\de{g}))
\end{equation*}
\end{definition}

Then one can show that this is indeed a category \cite{seiller-goim}. Notice the identities are defined by faxes (Definition \ref{faxs}) which are represented by finite graphs.

We now define a bifunctor $\otimes$, and for that we will use the functions $\phi: \mathbf{N}\times\{0,1\}\rightarrow \mathbf{N}$ defined\footnote{This function, called the \emph{Cantor pairing}, is heavily used in the first GoI constructions. Its use here has the exact same purpose, although it did appear in the very definition of connectives in former GoI models while we only use it to define the categorical models. This difference comes from the fact that we are defining a \emph{localised GoI} model, while former GoI constructions were \enquote{delocalised} from the start.} by $\phi((x,i))=2x+i$ and $\tau$
\begin{equation*}
\tau: \left\{\begin{array}{rcl}
\mathbf{N}\times\{0,1\}&\rightarrow&\mathbf{N}\times\{0,1\}\\
(2x+1,0)&\mapsto& (2x,1)\\
(2x,1)&\mapsto&(2x+1,0)\\
(x,i)&\mapsto&(x,i)\text{ otherwise}\end{array}\right.
\end{equation*}

\begin{definition}
We define on \catmll the bifunctor $\bar{\otimes}$ induced by the tensor product. It is defined on objects by
\begin{equation*}
\cond{A}\bar{\otimes}\cond{B}=\phi(\psi_{0}(\cond{A})\otimes\psi_{1}(\cond{B}))
\end{equation*}
and on morphisms by
\begin{equation*}
\de{f}\bar{\otimes} \de{g}=\tau(\psi_{0}(\phi(\de{f}))\otimes\psi_{1}(\phi(\de{g})))
\end{equation*}
\end{definition}

\begin{theorem}
The category \catmll is a $\ast$-autonomous category. More precisely, (\catmll,$\bar{\otimes}$,$\cond{1}_{\emptyset}$) is symmetric monoidal closed and the object $\cond{\bot}_{\emptyset}=\cond{1}_{\emptyset}^{\pol}$ is dualizing.
\end{theorem}

\begin{proof}
A proof of this result for directed weighted graphs can be found in our earlier paper \cite{seiller-goim}. The proof in this case (sliced graphs) is a straightforward adaptation of it.
\end{proof}

\subsection{Products and Coproducts}

Moreover, to have a denotational semantics for MALL, we need to define a product. The natural construction would be to define the category \catmall:
\begin{equation*}
\left.\begin{array}{l}\mathrm{Obj}=\{\cond{A}~|~\cond{A}=\cond{A}^{\pol\pol}\text{ behaviour with carrier }X_{\cond{A}}\subset\mathbf{N}\}\\
\mathrm{Mor}[\cond{A},\cond{B}]=\{\de{f}\in \psi_{0}(\cond{A})\multimap\psi_{1}(\cond{B})\}\end{array}\right.
\end{equation*}
and then define on this full subcategory of \catmll the bifunctor $\bar{\with}$ by $\cond{A},\cond{B}\mapsto \phi(\psi_{0}(\cond{A})\bar{\with}\psi_{1}(\cond{B}))$ on objects, and:
\begin{equation*}
\de{f}\bar{\with}\de{g}=\tau(\de{Distr}\deplug \psi_{0}(\phi(\de{f})\with\psi_{1}(\phi(\de{g})))
\end{equation*}

However, this does not define a categorical product. Indeed, as usual when dealing with geometry of interaction for additives, the problem lies in the elimination of the cut between additive connectives. 

Here is what happens on an example. Let us take two projects $\de{f,g}$ in respectively $\cond{A\multimap B}$ and $\cond{A\multimap C}$. Suppose moreover that both projects have only one slice to simplify the following discussion; we think of $\de{f,g}$ as interpretations of two sequent calculus proofs $\pi_{f}$ and $\pi_{g}$. Then, $\de{f}\bar{\with}\de{g}$ is a project with two slices in $\cond{A\multimap (B\with C)}$, where the first slice contains the graph $F\cup\emptyset_{V^{C}}$, and the second contains the graph $G\cup\emptyset_{V^{B}}$. We want $\de{f}\bar{\with}\de{g}$ to be the interpretation of the proof obtained from $\pi_{f}$ and $\pi_{g}$ by a $\with$ rule; we will denote this proof by $\pi_{\with}$. Now, let us take a project $\de{h}$ (once again we suppose it has only one slice to ease the discussion) in $B\multimap D$, and let us think of it as the interpretation of a proof $\pi_{h}$. Then the project $\de{h}\otimes \de{0}_{V^{C}}$ is in $\cond{(B\with C)\multimap D}$ (we use here Lemma \ref{kindofsemidistrib} which will be shown in the next section). This project will be the interpretation of the proof obtained from $\pi_{h}$ by applying a $\oplus$ rule, which we will denote by $\pi_{\oplus}$. Taking the cut between $\de{f}\bar{\with}\de{g}$ and $\de{h}\otimes \de{0}_{V^{C}}$ then gives us the interpretation of the proof $\pi$ obtained by applying a cut rule between $\pi_{\with}$ and $\pi_{\oplus}$. Applying one step of the cut-elimination procedure on $\pi$ then gives us a proof $\pi'$ whose interpretation should be $\de{f}\deplug\de{h}$. This raises the question: is $\de{p'}=\de{f}\deplug\de{h}$ equal to $(\de{f}\bar{\with}\de{g})\deplug(\de{h\otimes 0}_{V^{C}})$ ? One can easily see that it is never the case! Indeed, $(\de{f}\bar{\with}\de{g})\deplug\de{h}$ is equal to $\de{p}=((\de{f\otimes 0}_{V^{C}}) \deplug\de{h})+((\de{g\otimes 0}_{V^{B}})\deplug\de{0}_{V^{C}})$ (see the graphic representation in Figure \ref{figureadd}). It is clear that $\de{p}$ is not equal to $\de{p'}$ since $\de{p}$ has two slices, while $\de{p'}$ has only one. The situation is actually even worse: even though one slice of $\de{p}$ is equal to $\de{p'}$, the other one is not in general equal to the empty graph. Indeed, in the result of the execution $(G\cup(V^{B},\emptyset))\plug (H\cup(V^{C},\emptyset))$ one keeps the edges in $G$ whose source and target are in $V^{A}$ and the edges in $H$ whose source and target are in $V^{D}$.

\begin{figure*}
\begin{center}
\begin{minipage}{12cm}
\begin{center}
\begin{tikzpicture}[x=1.2cm,y=1.2cm]
	\node (S1) at (0,2) {Slice 1};
	\node (L) at (0,1) {Locations};
	\node (S2) at (0,0) {Slice 2};
	\node (F1) at (1,2.5) {$\de{f\otimes 0}_{V^{C}}$};
	\node (F2) at (1,1.5) {$\de{h\otimes 0}_{V^{C}}$};
	\node (F1) at (1,0.5) {$\de{g\otimes 0}_{V^{B}}$};
	\node (F2) at (1,-0.5) {$\de{h\otimes 0}_{V^{C}}$};
	\node (VA) at (2,1) {$V^{A}$};
	\node (VB) at (4,1) {$V^{B}$};
	\node (VC) at (6,1) {$V^{C}$};
	\node (VD) at (8,1) {$V^{D}$};
	\node (VA1) at (2,2) {$\cdot$};
	\node[shape=circle,draw] (VB1) at (4,2) {$\cdot$};
	\node[shape=circle,draw] (VC1) at (6,2) {$\cdot$};
	\node (VD1) at (8,2) {$\cdot$};
	\node (VA2) at (2,0) {$\cdot$};
	\node[shape=circle,draw] (VB2) at (4,0) {$\cdot$};
	\node[shape=circle,draw] (VC2) at (6,0) {$\cdot$};
	\node (VD2) at (8,0) {$\cdot$};

	\draw[-,dotted] (VA) -- (VA1) {};
	\draw[-,dotted] (VB) -- (VB1) {};
	\draw[-,dotted] (VC) -- (VC1) {};
	\draw[-,dotted] (VD) -- (VD1) {};
	\draw[-,dotted] (VA) -- (VA2) {};
	\draw[-,dotted] (VB) -- (VB2) {};
	\draw[-,dotted] (VC) -- (VC2) {};
	\draw[-,dotted] (VD) -- (VD2) {};
	
	\draw[->] (VA1) .. controls (1.5,3) and (2.5,3) .. (VA1) {};
	\draw[->] (VA1) .. controls (2.5,2.5) and (3.5,2.5) .. (VB1) {};
	\draw[->] (VB1) .. controls (3.5,3) and (2.5,3) .. (VA1) {};
	\draw[->] (VB1) .. controls (3.5,3) and (4.5,3) .. (VB1) {};
	
	\draw[->] (VB1) .. controls (3.5,1) and (4.5,1) .. (VB1) {};
	\draw[->] (VB1) .. controls (5,1.5) and (7,1.5) .. (VD1) {};
	\draw[->] (VD1) .. controls (7,1) and (5,1) .. (VB1) {};
	\draw[->] (VD1) .. controls (7.5,1) and (8.5,1) .. (VD1) {};

	\draw[->] (VA2) .. controls (1.5,1) and (2.5,1) .. (VA2) {};
	\draw[->] (VA2) .. controls (2.5,0.5) and (5.5,0.5) .. (VC2) {};
	\draw[->] (VC2) .. controls (5.5,1) and (2.5,1) .. (VA2) {};
	\draw[->] (VC2) .. controls (5.5,1) and (6.5,1) .. (VC2) {};

	\draw[->] (VB2) .. controls (3.5,-1) and (4.5,-1) .. (VB2) {};
	\draw[->] (VB2) .. controls (5,-0.5) and (7,-0.5) .. (VD2) {};
	\draw[->] (VD2) .. controls (7,-1) and (5,-1) .. (VB2) {};
	\draw[->] (VD2) .. controls (7.5,-1) and (8.5,-1) .. (VD2) {};	
	
\end{tikzpicture}
\end{center}
\footnotetext{Here an edge from $V^{i}$ to $V^{j}$ represents the set of edges whose sources are in $V^{i}$ and targets are in $V^{j}$. We did not represent those sets of edges that are necessarily empty (for instance from $V^{C}$ to $V^{C}$ in the graph of $\de{f\otimes 0}$). The circled dots represent the location of the cut, i.e. the vertices that disappear during the execution.}
\end{minipage}
\caption{Graphic representation of the plugging of $\de{f}\bar{\with}\de{g}$ and $\de{h\otimes 0}_{V^{C}}$.} \label{figureadd}
\end{center}
\end{figure*}

So, the categories considered are not a denotational semantics for MALL, and the $\with$ seems to be a bad candidate for defining a product. We will now see how one can solve this issue by using the observational equivalence.

\subsection{Observational Equivalence}

As we explained earlier, the $\with$ connective does not define a categorical product because, if $\de{f}\in\cond{A\multimap B}$, $\de{g}\in\cond{A\multimap C}$ and $\de{b}\in\cond{B^{\pol}}$, the computation of the cut $(\de{f}\bar{\with}\de{g})\deplug(\de{b\otimes 0})_{V^{C}}$ yields $\de{f}+\de{res}$ where $\de{res}$ is a residue equal to $(\de{g\otimes 0}_{V^{B}})\deplug(\de{b}\otimes\de{0}_{V^{C}})$. The following proposition shows, however, that this residue is not detected by the elements of $\cond{A^{\pol}}$, i.e. that $\sca{res}{a'}=0$ for all $\de{a'}\in\cond{A^{\pol}}$. This means that $\with$ defines a categorical product \emph{up to observational equivalence}.

\begin{proposition}\label{propprod}
Let $\de{f}\in\cond{A\multimap B}$ and $\de{g}\in\cond{A\multimap C}$ be projects, and write $\de{h}=\de{f}\bar{\with}\de{g}$. Then for all $\de{b}\in\cond{B}^{\pol}$, $\de{f}\deplug \de{b}\cong_{\cond{A}}\de{h}\deplug \de{b}\otimes\de{0}_{V^{C}}$.
\end{proposition}

\begin{proof}
For any $\de{a}\in \cond{A}$, we have $\de{h}\deplug (\extde{a}{V^{B}})=\de{f}\deplug\de{a}+\de{g}\deplug \de{0}_{V^{B}}$. Now, let $\de{b}\in\cond{B}^{\pol}$ be a project. Then $\sca{h\deplug a{\uparrow_{\text{$V$}^{\text{$B$}}}}}{b}=\sca{f\deplug a}{b}+\sca{g\deplug 0_{\text{$V$}^{\text{$B$}}}}{b}$. Now, suppose that $\sca{g\deplug 0_{\text{$V$}^{\text{$B$}}}}{b}=\lambda\neq 0$. Since $\de{f\deplug a}\in\cond{B}$, we have $\sca{f\deplug a}{b}=\mu\neq 0$. But, by the homothety lemma \ref{homothetie}, we get $\de{g}=-\frac{\lambda}{\mu}\de{f}\in\cond{A\multimap B}$. Since $\de{g}\deplug\de{a}=-\frac{\lambda}{\mu}\de{f}\deplug\de{a}$, we finally obtain that $\sca{distr\deplug (f\with g)}{b}=0$, which is a contradiction.

Thus, for any $\de{b}\in\cond{B}^{\pol}$, we have $\sca{g\deplug 0_{\text{$V$}^{\text{$B$}}}}{b}=0$.
\end{proof}

As  a consequence, we would like to quotient the category \catmll by the observational equivalence. For this, we need to show that the categorical structure we have does not collapse when taking the observational quotients. The following proposition — an easy consequence of the trefoil property (Theorem \ref{cyclicppty}) — and its corollaries\footnote{They do more than just that, they also ensure that the quotiented category inherits the monoidal structure of \catmll.} make sure of that.

\begin{proposition}\label{propquot}
Let $\de{f}\cong_{\cond{A\multimap B}}\de{f'}$ and $\de{g}\in\cond{B\multimap C}$ be projects. Then $\de{f\deplug g}\cong_{\cond{A\multimap C}}\de{f'\deplug g}$. 
\end{proposition}

\begin{proof}
For all $c\in\cond{C\multimap A}$, we have $\sca{f\deplug g}{c}=\sca{f}{g\deplug c}$. Therefore:
\begin{eqnarray*}
\sca{f\deplug g}{c}&=&\sca{f}{g\deplug c}\\
&=&\sca{f'}{g\deplug c}\\
&=&\sca{f'\deplug g}{c}
\end{eqnarray*}
\end{proof}

\begin{corollary}\label{corquot}
Let $\de{f},\de{f'},\de{g},\de{g'}$ be projects such that $\de{f}\cong_{\cond{A\multimap B}} \de{f'}$ and $\de{g}\cong_{\cond{B}\multimap \cond{C}}\de{g'}$. Then $\de{f\deplug g}\cong_{\cond{A\multimap C}}\de{f'\deplug g'}$.
\end{corollary}

\begin{corollary}
Let $\de{a}\cong_{\cond{A}}\de{a'}$ and $\de{f}\cong_{\cond{A\multimap B}}\de{g}$ be projects. Then $\de{f}\deplug\de{a}\cong_{\cond{B}}\de{g}\deplug\de{a'}$.
\end{corollary}

\begin{corollary}
If $\de{a}\cong_{\cond{A}}\de{a'}$ and $\de{b}\cong_{\cond{B}}\de{b'}$, then $\de{a}\otimes\de{b}\cong_{\cond{A\otimes B}}\de{a'}\otimes\de{b'}$.
\end{corollary}

Corollary \ref{corquot} shows that the observational equivalence defines a congruence on the category \catmll{}. This allows us to define the following quotient category:

\begin{definition}
Define the category \concat by
\begin{equation*}
\left.\begin{array}{l}\mathrm{Obj}=\{\cond{A}~|~\cond{A}=\cond{A}^{\pol\pol}\text{ with carrier }X_{\cond{A}}\subset\mathbf{N}\}\\
\mathrm{Mor}[\cond{A},\cond{B}]=\{[f]~|~f\in\psi_{0}(\cond{A})\multimap\psi_{1}(\cond{B})\}\end{array}\right.
\end{equation*}
\end{definition}

\begin{proposition}
The category \concat inherits the $\ast$-autonomous structure of the category \catmll.
\end{proposition}

\begin{proof}
Notice that we quotient only the hom-sets. The three corollaries of Proposition \ref{propquot} above ensure that we have indeed defined a category, and that it inherits the monoidal structure of \catmll. To show that \concat is closed, one shows that the isomorphism between $\mathfrak{Mor}[\cond{A},\cond{B}\bar{\multimap}\cond{C}]$ and $\mathfrak{Mor}[\cond{A}\bar{\otimes}\cond{B},\cond{C}]$ in \catmll is compatible with the equivalence relation. This compatibility is however obvious: equivalence is preserved by delocations. The fact that $\bot$ is dualizing is a direct consequence of the preservation of isomorphisms when one takes the quotient.
\end{proof}

\begin{definition}
Define the category \behcat by
\begin{equation*}
\left.\begin{array}{l}\mathrm{Obj}=\{\cond{A}~|~\cond{A}\text{ behaviour with carrier }X_{\cond{A}}\subset\mathbf{N}\}\\
\mathrm{Mor}[\cond{A},\cond{B}]=\{[f]~|~f\in\psi_{0}(\cond{A})\multimap\psi_{1}(\cond{B})\}\end{array}\right.
\end{equation*}
\end{definition}

\begin{proposition}
The category \behcat is a full subcategory of \concat closed under the monoidal product, the internalisation of Hom-sets and duality, which has products, coproducts and in which mix and weakening do not hold.
\end{proposition}

\begin{proof}
It is sufficient to prove that $\oplus$ is a coproduct, since \behcat is closed under taking the orthogonal. Let $\cond{A,B,C}$ be behaviours, and $\de{f}\in\cond{A\multimap C}$, $\de{g}\in\cond{B\multimap C}$ be projects. Then the project $\de{f}\bar{\with}\de{g}$ is a project in $\cond{(A\oplus B)\multimap C}$. Define $\iota_{\cond{A}}$ (resp. $\iota_{\cond{B}}$) as the identity on $\cond{A}$ (resp. on $\cond{B}$) tensored with $\de{0}_{V^{B}}$ (resp. $\de{0}_{V^{A}}$). Then, it is an easy consequence of Proposition \ref{propprod} that, for any representative $\de{h}$ of $[\de{f}\bar{\with}\de{g}]$ and any representative $\de{i}$ of $[\iota_{\cond{A}}]$, we have $\de{h}\deplug\de{i}\in[\de{f}]$. The verification concerning $\iota_{\cond{B}}$ is similar.
\end{proof}

So the categorical model we obtain has two layers (see Figure \ref{catmodels}). The first layer consists in a non-degenerate (i.e. $\otimes\neq\parr$ and $\cond{1}\neq\cond{\bot}$) $\ast$-autonomous category \concat{}, hence a denotational model for MLL with units. The second layer is the full subcategory \behcat which does not contain the multiplicative units but is a non-degenerate model (i.e. $\otimes\neq\parr$, $\otimes\neq\with$ and $\cond{0}\neq\cond{\top}$) of MALL with additive units that does not satisfy the mix and weakening rules.

\begin{figure}
\centering
\begin{tikzpicture}[x=1.2cm,y=0.8cm]
	\draw[fill,opacity=0.1] (0,0) .. controls  (0,4.5) and (0.5,5) .. (5,5) .. controls (9.5,5) and (10,4.5) .. (10,0) .. controls (10,-4.5) and (9.5,-5) .. (5,-5) .. controls (0.5,-5) and (0,-4.5) .. (0,0) ;
		\node (A) at (2,0) {\begin{tabular}{c}\small{\concat}\\\small{($\ast$-autonomous)}\end{tabular}};
	\draw[fill,opacity=0.2] (4,0) .. controls  (4,2.5) and (4.5,3) .. (6,3) .. controls (7.5,3) and (8,2.5) .. (8,0) .. controls (8,-2.5) and (7.5,-3) .. (6,-3) .. controls (4.5,-3) and (4,-2.5) .. (4,0) ;
		\node (B) at (6,0) {\begin{tabular}{c}\small{\behcat}\\\small{(closed under $\otimes,\multimap,\with,\oplus,(\cdot)^{\pol}$)}\\\small{NO weakening, NO mix}\end{tabular}};
	\node (bot) at (2,-3) {$\bullet_{\bot}$};
	\node (one) at (3,-3) {$\bullet_{\cond{1}}$};
	
	\node (top) at (6,-2) {$\bullet_{\cond{T}}$};
	\node (zero) at (7,-2) {$\bullet_{\cond{0}}$};
\end{tikzpicture}
\caption{The categorical models}\label{catmodels}
\end{figure}
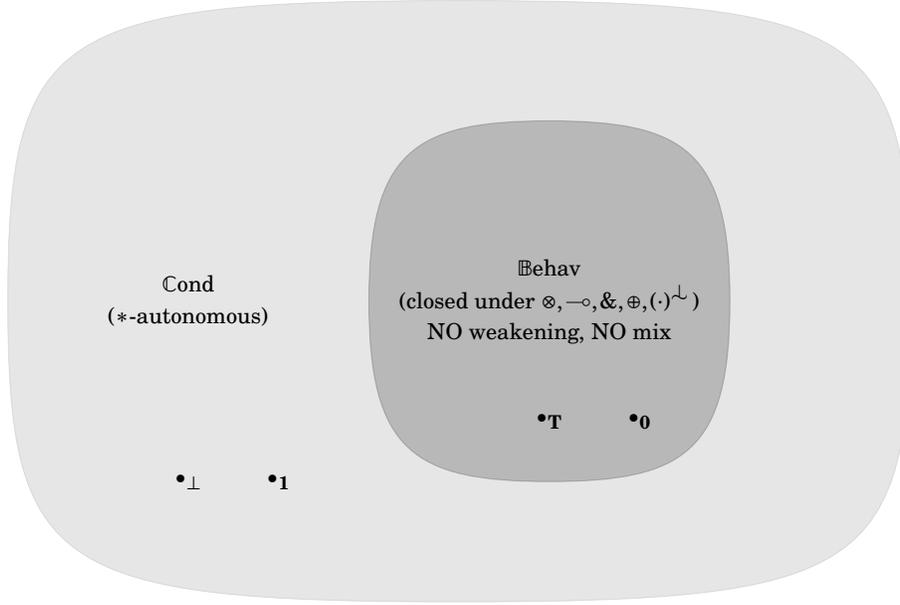

Moreover, all the results up to this point are independent of the choice of the parameter (the function $m$ that measures cycles), hence we did not define one, but a whole lot of such models.

\section{Truth and Soundness}\label{truthsection}

In this section, we suppose that the measurement map $m$ satisfies $m(1)=\infty$.

\subsection{Truth}

The notion of a \emph{successful project} captures the kind of projects that are interpretations of proofs of the sequent calculus. Keeping in mind that the graphs are a generalisation of the set of axiom links of a proof net, the graph of a successful project should be a disjoint union of transpositions, that is a graph in which each vertex is the source (resp. the target) of at most one edge, and such that for each edge $e \in E(v,w)$, there exists $e^{\ast} \in E(w,v)$ and $\omega(e) = 1 = \omega(e^{\ast})$. Moreover, the wager of such a project should be equal to zero, since a non-zero wager marks the appearance of a cycle in an application.

\begin{definition}[Success]\label{successdef}
A project $\de{a}=(a,A)$ is \emph{successful} when $a=0$ and $A=\sum_{i\in I^{A}} A_{i}$, and, for all $i\in I^{A}$, the graph $A_{i}$ is a disjoint union of transpositions.
\end{definition}

\begin{definition}[Truth]
A conduct $\cond{A}$ is \emph{true} when it contains a successful project.
\end{definition}

\begin{proposition}[Consistency]
The conducts $\cond{A}$ and $\cond{A^{\pol}}$ cannot be simultaneously true.
\end{proposition}

\begin{proof}
Let $\de{a}=(0,A)$ and $\de{a'}=(0,A')$ be two successful projects on the same carrier $V^{A}$. Their interaction is measured by the sum $$\sum_{(i,j)\in I^{A}\times I^{A'}} \meas{A_{i},A'_{j}}$$
It is easy to show the only possible cases are $\sca{a}{a'}=0$ and $\sca{a}{a'}=\infty$. This shows that $\de{a}$ and $\de{a'}$ cannot be orthogonal.
\end{proof}

\begin{proposition}[Compositionality]\label{compositiontruthmall}
Let $\de{f}=(0,F)$ and $\de{a}=(0,A)$ be two successful projects in the conducts $\cond{A\multimap B}$ and $\cond{A}$ respectively. Then $\cond{B}$ is true. Moreover, if $\cond{B}\neq\cond{T}_{V^{B}}$, then $\de{f\deplug a}$ is a successful project.
\end{proposition}

\begin{proof}
Since the wagers of $\de{f}$ and $\de{a}$ are equal to zero, we get $$\sca{f}{a}=\meas{F,A}=\sum_{(i,j)\in I^{F}\times I^{A}} \meas{F_{i},A_{j}}$$ Since each of the terms $\meas{F_{i},A_{j}}$ are either null or equal to $\infty$, we deduce that $\sca{f}{a}$ is either null or equal to $\infty$. 
\begin{itemize}
\item Suppose that $\sca{f}{a}=\infty$. Since $\de{f\deplug a}$ is, by definition of $\cond{A\multimap B}$, a project in $\cond{B}$ whose wager is equal to $\sca{f}{a}$, we have that $\cond{B}$ contains a project with infinite wager, and therefore $\cond{B}=\cond{T}_{V^{B}}$ by Proposition \ref{miseinfcondpleine}. As a consequence, $\cond{B}$ is true since it contains the successful project $\de{0}_{V^{B}}$.
\item Suppose now that $\sca{f}{a}=0$. In this case, the same reasoning we used to prove Theorem 44 in \cite{seiller-goim}, applied to each pairing of slices $(i,j)\in I^{F}\times I^{A}$, shows that $F_{i}\plug A_{j}$ is a disjoint union of transpositions. We then conclude that 
$$\de{f\deplug a}=(0,\sum_{(i,j)\in I^{F}\times I^{A}} F_{i}\plug A_{j})$$
is a successful project in $\cond{B}$.
\end{itemize}
In particuler, if one supposes that $\cond{B}\neq\cond{T}_{V^{B}}$, one finds himself in the second case (i.e. $\sca{f}{a}=0$) since $\sca{f}{a}=\infty$ implies $\cond{B}=\cond{T}_{V^{B}}$.
\end{proof}

\subsection{Full Soundness}

\begin{definition}
We fix $\mathcal{V}=\{X_{i}(j)\}_{i,j\in\naturalN}$ a set of \emph{localised variables\footnotemark\addtocounter{footnote}{-1}}. For $i\in\naturalN$, the set $X_{i}=\{X_{i}(j)\}_{j\in\naturalN}$ will be called the \emph{variable name $X_{i}$}, and an element of $X_{i}$ will be referred to as a \emph{variable of name\footnote{The \emph{variable names} are the variables in the usual sense, while the notion of localised variable is close to the usual notion of \emph{occurence of a variable}.} $X_{i}$}. Moreover, we suppose that each variable name $X_{i}$ has an associated \emph{size} $n_{i}\in\naturalN$.
\end{definition}

For $i,j\in\naturalN$ we define the \emph{location} $\sharp X_{i}(j)$ of the localised variable $X_{i}(j)$ as the set $$\{(i,m)~|~ jn_{i}\leqslant m\leqslant (j+1)n_{i}-1\}$$

\begin{definition}[Formulas of locMALL$_{\cond{T,0}}$]
We inductively define the formulas of \emph{localised multiplicative additive linear logic} locMALL as well as their \emph{locations} as follows:
\begin{itemize}
\item A localised variable $X_{i}(j)$ of name $X_{i}$ is a formula whose location is defined as $\sharp X_{i}(j)$;
\item If $X_{i}(j)$ is a localised variable of name $X_{i}$, then $(X_{i}(j))^{\pol}$ is a formula whoe location is $\sharp X_{i}(j)$.
\item If $A,B$ are formulas and $X,Y$ are their locations and satisfy $X\cap Y=\emptyset$, then $A\otimes B$ (resp. $A\parr B$, resp. $A\with B$, resp. $A\oplus B$) is a formula whose location is $X\cup   Y$;
\item The constants $\cond{T}_{\sharp \Gamma}$ and $\cond{0}_{\sharp\Gamma}$ are formulas whose location is $\sharp\Gamma$.
\end{itemize}
If $A$ is a formula, we will denote by $\sharp A$ the location of $A$. We also define the sequents $\vdash \Gamma$ of locMALL when the formulas of $\Gamma$ have pairwise disjoint locations\footnote{This is a natural condition since the comma in (one-sided) sequents corresponds to a $\parr$.}.
\end{definition}

\begin{definition}[Formulas of MALL$_{\cond{T,0}}$]
The formulas of MALL$_{\cond{T,0}}$ are defined by the following grammar:
\begin{equation*}
F:=X_{i}~|~X_{i}^{\pol}~|~F\otimes F~|~F\parr F~|~F\with F~|~F\oplus F~|~\cond{0}~|~\cond{T}
\end{equation*}
where the $X_{i}$ are variable names.
\end{definition}

\begin{remark}
Notice that multiplicative units are not considered here. This is coherent with the idea of a \emph{purely linear} fragment, where weakening does not hold. Indeed, multiplicative units $\cond{1}$ and $\cond{\bot}$ are not purely linear as they can be defined as $\cond{\oc T}$ and $\cond{\wn 0}$ respectively. This translates in our models as the fact that multiplicative units are \emph{not} behaviours, but merely conducts. A more involved sequent calculus which includes multiplicative units can be introduced by considering with some notion of polarities. The amount of additional work needed would not be justified for merely including multiplicative units in our calculi, and such considerations will be introduced in forthcoming works when extending our models with exponential connectives.
\end{remark}

\begin{remark}
To any formula of locMALL$_{\cond{T,0}}$ there corresponds a unique formula of MALL$_{\cond{T,0}}$ obtained by replacing localised variables by their name, i.e. by applying the transformation $X_{i}(j)\mapsto X_{i}$ to each localised variable $X_{i}(j)$. Conversely, it is always possible to \emph{localise} a formula of MALL$_{\cond{T,0}}$ (though not in a unique way): if $e$ is an enumeration of occurences of variable names in $A$, one can define in a natural way a formula $A^{e}$ of locMALL$_{\cond{T,0}}$.
\end{remark}

\begin{definition}[Proofs of locMALL$_{\cond{T,0}}$]
A proof of locMALL$_{\cond{T,0}}$ is a derivation obtained from the sequent calculus rules shown in Figure \ref{locMALL}, and such that any localised variable $X_{i}(j)$ and any negation $(X_{i}(j))^{\pol}$ of a localised variable appears at most once in a premise-free rule (axiom or $\top$ rule).
\end{definition}

\begin{definition}[Proofs of MALL$_{\cond{T,0}}$]
A proof of MALL$_{\cond{T,0}}$ is a derivation obtained from the sequent calculus rules shown in Figure \ref{MALL}.
\end{definition}

\begin{remark}
To any proof of locMALL$_{\cond{T,0}}$ corresponds a unique proof of MALL$_{\cond{T,0}}$ obtained by replacing localised variables by their names. Conversely, if $e$ is an enumeration of the occurrences of variable names appearing in the axiom and $\top$ rules of a proof of MALL$_{\cond{T,0}}$ $\pi$, one one can extend this enumeration to the whole derivation tree and obtain in this way a proof of locMALL$_{\cond{T,0}}$ $\pi^{e}$.
\end{remark}

\begin{figure}
\begin{center}
$\begin{array}{cc}
\begin{minipage}{5cm}
\begin{prooftree}
\AxiomC{}
\RightLabel{\scriptsize{Ax $(j\neq j')$}}
\UnaryInfC{$\vdash X_{i}(j)^{\pol},X_{i}(j')$}
\end{prooftree}
\end{minipage}
&
\begin{minipage}{5cm}
\begin{prooftree}
\AxiomC{$\vdash A,\Delta$}
\AxiomC{$\vdash A^{\pol},\Gamma$}
\RightLabel{\scriptsize{Cut\footnotemark\addtocounter{footnote}{-1}}}
\BinaryInfC{$\vdash \Delta,\Gamma$}
\end{prooftree}
\end{minipage}
\\
\begin{minipage}{5,5cm}
\begin{prooftree}
\AxiomC{$\vdash A,\Delta$}
\AxiomC{$\vdash B,\Gamma$}
\RightLabel{\scriptsize{$\otimes\footnotemark\addtocounter{footnote}{-1}$}}
\BinaryInfC{$\vdash A\otimes B,\Delta,\Gamma$}
\end{prooftree}
\end{minipage}
&
\begin{minipage}{5cm}
\begin{prooftree}
\AxiomC{$\vdash A,B,\Gamma$}
\RightLabel{\scriptsize{$\parr$}}
\UnaryInfC{$\vdash A\parr B,\Gamma$}
\end{prooftree}
\end{minipage}
\\
\begin{minipage}{5cm}
\begin{prooftree}
\AxiomC{$\vdash A_{i},\Gamma$}
\RightLabel{\scriptsize{$\oplus^{i}$}}
\UnaryInfC{$\vdash A_{0}\oplus A_{1},\Gamma$}
\end{prooftree}
\end{minipage}
&
\begin{minipage}{5cm}
\begin{prooftree}
\AxiomC{$\vdash A,\Gamma$}
\AxiomC{$\vdash B,\Gamma$}
\RightLabel{\scriptsize{$\with$}}
\BinaryInfC{$\vdash A\with B,\Gamma$}
\end{prooftree}
\end{minipage}
\\
\begin{minipage}{5cm}
\begin{prooftree}
\AxiomC{}
\RightLabel{\scriptsize{$\top_{\sharp\Gamma}$}}
\UnaryInfC{$\vdash \top,\Gamma$}
\end{prooftree}
\end{minipage}
&
\begin{minipage}{5cm}
\centering
\begin{prooftree}
\AxiomC{}
\noLine
\UnaryInfC{No rule for $\cond{0}$.}
\end{prooftree}
\end{minipage}
\end{array}
$\end{center}
\caption{localised sequent calculus locMALL}\label{locMALL}
\end{figure}
\addtocounter{footnote}{1}\footnotetext{It is necessary that $(\sharp A\cup  \sharp\Delta)\cap (\sharp B\cup  \sharp\Gamma)=\emptyset$ in order to apply the $\otimes$ rule and that $\sharp\Delta\cap\sharp\Gamma=\emptyset$ in order to apply the cut rule. Similarly, one should have that $(\sharp A)\cap(\sharp B)=\emptyset$ in order to apply the $\with$ rule.}

\begin{figure}
\begin{center}
$\begin{array}{cc}
\begin{minipage}{5cm}
\begin{prooftree}
\AxiomC{}
\RightLabel{\scriptsize{Ax}}
\UnaryInfC{$\vdash X_{i}^{\pol},X_{i}$}
\end{prooftree}
\end{minipage}
&
\begin{minipage}{5cm}
\begin{prooftree}
\AxiomC{$\vdash A,\Delta$}
\AxiomC{$\vdash A^{\pol},\Gamma$}
\RightLabel{\scriptsize{Cut}}
\BinaryInfC{$\vdash \Delta,\Gamma$}
\end{prooftree}
\end{minipage}
\\
\begin{minipage}{5,5cm}
\begin{prooftree}
\AxiomC{$\vdash A,\Delta$}
\AxiomC{$\vdash B,\Gamma$}
\RightLabel{\scriptsize{$\otimes$}}
\BinaryInfC{$\vdash A\otimes B,\Delta,\Gamma$}
\end{prooftree}
\end{minipage}
&
\begin{minipage}{5cm}
\begin{prooftree}
\AxiomC{$\vdash A,B,\Gamma$}
\RightLabel{\scriptsize{$\parr$}}
\UnaryInfC{$\vdash A\parr B,\Gamma$}
\end{prooftree}
\end{minipage}
\\
\begin{minipage}{5cm}
\begin{prooftree}
\AxiomC{$\vdash A_{i},\Gamma$}
\RightLabel{\scriptsize{$\oplus^{i}$}}
\UnaryInfC{$\vdash A_{0}\oplus A_{1},\Gamma$}
\end{prooftree}
\end{minipage}
&
\begin{minipage}{5cm}
\begin{prooftree}
\AxiomC{$\vdash \Gamma,A$}
\AxiomC{$\vdash \Gamma,B$}
\RightLabel{\scriptsize{$\with$}}
\BinaryInfC{$\vdash \Gamma,A\with B$}
\end{prooftree}
\end{minipage}
\\
\begin{minipage}{5cm}
\begin{prooftree}
\AxiomC{}
\RightLabel{\scriptsize{$\top$}}
\UnaryInfC{$\vdash \top,\Gamma$}
\end{prooftree}
\end{minipage}
&
\begin{minipage}{5cm}
\centering
\begin{prooftree}
\AxiomC{}
\noLine
\UnaryInfC{No rule for $\cond{0}$.}
\end{prooftree}
\end{minipage}
\end{array}
$\end{center}
\caption{Sequent calculus MALL$_{\cond{T,0}}$}\label{MALL}
\end{figure}

\begin{definition}[Interpretations]
We define an \emph{interpretation basis} as a map $\Phi$ which associates to any variable name $X_{i}$ a behaviour with carrier $\{0,\dots,n_{i}-1\}$.
\end{definition}

\begin{definition}[Interpretation of locMALL$_{\cond{T,0}}$ formulas]
Let $\Phi$ be an interpretation basis. We define the interpretation $I_{\Phi}(F)$ along $\Phi$ of a formula $F$ inductively:
\begin{itemize}
\item If $F=X_{i}(j)$, then $I_{\Phi}(F)$ is the delocation (i.e. a behaviour) of $\Phi(X_{i})$ along the bijection $x\mapsto (i,jn_{i}+x)$;
\item If $F=(X_{i}(j))^{\pol}$, we define the behaviour $I_{\Phi}(F)=(I_{\Phi}(X_{i}(j)))^{\pol}$;
\item If $F=\cond{T}_{\sharp\Gamma}$ (resp. $F=\cond{0}_{\sharp\Gamma}$), we define $I_{\Phi}(F)$ as the behaviour $\cond{T}_{\sharp\Gamma}$ (resp. $\cond{0}_{\sharp\Gamma}$);
\item If $F=A\otimes B$, we define the behaviour $I_{\Phi}(F)=I_{\Phi}(A)\otimes I_{\Phi}(B)$;
\item If $F=A\parr B$, we define the behaviour $I_{\Phi}(F)=I_{\Phi}(A)\parr I_{\Phi}(B)$;
\item If $F=A\oplus B$, we define the behaviour $I_{\Phi}(F)=I_{\Phi}(A)\oplus I_{\Phi}(B)$;
\item If $F=A\with B$, we define the behaviour $I_{\Phi}(F)=I_{\Phi}(A)\with I_{\Phi}(B)$.
\end{itemize}
Moreover, a sequent $\vdash \Gamma$ will be interpreted as the $\parr$ of the formulas in $\Gamma$, which we will denote by $\bigparr \Gamma$.
\end{definition}

\begin{definition}[Interpretation of locMALL$_{\cond{T,0}}$ proofs]\label{interpretationpreuvesmall}
Let $\Phi$ be an interpretation basis. We define the interpretation $I_{\Phi}(\pi)$ of a proof $\pi$ of locMALL$_{\cond{T,0}}$ inductively:
\begin{itemize}
\item if $\pi$ consists only in an axiom rule $\vdash (X_{i}(j))^{\pol},X_{i}(j')$, we define $I_{\Phi}(\pi)$ as the project $\de{Fax}$ corresponding to the bijection $(i,jn_{i}+x)\mapsto (i,j'n_{i}+x)$;
\item if $\pi$ consists only in a $\cond{T}_{\sharp \Gamma}$ rule, we define $I_{\Phi}(\pi)=(0,(\sharp \Gamma,\emptyset))$;
\item if $\pi$ is obtained from $\pi'$ by a $\parr$ rule, then $I_{\Phi}(\pi)=I_{\Phi}(\pi')$;
\item if $\pi$ is obtained from $\pi_{1}$ and $\pi_{2}$ by a $\otimes$ rule, we define $I_{\Phi}(\pi)=I_{\Phi}(\pi_{1})\otimes I_{\Phi}(\pi')$;
\item if $\pi$ is obtained from $\pi'$ by a $\oplus^{i}$ rule introducing a formula whose location is $V$, we define $I_{\Phi}(\pi)=I_{\Phi}(\pi')\otimes\de{0}_{V}$;
\item if $\pi$ of conclusion $\vdash \Gamma, A_{0}\with A_{1}$ is obtained from $\pi_{0}$ and $\pi_{1}$, of respective conclusions $\vdash \Gamma, A_{0}$ and $\vdash \Gamma, A_{1}$, by a $\with$ rule, we define:
\begin{eqnarray*}
\psi_{i}&:& x\mapsto (x,i)~~~~(i=0,1)\\
\tilde{\psi_{i}}&=&((\psi_{i})\restr{\sharp\Gamma})^{-1}~~~~(i=0,1)\\
\dot{\psi_{i}}&=&((\psi_{i})\restr{\sharp A_{i}})^{-1}~~~~(i=0,1)
\end{eqnarray*}
The interpretation of $\pi$ is then defined as:
\begin{equation*}
I_{\Phi}(\pi)=\de{Distr}^{\tilde{\psi_{0}},\tilde{\psi_{1}}}_{\dot{\psi_{0}},\dot{\psi_{1}}}\deplug(\psi_{0}(I_{\Phi}(\pi_{0}))\otimes\de{0}_{\sharp A_{1}}+\psi_{1}(I_{\Phi}(\pi_{1}))\otimes\de{0}_{\sharp A_{0}})
\end{equation*}
\item if $\pi$ is obtained from $\pi_{1}$ and $\pi_{2}$ by a cut rule, we define
$$I_{\Phi}(\pi)=I_{\Phi}(\pi_{1})\deplug I_{\Phi}(\pi_{2})$$
\end{itemize}
\end{definition}

Figure \ref{interpretationwith} represent the different steps in the interpretation of the $\with$ rule. Notice that the result can be easily understood and is a quite natural definition: the shared context $\Gamma$ is superimposed on two slices, and these slices contain the two projects interpreting the premisses of the rule.

\begin{lemma}\label{kindofsemidistrib}
Let $\cond{A,B,C}$ be behaviours. Then $\cond{(A\multimap B)\oplus C\subset A\multimap (B\oplus C)}$.
\end{lemma}

\begin{proof}
It is equivalent to show the inclusion $\cond{A\otimes (B^{\pol}\with C^{\pol})\subset (A\otimes B^{\pol})\with C^{\pol}}$. Using the definition of $\with$, we get:
\begin{eqnarray*}
\cond{A\otimes (B^{\pol}\with C^{\pol})}&=&\cond{A\otimes ((B{\uparrow_{\mathnormal{C}}})^{\pol} \cap (C{\uparrow_{\mathnormal{B}}})^{\pol})}\\
&=&\{\de{a}\otimes \de{d}~|~\de{a}\in\cond{A},\de{d}\in\cond{((B{\uparrow_{\mathnormal{C}}})^{\pol} \cap (C{\uparrow_{\mathnormal{B}}})^{\pol})}\}^{\pol\pol}\\
\cond{(A\otimes B^{\pol})\with C^{\pol}}&=&\cond{(((A\otimes B^{\pol})^{\pol}){\uparrow_{\mathnormal{C}}})^{\pol}\cap (C{\uparrow_{\mathnormal{A,B}}})^{\pol}}\\
&=&\cond{((A\multimap B){\uparrow_{\mathnormal{C}}})^{\pol}\cap (C{\uparrow_{\mathnormal{A,B}}})^{\pol}}\\
&=&\cond{((A\multimap B){\uparrow_{\mathnormal{C}}}\cup C{\uparrow_{\mathnormal{A,B}}})^{\pol}}
\end{eqnarray*}
All which is left to do is to show that a project of the form $\de{a}\otimes\de{d}$, where $\de{a}$ is an element of $\cond{A}$ and $\de{d}\in\cond{((B{\uparrow_{\mathnormal{C}}})^{\pol} \cap (C{\uparrow_{\mathnormal{B}}})^{\pol})}$, is orthogonal to any project in $\cond{E}=\cond{(A\multimap B){\uparrow_{\mathnormal{C}}}\cup C{\uparrow_{\mathnormal{A,B}}}}$. Let $\de{e}$ be a project in $\cond{E}$. Then:
\begin{enumerate}
\item either $\de{e}\in\cond{C}{\uparrow_{\mathnormal{A,B}}}$, i.e. $\de{e}=\de{c}\otimes \de{0}_{V^{A}\cup V^{B}}$. Then $\sca{a\otimes d}{e}=\sca{d}{c\otimes 0_{\text{$V^{B}$}}}$. But, since $\de{d}\in\cond{(C{\uparrow_{\mathnormal{B}}})^{\pol})}$, this entails that $\sca{d}{c\otimes 0_{\text{$V^{B}$}}}\neq0,\infty$. Therefore $\de{e}\poll\de{a\otimes d}$.
\item or else $\de{e}\in\cond{(A\multimap B){\uparrow_{\mathnormal{C}}}}$, i.e. $\de{e}=\de{f}\otimes \de{0}_{V^{C}}$ with $\de{f}\in\cond{A\multimap B}$. Then: 
\begin{eqnarray*}
\sca{e}{a\otimes d}&=&\sca{f\otimes 0_{\text{$V^{C}$}}}{a\otimes d}\\
&=&\sca{f}{(a\otimes d)\deplug 0_{\text{$V^{C}$}}}\\
&=&\sca{f}{a\otimes (d\deplug 0_{\text{$V^{C}$}})}\\
&=&\sca{f\deplug a}{d\deplug 0_{\text{$V^{C}$}}}\\
&=&\sca{(f\deplug a)\otimes 0_{\text{$V^{C}$}}}{d}
\end{eqnarray*}
But, since $\de{f\deplug a}\in\cond{B}$, we have $\de{(f\deplug a)\otimes 0_{\text{$V^{C}$}}}\in\cond{B}{\uparrow_{\mathnormal{C}}}$. The project $\de{d}$ is by definition in $(\cond{B}{\uparrow_{\mathnormal{C}}})^{\pol}$, and finally we obtain $\de{e}\poll\de{a\otimes d}$.\qedhere
\end{enumerate}
\end{proof}

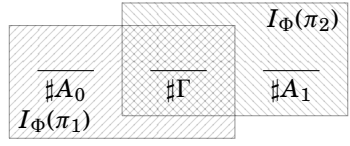
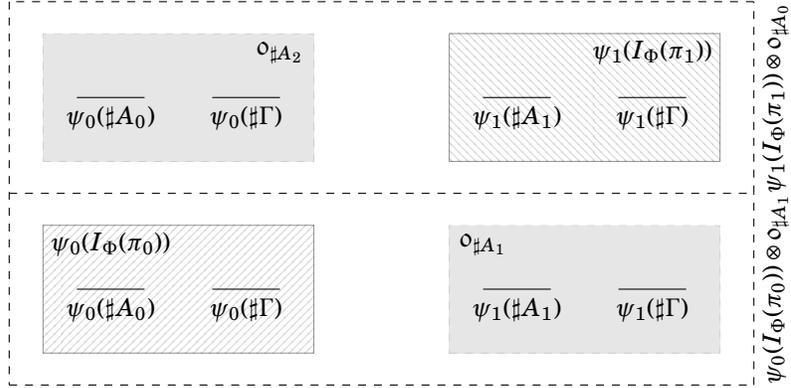
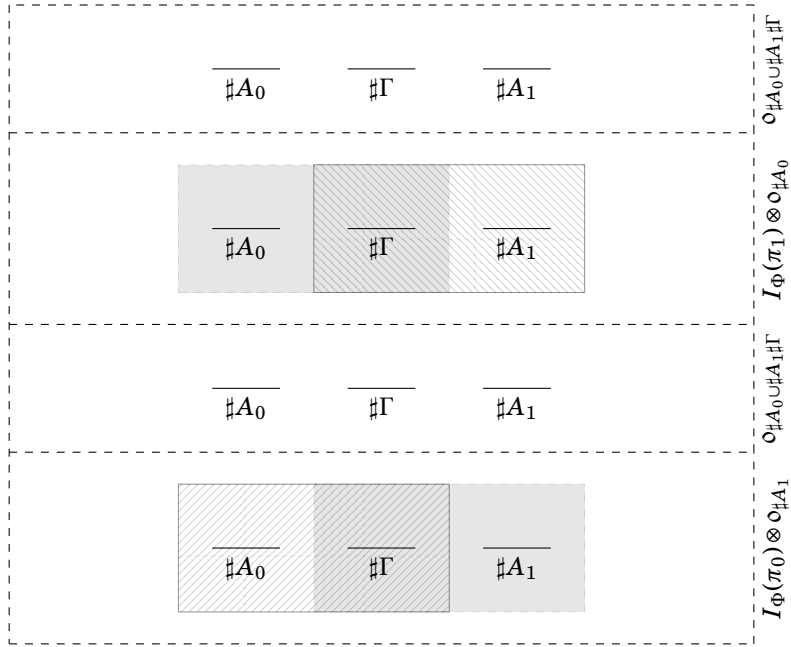
\begin{figure}
\centering
\subfigure[Interpretations of $\pi_{0}$ and $\pi_{1}$]{
\begin{tikzpicture}[x=0.75cm,y=0.75cm]

	\draw (0,0,0) -- (1,0,0) node [midway,below] {$\sharp A_{0}$};
	\draw (2,0,0) -- (3,0,0) node [midway,below] {$\sharp\Gamma$};
	\draw (4,0,0) -- (5,0,0) node [midway,below] {$\sharp A_{1}$};
	
	\draw[pattern color=black,pattern=north east lines,opacity=0.4] 
	(-0.5,-1.2,0) -- 
	(3.5,-1.2,0)  -- 
	(3.5,0.8,0) -- 
	(-0.5,0.8,0) -- (-0.5,-1.2,0) {};

	\draw[pattern color=black,pattern=north west lines,opacity=0.4] 
	(1.5,-0.8,0) -- 
	(5.5,-0.8,0) --
	(5.5,1.2,0) -- 
	(1.5,1.2,0) -- (1.5,-0.8,0) {};
	
	\node (P11) at (0.3,-0.9,0) {$I_{\Phi}(\pi_{1})$};
	\node (P21) at (4.7,0.9,0) {$I_{\Phi}(\pi_{2})$};
\end{tikzpicture}}

\subfigure[Summation of the delocations]{
\begin{tikzpicture}[x=0.9cm,y=0.85cm]

	\draw (-2,0) -- (-1,0) node [midway,below] {$\psi_{0}(\sharp A_{0})$};
	\draw (0,0) -- (1,0) node [midway,below] {$\psi_{0}(\sharp\Gamma)$};
	\draw (4,0) -- (5,0) node [midway,below] {$\psi_{1}(\sharp A_{1})$};
	\draw (6,0) -- (7,0) node [midway,below] {$\psi_{1}(\sharp\Gamma)$};
	\draw (-2,3) -- (-1,3) node [midway,below] {$\psi_{0}(\sharp A_{0})$};
	\draw (0,3) -- (1,3) node [midway,below] {$\psi_{0}(\sharp\Gamma)$};
	\draw (4,3) -- (5,3) node [midway,below] {$\psi_{1}(\sharp A_{1})$};
	\draw (6,3) -- (7,3) node [midway,below] {$\psi_{1}(\sharp\Gamma)$};
	
	\draw[dashed] (-3,-1.5) -- (8,-1.5) {};
	\draw[dashed] (8,-1.5) -- (8,1.5) node [sloped,midway,below] {$\psi_{0}(I_{\Phi}(\pi_{0}))\otimes\de{0}_{\sharp A_{1}}$};
	\draw[dashed] (8,1.5) -- (8,4.5) node [sloped,midway,below] {$\psi_{1}(I_{\Phi}(\pi_{1}))\otimes\de{0}_{\sharp A_{0}}$};
	\draw[dashed] (8,4.5) -- (-3,4.5) {};
	\draw[dashed] (-3,4.5) -- (-3,-1.5) node [sloped,midway,below,white] {$(I)$};
	\draw[dashed] (-3,1.5) -- (8,1.5) {};

	\draw[pattern=north east lines,pattern color=black,opacity=0.4] 
	(-2.5,-1) -- 
	(1.5,-1) --
	(1.5,1) --
	(-2.5,1) -- (-2.5,-1) {};

	\draw[dashed,fill=black,opacity=0.1] 
	(3.5,-1) -- 
	(7.5,-1) --
	(7.5,1) --
	(3.5,1) -- (3.5,-1) {};

	\draw[fill=black,dashed,opacity=0.1] 
	(-2.5,2) -- 
	(1.5,2) --
	(1.5,4) --
	(-2.5,4) -- (-2.5,2) {};

	\draw[pattern color=black,pattern=north west lines,opacity=0.4] 
	(3.5,2) -- 
	(7.5,2) --
	(7.5,4) --
	(3.5,4) -- (3.5,2) {};

	\node (P12) at (-1.5,0.7)  {$\psi_{0}(I_{\Phi}(\pi_{0}))$};
	\node (Q12) at (4,0.7) {$\de{0}_{\sharp A_{1}}$};
	\node (P22) at (6.5,3.7) {$\psi_{1}(I_{\Phi}(\pi_{1}))$};
	\node (Q12) at (1,3.7) {$\de{0}_{\sharp A_{2}}$};

\end{tikzpicture}}

\subfigure[Interpretation of $\pi$]{
\begin{tikzpicture}[x=0.9cm,y=0.85cm]
	\draw (0,0) -- (1,0) node [midway,below] {$\sharp A_{0}$};
	\draw (2,0) -- (3,0) node [midway,below] {$\sharp\Gamma$};
	\draw (4,0) -- (5,0) node [midway,below] {$\sharp A_{1}$};
	\draw (4,2.5) -- (5,2.5) node [midway,below] {$\sharp A_{1}$};
	\draw (2,2.5) -- (3,2.5) node [midway,below] {$\sharp\Gamma$};
	\draw (0,2.5) -- (1,2.5) node [midway,below] {$\sharp A_{0}$};	
	\draw (0,5) -- (1,5) node [midway,below] {$\sharp A_{0}$};
	\draw (2,5) -- (3,5) node [midway,below] {$\sharp\Gamma$};
	\draw (4,5) -- (5,5) node [midway,below] {$\sharp A_{1}$};
	\draw (4,7.5) -- (5,7.5) node [midway,below] {$\sharp A_{1}$};
	\draw (2,7.5) -- (3,7.5) node [midway,below] {$\sharp\Gamma$};
	\draw (0,7.5) -- (1,7.5) node [midway,below] {$\sharp A_{0}$};	

	\draw[dashed] (-3,-1.5) -- (8,-1.5) {};
	\draw[dashed] (8,-1.5) -- (8,1.5) node [sloped,midway,below] {$I_{\Phi}(\pi_{0})\otimes\de{0}_{\sharp A_{1}}$};
	\draw[dashed] (8,1.5) -- (8,3.5) node [sloped,midway,below] {$\de{0}_{\sharp A_{0}\cup\sharp A_{1}\sharp \Gamma}$};
	\draw[dashed] (8,3.5) -- (8,6.5) node [sloped,midway,below] {$I_{\Phi}(\pi_{1})\otimes\de{0}_{\sharp A_{0}}$};
	\draw[dashed] (8,6.5) -- (8,8.5) node [sloped,midway,below] {$\de{0}_{\sharp A_{0}\cup\sharp A_{1}\sharp \Gamma}$};

	\draw[dashed] (8,8.5) -- (-3,8.5) {};
	\draw[dashed] (-3,8.5) -- (-3,-1.5) node [sloped,midway,below,white] {$(I)$};
	\draw[dashed] (-3,1.5) -- (8,1.5) {};
	\draw[dashed] (-3,3.5) -- (8,3.5) {};
	\draw[dashed] (-3,6.5) -- (8,6.5) {};

	\draw[pattern color=black,pattern=north east lines,opacity=0.4] 
	(-0.5,-1) -- 
	(3.5,-1) --
	(3.5,1) --
	(-0.5,1) -- (-0.5,-1) {};

	\draw[dashed,fill=black,opacity=0.1] 
	(1.5,-1) -- 
	(5.5,-1) --
	(5.5,1) --
	(1.5,1) -- (1.5,-1) {};

	\draw[fill=black,dashed,opacity=0.1] 
	(-0.5,4) -- 
	(3.5,4) --
	(3.5,6) --
	(-0.5,6) -- (-0.5,4) {};

	\draw[pattern color=black,pattern=north west lines,opacity=0.4] 
	(1.5,4) -- 
	(5.5,4) --
	(5.5,6) --
	(1.5,6) -- (1.5,4) {};

\end{tikzpicture}}
\caption{Interpretation of the $\with$ rule (of conclusion $\vdash \Gamma, A_{0}\with A_{1}$) applied to the proofs $\pi_{0}$ and $\pi_{1}$ of respective conclusions $\vdash \Gamma, A_{0}$ and $\vdash \Gamma, A_{1}$}\label{interpretationwith}
\end{figure}	

\begin{proposition}[localised soundness]\label{locsoundmall}
Let $\Phi$ be an interpretation basis. If $\pi$ is a proof of conclusion $\vdash \Delta$, then $I_{\Phi}(\pi)$ is a successful project in the behaviour $I_{\Phi}(\vdash\Delta)$.
\end{proposition}

\begin{proof}
We show this result by induction on the last rule in $\pi$. By definition, the interpretation of the axiom rule introducing $\vdash (X_{i}(j))^{\pol},X_{i}(j')$ is a successful project in $I_{\Phi}(X_{i}(j))\multimap I_{\Phi}(X_{i}(j'))$ which is equal to $I_{\Phi}((X_{i}(j))^{\pol}\parr X_{i}(j'))$.
Then:
\begin{itemize}
\item if $\pi$ consists only in the rule $\cond{T}_{\sharp\Gamma}$, then $I_{\Phi}(\pi)=(0,\de{0}_{\sharp\Gamma})$ is successful and an element of $\cond{T}_{\sharp\Gamma}$;
\item the cases of multiplicative connectives are dealt with as in our former paper (\cite{seiller-goim}, Proposition 54);
\item if the last rule is a $\oplus$ rule — we suppose without loss of generality that it is a $\oplus^{1}$ rule:

\begin{prooftree}
\AxiomC{$\vdots{}^{\pi'}$}
\noLine
\UnaryInfC{$\vdash \Gamma, A_{1}$}
\RightLabel{\scriptsize{$\oplus^{1}$}}
\UnaryInfC{$\vdash \Gamma, A_{1}\oplus A_{2}$}
\end{prooftree}
Then $I_{\Phi}(\pi)=I_{\Phi}(\pi')\otimes \de{0}_{V}$, and $I_{\Phi}(\vdash \Gamma, A_{1}\oplus A_{2})=(\bigparr \Gamma)\parr(A_{1}\oplus A_{2})$. We now use the fact that $\cond{(A\multimap B)\oplus C\subset A\multimap (B\oplus C)}$ (Lemma \ref{kindofsemidistrib}) to show the inclusion $I_{\Phi}(\vdash \Gamma,A_{1})\oplus I^{\Phi}(A_{2})\subset I_{\Phi}(\vdash \Gamma, A_{1}\oplus A_{2})$. But, since $I_{\Phi}(\pi')$ is a successful project in $I_{\Phi}(\vdash \Gamma,A_{1})$, $I_{\Phi}(\pi)$ is a successful project in $I_{\Phi}(\vdash \Gamma,A_{1})\oplus I^{\Phi}(A_{2})$. As a consequence, it is a successful project in $I_{\Phi}(\vdash \Gamma, A_{1}\oplus A_{2})$;
\item if the last rule is a $\with$ rule:

\begin{prooftree}
\AxiomC{$\vdots{}^{\pi_{0}}$}
\noLine
\UnaryInfC{$\vdash \Gamma,A_{0}$}
\AxiomC{$\vdots{}^{\pi_{1}}$}
\noLine
\UnaryInfC{$\vdash \Gamma,A_{1}$}
\RightLabel{\scriptsize{$\with$}}
\BinaryInfC{$\vdash \Gamma, A_{0}\with A_{1}$}
\end{prooftree}
Using the notations of Définition \ref{interpretationpreuvesmall}, we have:
\begin{equation*}
I_{\Phi}(\pi)=\de{Distr}^{\tilde{\psi_{0}},\tilde{\psi_{1}}}_{\dot{\psi_{0}},\dot{\psi_{1}}}\deplug(\psi_{0}(I_{\Phi}(\pi_{0}))\otimes\de{0}_{\sharp A_{1}}+\psi_{1}(I_{\Phi}(\pi_{1}))\otimes\de{0}_{\sharp A_{0}})
\end{equation*}
By definition, the interpretations $I_{\Phi}(\pi_{i})$ are successful projects in the behaviours $I_{\Phi}(\vdash \Gamma, A_{i})$. We deduce from this that the projects $\psi_{i}(I_{\Phi}(\pi_{i}))$ are successful in the behaviours $\psi_{i}(I_{\Phi}(\vdash \Gamma, A_{i}))$ (delocations obviously preserve success). Since $\cond{A+B}\subset\cond{A\with B}$ when $\cond{A,B}$ are non-empty behaviours, we have\footnote{Indeed, the interpretations of the sequents $\vdash \Gamma,A_{0}$ and $\vdash \Gamma,A_{1}$ are non-empty by construction.} that $$\psi_{0}(I_{\Phi}(\pi_{0}))\otimes\de{0}_{\sharp A_{1}}+\psi_{1}(I_{\Phi}(\pi_{1}))\otimes\de{0}_{\sharp A_{0}}$$ is a successful project in $$\psi_{0}(I_{\Phi}(\vdash \Gamma, A_{0}))\with\psi_{1}(I_{\Phi}(\vdash \Gamma, A_{1}))$$
Since delocations are implemented by successful projects and the project implementing distributivity is defined as the sum of two delocations, $I_{\Phi}(\pi)$ is a successful project. Moreover, it is an element of the interpretation $I_{\Phi}(\vdash \Gamma,A_{0}\with A_{1})$ of $\vdash \Gamma,A_{0}\with A_{1}$ by Proposition \ref{distributivity}.
\item if $\pi$ is obtained by a cut rule between $\pi_{1}$ and $\pi_{2}$, of respective conclusions $\vdash A,\Gamma_{1}$ and $\vdash A^{\pol},\Gamma_{2}$, then Theorem \ref{compositiontruthmall} ensures us\footnote{Notice that in some cases, Theorem \ref{compositiontruthmall} does not ensure that the resulting project is successful, since the wager can be infinite. We are however in a particular case, since the graph of the project interpreting the $\cond{T}_{\sharp\Gamma}$ rule is empty, it cannot produce cycles. We are thus necessarily in the case of a null wager, i.e. the case where the produced project is successful.\label{footnoteadeq}} that $I_{\Phi}(\pi_{1})\deplug I_{\Phi}(\pi_{2})$ is a successful project in $\bigparr \Gamma$.
\end{itemize}
\end{proof}

Following the remarks we made earlier concerning the translation from the non-localised system to the localised one, we can chose an enumeration of the variable names in a proof $\pi$ of MALL$_{\cond{T,0}}$: we thus obtain a proof $\pi^{e}$ of locMALL$_{\cond{T,0}}$ whose conclusion is $A^{e}$. The following theorem is then a simple consequence of the preceding one.

\begin{theorem}[Full soundness for MALL$_{\cond{T,0}}$]
Let $\Phi$ be an interpretation basis, $\pi$ a proof of MALL$_{\cond{T,0}}$ of conclusion $\vdash \Gamma$, and $e$ an enumeration of the occurrences of the variable names in axiom and $\top$ rules in $\pi$. Then $I_{\Phi}(\pi^{e})$ is a successful project in $I_{\Phi}(\vdash \Gamma^{e})$.
\end{theorem}

\section{Graphs and Operators}\label{embedsection}

In this section, we study two particular values of the parameter $m$. The first, $m(x)=-\log(1-x)$, will give us a combinatorial version of Girard's Geometry of Interaction in the Hyperfinite Factor (GoI5); the second, $m(x)=\infty$, will give us a refined version of (the multiplicative fragment) of more ancient versions of GoI \cite{goi1,goi2,goi3}. Finally, we will relate our model to Girard's first GoI construction with permutations \cite{multiplicatives}.

Let $\hil{H}$ be a separable infinite-dimensional Hilbert space (for instance, the space $l^{2}(\mathbb{N})$ of square-summable sequences), and let $\{e_{i}\}_{i\in\mathbb{N}}$ be a base of $\hil{H}$. For every finite subset $S\subset\mathbb{N}$ there is a projection on the subspace generated by $\{e_{s}~|~s\in S\}$ that we will denote by $p_{S}$.

\subsection{localised Adjacency Matrices and the Contraction Property}

\begin{definition}
From a directed weighted graph $G$, we can define a simple graph $\what{G}$ with weights in $\mathbb{R}_{>0}\cup\{\infty\}$:
\begin{eqnarray*}
V^{\what{G}}&=&V^{G}\\
E^{\what{G}}&=&\{(v,w)~|~\exists e\in E^{G}, s^{G}(e)=v, t^{G}(e)=w\}\\
\omega^{\what{G}}&:&(v,w)\mapsto\sum_{e\in E^{G}(v,w)}\omega^{G}(e)
\end{eqnarray*}
If the weights of $\what{G}$ are in $\mathbb{R}_{>0}$, we will say it is total.
\end{definition}

\begin{definition}[localised weight matrix]
If $G$ is a weighted graph, the weight matrix of $\what{G}$ defines an operator in $p_{V_{G}}\mathcal{B}(\hil{H})p_{V_{G}}$ (hence in $\mathcal{B}(\hil{H})$). We will make an abuse of notation and denote this operator, the \emph{localised weight matrix of $G$}, by $\mat{G}$. 
\end{definition}

\begin{lemma}[Contraction Lemma for $m(x)=-\log(1-x)$]
Let $m: ]0,1]\rightarrow\realN\cup\{\infty\}$ be defined as $m(x)=-\log(1-x)$. Then, for any graphs $F,G$:
$$\scalar{F,G}_{m}=\scalar{\what{F},\what{G}}$$
\end{lemma}

\begin{proof}
The proof is quite involved and can be found in our earlier paper \cite{seiller-goim}.
\end{proof}

\begin{lemma}[Contraction Lemma for $m(x)=\infty$]
Let $m: ]0,1]\rightarrow\realN\cup\{\infty\}$ be defined as $m(x)=-\infty$. Then, for any graphs $F,G$:
$$\scalar{F,G}_{m}=\scalar{\what{F},\what{G}}$$
\end{lemma}

\begin{proof}
We have either $\scalar{F,G}=0$ if there are no alternating $1$-circuits in $F\bicol G$, or $\scalar{F,G}=\infty$ if there is at least one such $1$-circuit. 

In the first case, the replacement of $F$ by $\what{F}$ does not create any $1$-circuit. Indeed, if such a $1$-circuit existed in $\what{F}\bicol G$, for instance $\bar{\pi}=\bar{f}_{0}g_{0}\dots\bar{f}_{k} g_{k}$, then for each edge $\bar{f}_{i}$ in $E^{\what{F}}$ one can chose an edge $f_{i}$ in $E^{F}$ with same source and target as $\bar{f}_{i}$. Then $\pi=f_{0}g_{0}\dots f_{k}g_{k}$ is a $1$-circuit in $F\bicol G$, contradicting the fact that $\scalar{F,G}=0$.

In the second case, that is when there exists at least one $1$-circuit in $F\bicol G$, then there exists at least one $1$-circuit in $\what{F}\bicol G$. To see that, let us denote by $\pi=f_{0}g_{0}\dots f_{k}g_{k}$ a $1$-circuit in $F\bicol G$. Then the edges $f_{i}$ are replaced by edges with same source and target in $\what{F}$, which we will denote $\bar{f_{i}}$. Then there is a $1$-circuit in $\what{F}\bicol G$, namely the circuit $\bar{\pi}=\bar{f_{0}}g_{0}\dots \bar{f_{k}}g_{k}$.

We have just shown that $\scalar{F,G}=\scalar{\what{F},G}$. By symmetry of $\scalar{\cdot,\cdot}$ and using this result on $G$, we obtain $\scalar{F,G}=\scalar{\what{F},\what{G}}$.
\end{proof}

\subsection{The Feedback Equation}

From the very beginning, the geometry of interaction construction has been related with the computation of paths in proof nets \cite{PnHilb,pathslambda}. As we will see, the operation of execution between graphs captures exactly the corresponding operation in the setting of operators. To show this result, we will first recall the \emph{feedback equation}, and define execution between operators.

The feedback equation \cite{feedback} is the operator-theoretic counterpart of the cut-elimination procedure. A solution of a feedback equation corresponds to the normal form of a proof net containing a cut. Let $u,v$ be operators acting on the Hilbert spaces $\hil{H}\oplus\hil{H}'$ and $\hil{H}'\oplus\hil{H''}$ respectively. A \emph{solution to the feedback equation involving $u$ and $v$} is an operator $w$ acting on the Hilbert space $\hil{H}\oplus\hil{H}''$ such that $w(x\oplus z)=x'\oplus z'$ when there exists $y,y'\in\hil{H}'$ such that:
\begin{eqnarray*}
u(x\oplus y)&=&x'\oplus y'\\
v(y'\oplus z)&=&y\oplus z'
\end{eqnarray*}
This equation is usually illustrated as in Figure \ref{feedbacktrace}. Figure \ref{illustrationfeedback} illustrates the feedback equation with proof nets: if $u$ and $v$ represent two proof nets (shown with two conclusions only in order to simplify the exposition), then cutting the two proof nets corresponds to adding the equation $y=w'$ and $y'=w$ to the picture. The cut-free proof net obtained from the cut-elimination procedure would then be represented by an operator $w$ such that $w(x\oplus z)=x'\oplus z'$ when there exist $y,y',w,w'$ satisfying:
\begin{eqnarray*}
u(x\oplus y)&=&x'\oplus y'\\
v(w\oplus z)&=&w'\oplus z'\\
w&=&y'\\
w'&=&y
\end{eqnarray*}
This is  the same as saying that the cut-free proof net obtained is represented by a solution to the feedback equation involving $u$ and $v$.

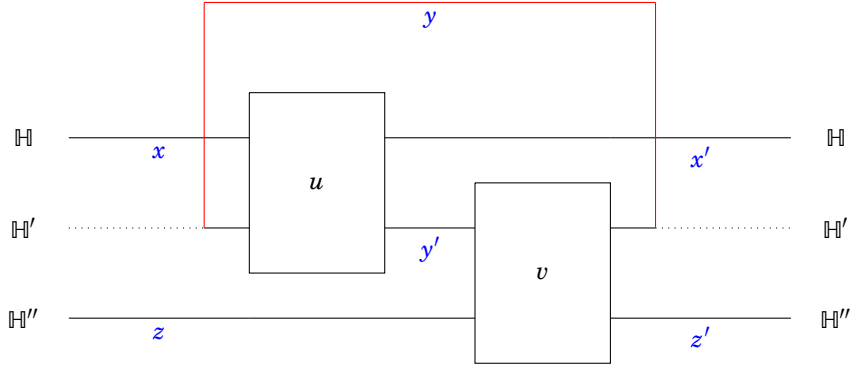
\begin{figure}
\begin{center}
\begin{tikzpicture}[x=0.6cm,y=0.6cm]
	\draw (-2,0) -- (2,0) node [midway,below,blue] {$x$};
		\node (Hi) at (-3,0) {$\hil{H}$};
	\draw (1,-2) -- (2,-2) {};
	\draw[dotted] (-2,-2) -- (1,-2) {};
		\node (HHi) at (-3,-2) {$\hil{H'}$};
	\draw (-2,-4) -- (2,-4) node [midway,below,blue] {$z$};
	\draw (2,-4) -- (7,-4) {};
		\node (HHHi) at (-3,-4) {$\hil{H''}$};
	\draw (2,1) -- (2,-3) -- (5,-3) -- (5,1) -- (2,1);
		\node (U) at (3.5,-1) {$u$};
	\draw (7,-1) -- (7,-5) -- (10,-5) -- (10,-1) -- (7,-1);
		\node (V) at (8.5,-3) {$v$};
	\draw (5,0) -- (10,0) {};
	\draw (10,0) -- (14,0) node [midway,below,blue] {$x'$};
		\node (Ho) at (15,0) {$\hil{H}$};
	\draw (5,-2) -- (7,-2) node [midway,below,blue] {$y'$};
	\draw (10,-2) -- (11,-2) {};
	\draw[dotted] (11,-2) -- (14,-2) {};
		\node (HHo) at (15,-2) {$\hil{H'}$};
	\draw (10,-4) -- (14,-4) node [midway,below,blue] {$z'$};
		\node (HHHo) at (15,-4) {$\hil{H''}$};
	\draw[red] (11,-2) -- (11,3) {};
	\draw[red] (11,3) -- (1,3) node [midway,below,blue] {$y$};	
	\draw[red] (1,3) -- (1,-2) {};
\end{tikzpicture}
\end{center}
\caption{Illustration of the feedback equation}\label{feedbacktrace}
\end{figure}

\begin{figure}
\centering
\subfigure[Two proof nets represented by operators]{
\centering
\begin{tikzpicture}
	\draw[fill,opacity=0.2] (0,-1) .. controls (0,-0.4) and (0.4,0) .. (2,0) .. controls (3.6,0) and (4,-0.4) .. (4,-1) .. controls (4,-1.6) and (3.6,-2) .. (2,-2) .. controls (0.4,-2) and (0,-1.6) .. (0,-1) {};
	\draw[fill,opacity=0.2] (6,-1) .. controls (6,-0.4) and (6.4,0) .. (8,0) .. controls (9.6,0) and (10,-0.4) .. (10,-1) .. controls (10,-1.6) and (9.6,-2) .. (8,-2) .. controls (6.4,-2) and (6,-1.6) .. (6,-1) {};
	\draw (1,-1.9) -- (1,-3);
		\draw[<-,dashed] (0.9,-1.9) -- (0.9,-3) node [midway,left] {\scriptsize{$x$}};
		\draw[->,dashed] (1.1,-1.9) -- (1.1,-3) node [midway,right] {\scriptsize{$x'$}};
	\draw (3,-1.9) -- (3,-3);
		\draw[<-,dashed] (2.9,-1.9) -- (2.9,-3) node [midway,left] {\scriptsize{$y$}};
		\draw[->,dashed] (3.1,-1.9) -- (3.1,-3) node [midway,right] {\scriptsize{$y'$}};
	\draw (7,-1.9) -- (7,-3);
		\draw[<-,dashed] (6.9,-1.9) -- (6.9,-3) node [midway,left] {\scriptsize{$w$}};
		\draw[->,dashed] (7.1,-1.9) -- (7.1,-3) node [midway,right] {\scriptsize{$w'$}};
	\draw (9,-1.9) -- (9,-3);
		\draw[<-,dashed] (8.9,-1.9) -- (8.9,-3) node [midway,left] {\scriptsize{$z$}};
		\draw[->,dashed] (9.1,-1.9) -- (9.1,-3) node [midway,right] {\scriptsize{$z'$}};
	
	\node (A) at (2,-1) {$u$};
	\node (B) at (8,-1) {$v$};
\end{tikzpicture}}
\subfigure[Cutting the proof nets induces the equalities $y=w'$ and $y'=w$]{
\centering
\begin{tikzpicture}
	\draw[fill,opacity=0.2] (0,-1) .. controls (0,-0.4) and (0.4,0) .. (2,0) .. controls (3.6,0) and (4,-0.4) .. (4,-1) .. controls (4,-1.6) and (3.6,-2) .. (2,-2) .. controls (0.4,-2) and (0,-1.6) .. (0,-1) {};
	\draw[fill,opacity=0.2] (6,-1) .. controls (6,-0.4) and (6.4,0) .. (8,0) .. controls (9.6,0) and (10,-0.4) .. (10,-1) .. controls (10,-1.6) and (9.6,-2) .. (8,-2) .. controls (6.4,-2) and (6,-1.6) .. (6,-1) {};
	\draw (1,-1.9) -- (1,-3);
		\draw[<-,dashed] (0.9,-1.9) -- (0.9,-3) node [midway,left] {\scriptsize{$x$}};
		\draw[->,dashed] (1.1,-1.9) -- (1.1,-3) node [midway,right] {\scriptsize{$x'$}};
	\draw (3,-1.9) -- (3,-3);
		\draw[<-,dashed] (2.9,-1.9) -- (2.9,-3) node [midway,left] {\scriptsize{$y$}};
		\draw[-,dashed] (3.1,-1.9) -- (3.1,-3) node [midway,right] {\scriptsize{$y'$}};
	\draw (7,-1.9) -- (7,-3);
		\draw[<-,dashed] (6.9,-1.9) -- (6.9,-3) node [midway,left] {\scriptsize{$w$}};
		\draw[-,dashed] (7.1,-1.9) -- (7.1,-3) node [midway,right] {\scriptsize{$w'$}};
	\draw (9,-1.9) -- (9,-3);
		\draw[<-,dashed] (8.9,-1.9) -- (8.9,-3) node [midway,left] {\scriptsize{$z$}};
		\draw[->,dashed] (9.1,-1.9) -- (9.1,-3) node [midway,right] {\scriptsize{$z'$}};
	
	\draw (3,-3) .. controls (3,-3.8) and (3.2,-4) .. (5,-4) {};
		\draw[dashed,-] (3.1,-3) .. controls (3.1,-3.7) and (3.3,-3.9) .. (5,-3.9) .. controls (6.8,-3.9) and (6.9,-3.7) .. (6.9,-3);
		\draw[dashed,-] (2.9,-3) .. controls (2.9,-3.9) and (3.1,-4.1) .. (5,-4.1) .. controls (6.8,-4.1) and (7.1,-3.9) .. (7.1,-3) {};
	\draw (5,-4) .. controls (6.8,-4) and (7,-3.8) .. (7,-3) {};

	\node (cut) at (5,-4.3) {cut};
	
	\node (A) at (2,-1) {$u$};
	\node (B) at (8,-1) {$v$};
\end{tikzpicture}}
\caption{Illustration of the feedback equation}\label{illustrationfeedback}
\end{figure}
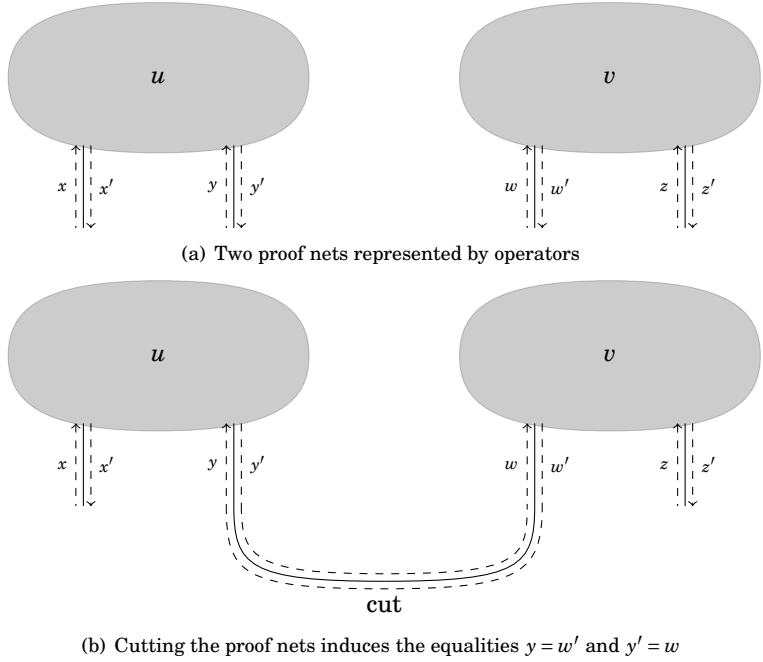

If we write $p,p',p''$ the projections onto the spaces $\hil{H,H',H''}$ respectively, the execution formula $u\plug v=(p+p''v)\left(\sum_{i\geqslant 0} (uv)^{i}\right)(up+p'')$, when it is defined, gives a solution to the feedback equation involving $u$ and $v$. More generally, the formula $(p+p''v)(1-uv)^{-1}(up+p'')$, when $1-uv$ is invertible, defines such a solution.

Girard studied, in his paper entitled \emph{Geometry of Interaction $\text{IV}$: the Feedback Equation} \cite{feedback}, an extension of this solution when the operator $1-uv$ is non-invertible. He then showed that for the couples of operators $(u,v)$ where $u,v$ are hermitians whose norm is at most $1$, the solution involving the inverse operator $(1-uv)^{-1}$ defines a sort of (partial) functional application, which can be extended to all couples of operators $(u,v)$ with $u,v$ hermitians whose norm is at most $1$. Moreover, this extension is the unique such extension preserving associativity and verifying some continuity properties.

In his first constructions, Girard defined the execution between two operators $u,v$ to be equal to $u\plug v=(p+p''v)\left(\sum_{i\geqslant 0} (uv)^{i}\right)(up+p'')$ when this expression was defined, i.e. when the product $uv$ was nilpotent. The fact that this definition was partial conveyed the (somewhat wrong) impression that orthogonality was introduced to deal with this problem: two operators were orthogonal when the execution between them was defined. In his latest construction, i.e. the geometry of interaction in the hyperfinite factor, Girard builds upon his general solution to the feedback equation a geometry of interaction in which execution is always defined. The notion of orthogonality, which uses the Fuglede-Kadison determinant\footnote{The Fuglede-Kadison determinant is a generalisation of the usual determinant of matrices that can be defined in any type $\text{II}_{1}$ factor.} is then in its rightful place: it does not ensure that execution is defined but gives information on how execution is performed.

\begin{definition}
Let $u,v$ be operators. We denote by $\ex{u,v}$ the solution, when it is defined, of the feedback equation involving $u$ and $v$.
\end{definition}

\subsection{Paths as a Solution to the Feedback Equation}

Girard's general solution to the feedback equation is actually restricted to the case where the two operators involved are hermitians of norm at most $1$. For this reason, the geometry of interaction in the hyperfinite factor only deals with hermitian operators of norm at most $1$. Although the restriction to hermitian operators does not seem crucial for his result, the condition of the norm is, on the other hand, essential in the proofs. We therefore consider the notion of \emph{operator graphs} whose weight matrix is of norm at most $1$, and the special case of \emph{symmetric operator graphs} whose matrix is moreover self-adjoint.

\begin{definition}[Operator Graphs]
An \emph{operator graph} is a graph $F$ such that $\norm{\mat{{F}}}\leqslant 1$. An operator graph is said to be \emph{symmetric} when $\mat{{F}}=\mat{{F}}^{\ast}$, i.e. when for all edges $e\in E^{\what{F}}(v,w)$ ($v,w\in V^{\what{F}}$), there is an edge $e^{\ast}\in E^{\what{F}}(w,v)$ such that $\omega^{\what{F}}(e)=\omega^{\what{F}}(e^{\ast})$.
\end{definition}

As it turns out, even though the general solution of the feedback equation is constructed with the use of unbounded operators, it can be constructed in a very simple manner when one restricts to operators acting on finite-dimensional Hilbert spaces. In our earlier paper \cite{seiller-goim}, we showed that, provided $F,G$ satisfy $\meas[-\log(1-x)]{F,G}\neq\infty$, execution between operator graphs $F,G$ corresponded to solving the feedback equation between the associated matrices of weights. Here, we will extend this result by showing it holds in full generality. Indeed, the matrix of weights of the graph $\what{F\plug H}$ can be shown to be equal to the matrix $\sum_{i=0}^{\infty} (p_{F}\what{F}+p_{G})(\what{F})^{i}(p_{F}+\what{G}p_{G})$. Since this series is always convergent when the operators $\what{F},\what{G}$ are of norm at most $1$, this gives the result.

\begin{proposition}
Let $u$ (resp. $v$) be a $n+m$ (resp. $m+k$) square matrix of norm at most $1$ seen as $n+m+k$ matrices:
$$u=\left(\begin{array}{ccc} u_{n,n} & u_{n,m} & 0\\
					u_{m,n} & u_{m,m} & 0\\
					0 & 0 & 0\end{array}\right)
		~~~~~~
v=\left(\begin{array}{ccc} 0 & 0 & 0\\
					0 & v_{m,m} & v_{k,m}\\
					0 & v_{m,k} & v_{k,k}\end{array}\right)$$
Let $p,q,r$ be the projections defined as the block matrices:
$$
p=\left(\begin{array}{ccc} \text{Id} & 0 & 0\\
					0 & 0 & 0\\
					0 & 0 & 0\end{array}\right)
		~~~~~~
q=\left(\begin{array}{ccc} 0 & 0 & 0\\
					0 & \text{Id} & 0\\
					0 & 0 & 0\end{array}\right)
		~~~~~~
r=\left(\begin{array}{ccc} 0 & 0 & 0\\
					0 & 0 & 0\\
					0 & 0 & \text{Id}\end{array}\right)
$$
Then the series:
$$\sum_{i\geqslant 0} (pu+r)(vu)^{i}(p+rv)$$
is convergent and defines a solution to the feedback equation involving $u$ and $v$.
\end{proposition}

\begin{proof}
Let us first compute the product $vu$.
$$vu=\left(\begin{array}{ccc}
		0 & 0 & 0\\
		v_{m,m}u_{m,n} & v_{m,m}u_{m,m} & 0\\
		v_{m,k}u_{m,n} & v_{m,k}u_{m,m} & 0
		\end{array}\right)$$

Let us write $q_{1}$ as the largest projection such that $q_{1}\leqslant q$ and $q_{1}v_{m,m}u_{m,m}q_{1}=q_{1}$ (this may be equal to $0$), and denote by $q_{<1}$ the projection $q-q_{1}$. Notice that $q_{1}$ is the projection on the subspace corresponding to the eigenvalue $1$ of $v_{m,m}u_{m,m}$. We can write $vu$ as a $4\times 4$ block matrix along the projections $p,q_{1},q_{<1},r$ as follows:

$$vu=\left(\begin{array}{cccc}
		0 & 0 & 0 & 0\\
		w_{2,1} & \text{Id} & w_{2,3} & 0\\
		w_{3,1} & w_{3,2} & w_{3,3} & 0\\
		w_{4,1} & w_{4,2} & w_{4,3} & 0
		\end{array}\right)$$
But since $\norm{vu}\leqslant 1$, we have that $w_{2,1}=w_{2,3}=w_{3,2}=w_{4,2}=0$:
$$vu=\left(\begin{array}{cccc}
		0 & 0 & 0 & 0\\
		0 & \text{Id} & 0 & 0\\
		w_{3,1} & 0 & w_{3,3} & 0\\
		w_{4,1} & 0 & w_{4,3} & 0
		\end{array}\right)$$
I.e. $vu=q_{1}+z$ where $(1-q_{1})z(1-q_{1})=z$. Using the fact that $q_{1}z=zq_{1}=0$ we obtain, for any integer $k>0$:
$$ (vu)^{k}=q_{1}+z^{k}$$
Hence:
$$\sum_{i\geqslant 0} (pu+r)(vu)^{i}(p+rv)=\sum_{i\geqslant 0} (pu+r)(q_{1}+z^{i})(p+rv)=\sum_{i\geqslant 0} (pu+r)z^{i}(p+rv)$$
Now, we can make $z$ into a triangular matrix by making $w_{3,3}$ triangular. It is clear that the non-null elements on the diagonal are strictly lesser than $1$, or else it would contradict the fact that $q_{1}$ was the largest projection satisfying $q_{1}\leqslant q$ and $q_{1}v_{m,m}u_{m,m}q_{1}=q_{1}$. This implies that $1-z$ is invertible since all its diagonal coefficients are different from zero. Hence $(pu+r)(1-z)^{i}(p+rv)=(pu+r)\sum_{i\geqslant 0} z^{i}(p+rv)$ is always defined and is equal to $\sum_{i\geqslant 0} (pu+r)z^{i}(p+rv)$.
\end{proof}

\begin{proposition}\label{execution}
Let $F$ and $G$ be operator graphs, $\mat{{F}}$ and $\mat{{G}}$ their localised weight matrices. Then:
$$\mat{{F\plug G}}=\ex{\mat{F},\mat{G}}$$
\end{proposition}

\begin{proof}
Let $p,q,r$ be the projections corresponding to the sets $V^{F}-V^{F}\cap V^{G}, V^{F}\cap V^{G}, V^{G}-V^{F}\cap V^{G}$ respectively. 

We first write $F\plug G$ as the union $\cup_{k\geqslant 1} C_{k}$ where $C_{k}$ is the graph of alternating paths in $F\bicol G$ with source and target in $V^{F}\Delta V^{G}$ and of length equal to $k$. Each $C_{k}$ can be written as the union $\cup_{i,j\in\{F,G}\} C_{k}^{i,j}$ where $C_{k}^{i,j}$ is the graph of paths in $C_{k}$ with source in $i$ and target in $j$. It is easily seen that:
\begin{eqnarray*}
C_{k}^{F,F}&=&\left\{\begin{array}{cl}p\mat{F}(\mat{G}\mat{F})^{\frac{k-1}{2}}p & \text{$k$ odd}\\ 0 &\text{otherwise}\end{array}\right.\\
C_{k}^{G,G}&=&\left\{\begin{array}{cl}q(\mat{G}\mat{F})^{\frac{k-1}{2}}\mat{G}q & \text{$k$ odd}\\ 0 &\text{otherwise}\end{array}\right.\\
C_{k}^{F,G}&=&\left\{\begin{array}{cl}q(\mat{G}\mat{F})^{\frac{k}{2}}p & \text{$k$ even}\\ 0 &\text{otherwise}\end{array}\right.\\
C_{k}^{G,F}&=&\left\{\begin{array}{cl}p\mat{F}(\mat{G}\mat{F})^{\frac{k-2}{2}}\mat{G}q & \text{$k>2$ even}\\  0 &\text{otherwise}\end{array}\right.
\end{eqnarray*}
Now, since $\mat{{F\plug G}}=\sum_{k>0} \sum_{i,j\in\{F,G\}} \mat{C_{k}^{i,j}}$, we obtain:
\begin{eqnarray*}
\mat{{F\plug G}}&=&\sum_{k>0\text{ even}} \left(q(\mat{G}\mat{F})^{\frac{k}{2}}p + p\mat{F}(\mat{G}\mat{F})^{\frac{k}{2}}\mat{G}q\right)\\
&& ~~~~~~ + \sum_{k>0\text{ odd}} \left(p\mat{F}(\mat{G}\mat{F})^{\frac{k-1}{2}}p + q(\mat{G}\mat{F})^{\frac{k-1}{2}}\mat{G}q\right)\\
&=&\sum_{k>0} \left(q(\mat{G}\mat{F})^{k}p + p\mat{F}(\mat{G}\mat{F})^{k}\mat{G}q\right)\\
&& ~~~~~~ + \sum_{k>0} \left(p\mat{F}(\mat{G}\mat{F})^{k}p + q(\mat{G}\mat{F})^{k}\mat{G}q\right)\\
&=&\sum_{k> 0} (p\mat{{F}}+r)(\mat{{G}}\mat{{F}})^{i}(p+r\mat{{G}})
\end{eqnarray*}
From the previous proposition, this series converges, and therefore defines a solution to the feedback equation involving $\mat{F}$ and $\mat{G}$. Hence, we have just shown that $\mat{{F\plug G}}=\ex{\mat{F},\mat{G}}$.
\end{proof}

\subsection{GoI in the Hyperfinite Factor}

We will now extend the previous results to sliced graphs and show how the choice of the measurement map $m(x)=-\log(1-x)$ defines a combinatorial version of Girard's hyperfinite GoI. However, before doing so, we will recall the important definitions of the latter, and discuss two constructions of the additive connectives.

Girard defines his latest geometry of interaction in the hyperfinite factor $\infhyp$ of type $\text{II}_{\infty}$ with a fixed trace $tr$. This von Neuman algebra can be obtained as the tensor product of $\mathcal{B}(\hil{H})$ with the hyperfinite factor of type $\text{II}_{1}$, usually denoted $\mathcal{R}$. We will therefore work with operators in $\mathcal{B}(\hil{H})\otimes\mathcal{R}$ and the trace defined as the tensor product of the normalized trace on $\mathcal{B}(\hil{H})$ (i.e. the trace of minimal projections is $1$) and the normalized trace on $\mathcal{R}$ (i.e. the trace of the identity is $1$). 

We first recall the notion of project used by Girard. We will refer to them as \emph{hyperfinite projects} to avoid a collapse of terminology.

\begin{definition}[Hyperfinite projects]
A \emph{hyperfinite project} is a $5$-tuple:
$$\de{a}=(p,a,\mathcal{A},\alpha,A)$$
consisting of:
\begin{itemize}
\item a finite projection $p^{\ast}=p^{2}=p\in\infhyp $, the \emph{carrier} of the project $\de{a}$;
\item a finite and hyperfinite von Neumann algebra $\mathcal{A}$, the \emph{idiom} of $\de{a}$;
\item a normal hermitian tracial form $\alpha$ on $\mathcal{A}$, the \emph{pseudo-trace} of $\de{a}$;
\item a real number $a\in\mathbb{R}\cup\{\alpha(1_{\mathcal{A}})\infty\}$, the \emph{wager} of $\de{a}$;
\item an hermitian $A\in (p\infhyp  p)\otimes\mathcal{A}$ such that $\norm{A}\leqslant 1$.
\end{itemize}
As in Girard's paper, we will denote such an object by $\de{a}=a\cdot +\cdot \alpha+A$.

We will also distinguish those hyperfinite projects satisfying $\alpha(1_{\vn{A}})\neq 0$, and call these \emph{strict hyperfinite projects}.
\end{definition}

Let $\de{a},\de{b}$ be two hyperfinite projects, and $A\in \infhyp\otimes\vn{A}$ and $B\in\infhyp\otimes\vn{B}$ be their associated operators. We will denote by $A^{\dagger_{\vn{B}}}$ and $B^{\ddagger_{\vn{A}}}$ the operators in $\infhyp\otimes\vn{A}\otimes\vn{B}$ defined by:
\begin{eqnarray*}
A^{\dagger_{\vn{B}}}&=&A\otimes 1_{\vn{B}}\\
B^{\ddagger_{\vn{A}}}&=&(\text{Id}_{\infhyp}\otimes\sigma)(B\otimes 1_{\vn{A}})
\end{eqnarray*}
where $\sigma$ is the natural isomorphism $\vn{B}\otimes \vn{A}\rightarrow \vn{A}\otimes \vn{B}$.

The geometry of interaction in the hyperfinite factor uses a generalisation of the usual determinant of matrices known as the \emph{Fuglede-Kadison determinant} \cite{FKdet}. Indeed, a type $\text{II}_{1}$ factor possesses a trace, and it is therefore possible to define a determinant by using the identity $\det(\exp(A))=\exp(\tr(A))$ which relates the determinant to the trace and which is satisfied for any $A$ acting on a Hilbert space of finite dimension.

The Fuglede-Kadison determinant is defined as follows: if $Tr$ denotes the normalized trace (i.e. $Tr(1)=1$), one defines, for any invertible operator $A$:
\begin{equation*}
\det{}^{FK}(A)=e^{Tr(\log(\abs{A}))}
\end{equation*}
This definition is then extended to non-invertible operators\footnote{This extension is not unique, but the results we show are independant of the chosen extension.}.

Since GoI5 deals with pseudo-traces, we need to extend this definition.
\begin{definition}
If $\alpha$ is a pseudo-trace on $\vn{A}$, and $\tr$ a trace on $\finhyp$, we define, for all invertible $A\in\finhyp\otimes\vn{A}$:
\begin{equation*}
\det{}^{FK}_{\tr\otimes\alpha}(A)=e^{\tr\otimes\alpha(\log(\abs{A}))}
\end{equation*}
We will abusively denote by $\det{}^{FK}_{\tr\otimes\alpha}$ any extension of this definition to $\finhyp\otimes\vn{A}$.
\end{definition}

\begin{definition}[Measurement between hyperfinite projects]
Let $\de{a}=a\cdot +\cdot \alpha+ A$ and $\de{b}=b\cdot+\cdot\beta+B$ be two hyperfinite projects. One defines the \emph{measurement} of the interaction between $\de{a}$ and $\de{b}$ as:
$$\sca[]{a}{b}=a\beta(1_{\vn{B}})+b\alpha(1_{\vn{A}})-\log(\det{}^{FK}_{\tr\otimes\alpha\otimes\beta}(1-A^{\dagger_{\vn{B}}}B^{\ddagger_{\vn{A}}}))$$
\end{definition}

\begin{definition}[Execution between hyperfinite projects]
Let $\de{a}=a\cdot +\cdot \alpha+ A$ and $\de{b}=b\cdot+\cdot\beta+B$ be two hyperfinite projects. One defines the \emph{execution} of $\de{a}$ and $\de{b}$, denoted by $\de{a\plug b}$, as the project $\sca[]{a}{b}\cdot+\cdot\alpha\otimes\beta+\ex{A^{\dagger_{\vn{B}}},B^{\ddagger_{\vn{A}}}}$.
\end{definition}

The definition of the tensor product of two hyperfinite projects is, as in the case of graphs, a particular case of execution: if $\de{a}$ and $\de{b}$ are hyperfinite projects of respective carriers $p$ and $q$ such that $pq=0$, then their tensor product is defined as $\de{a}\plug\de{b}$. Notice that in this configuration, the solution to the feedback equation is equal to $A^{\dagger_{\vn{B}}}+B^{\ddagger_{\vn{A}}}$.

\subsection{Additives in GoI5: the two versions}

The construction of additive connectives in Girard's paper differs slightly from the one that has led to our graph-theoretic constructions. Indeed, the approach chosen in this paper is inspired from a remark Girard made after his paper on hyperfinite GoI was published, namely that the restriction to strict hyperfinite projects (a restriction imposed in the first version of GoI5) can be omitted. This slight difference leads to a more natural approach to additive connectives: while in the initial construction one had to consider two kinds of conducts (namely positive and negative conducts), we can now work in a more uniform framework where one can consider only one restriction on conducts (\emph{dichologies} in the words of Girard, which correspond to our \emph{behaviours}) which is self-dual. We will present here the two different constructions in order to help the reader grasp the small differences arising from this change.

In both constructions, we will be using the following formal weighted sum of projects:
\begin{equation*}
\de{a+\lambda b}=(p,a,\vn{A},\alpha,A)+\mu(q,b,\vn{B},\beta,B)=(p\wedge q,a+\mu b,\vn{A\oplus B},\alpha\oplus\mu\beta,A\oplus B)
\end{equation*}
where $p\wedge q$ denotes the smallest projection $r$ such that $p\leqslant r$ and $q\leqslant r$.

We will also use the notion of extension: if $\de{a}=(p,a,\vn{A},\alpha,A)$ is a hyperfinite project and $q$ a projection such that $pq=0$, then one defines $\extde{a}{q}$ as the project $(p+q,a,\vn{A},\alpha,A)$. Notice that $\extde{a}{q}=\de{a}\otimes \de{0}_{q}$ with $\de{0}_{q}=(q,0,\complexN,1_{\complexN},0)$. As in Section \ref{gdisection}, we define $\extend{A}{q}=\{\extde{a}{q}~|~\de{a}\in\cond{A}\}$ for conducts $\cond{A}$ with carrier $p$ such that $pq=0$. In practice, the projection $q$ will be the carrier of a conduct $\cond{B}$, in which case we will denote $\extend{A}{q}$ (resp. $\extde{a}{q}$) by $\extend{A}{B}$ (resp. $\extde{a}{B}$).

\subsubsection{The initial construction: GoI5.1}

The initial framework dealt only with strict hyperfinite projects. Additives can, as in the graphs setting, be defined on conducts by:
$$\cond{A\oplus B}=((\extend{A}{B})^{\pol\pol}\cup(\extend{B}{A})^{\pol\pol})^{\pol\pol}$$
$$\cond{A\with B}=(\extend{(A^{\pol})}{B})^{\pol}\cap(\extend{(B^{\pol})}{A})^{\pol}$$
In order to define those constructions on projects, one needs to restrict the type of conducts considered. In particular, one wants the elements of the form $\de{a+b}$ ($\de{a}\in\cond{A}$, $\de{b}\in\cond{B}$) to be in the conduct $\cond{A\with B}$. As a consequence, one needs to restrict to conducts such that $a=0$ for all $\de{a}$. This leads to the definition of \emph{negative conducts}. This restriction is however not self-dual, and leads to the notion of \emph{positive conducts} which are orthogonals of negative conducts.

\begin{definition}[Negative Conducts]
A conduct $\cond{A}$ is \emph{negative} if:
\begin{itemize}
\item for all $\de{a}=a\cdot +\cdot\alpha+A \in\cond{A}$, we have $a=0$;
\item for all $\de{a}\in\cond{A}$ and all $\lambda$ such that $\lambda+\alpha(1_{\vn{A}})\neq 0$, the project $\de{a+\lambda 0}$ is in $\cond{A}$.
\end{itemize}
A negative conduct is \emph{proper} if it is non-empty.
\end{definition}

Dually, one defines the notion of positive conducts.
\begin{definition}[Positive Conducts]
A conduct $\cond{B}$ is \emph{positive} if:
\begin{itemize}
\item it contains all \emph{daimons}: for all $\lambda\neq 0$, $\de{Dai}_{\lambda}=\lambda\cdot+\cdot 1_{\complexN}+0\in\cond{B}$;
\item if $b\cdot+\cdot \beta+B\in\cond{B}$, with $b\neq 0$, then for all $c\neq 0$, $c\cdot+\cdot\beta+B\in\cond{B}$.
\end{itemize}
A positive conduct is \emph{proper} if it does not contain the project $\de{Dai}_{0}=0\cdot+\cdot 1_{\complexN}+0$.
\end{definition}

\noindent One then can define $\with$ only between negative conducts, and $\oplus$ on positive ones. 

\begin{definition}[Additives on polarized conducts]
Let $\cond{A},\cond{B}$ be negative conducts. Then $\cond{A\with B}$ is negative and $\cond{A^{\pol}\oplus B^{\pol}}$ is positive.
\end{definition}

One then needs to restrict to the proper cases, in order to obtain some results concerning the polarization of tensor products. Moreover, the restriction to positive/negative proper conducts is not completely satisfying since it yields in some cases unpolarized conducts, and the table of polarities is far from being what one would expect (i.e. it does not correspond to the usual notion of polarization in linear logic).

\begin{proposition}[Multiplicatives on polarized conducts]
The relationship between polarization and multiplicative connectives is as depicted in Figure \ref{polarizedmult}.
\end{proposition}

\begin{figure}
\centering
\subfigure[Tensor]{
\begin{tabular}{c||c|c}
$\otimes$ & N & P\\
\hline\hline
N & N & P\\
P & P & ?
\end{tabular}
}
\subfigure[Parr]{
\begin{tabular}{c||c|c}
$\parr$ & N & P\\
\hline\hline
N & ? & N\\
P & N & P
\end{tabular}
}
\caption{Multiplicative Connectives and Polarized Conducts}\label{polarizedmult}
\end{figure}

\subsubsection{The modified construction: GoI5.2}

The omission of the restriction to strict hyperfinite projects was suggested by Girard \cite{syntran}, and leads to a more satisfying construction. Indeed, one still defines additives on conducts as:
$$\cond{A\oplus B}=((\extend{A}{B})^{\pol\pol}\cup(\extend{B}{A})^{\pol\pol})^{\pol\pol}$$
$$\cond{A\with B}=(\extend{(A^{\pol})}{B})^{\pol}\cap(\extend{(B^{\pol})}{A})^{\pol}$$
But now the restriction to conducts such that $a=0$ for all $\de{a}$ can be made self-dual, leading to the notion of \emph{dichologies} which correspond to the notion of behaviours defined in Section \ref{gdisection}.

\begin{definition}[Dichologies]
A \emph{dichology} is a conduct $\cond{A}$ such that:
\begin{itemize}
\item if $\de{a}\in\cond{A}$ then for all $\lambda\in\realN$, $\de{a}+\lambda\de{0}\in\cond{A}$.
\item if $\de{a}\in\cond{A^{\pol}}$ then for all $\lambda\in\realN$, $\de{a}+\lambda\de{0}\in\cond{A^{\pol}}$.
\end{itemize}
It is proper if neither $\cond{A}$ nor $\cond{A^{\pol}}$ are empty.
\end{definition}

One can then get a characterisation of proper dichologies in the same way we characterised proper behaviours in the sliced graph setting (Propositions \ref{charact} and \ref{charac}). As a matter of fact, all results stated in Section \ref{gdisection} can be obtained in this setting: even the proofs are the same (using the adjunction for operators rather than the adjunction for graphs). Additive connectives are therefore defined between any two dichologies, without restriction. Moreover, the multiplicative constructions behave in a much more satisfying way.

\begin{proposition}[Connectives and Dichologies]
If $\cond{A,B}$ are dichologies with disjoint carriers, then $\cond{A^{\pol}, A\parr B,A\otimes B,A\with B,A\oplus B}$ are dichologies.
\end{proposition}

We will therefore consider the embedding in GoI5.2 rather than the original version. We want however to stress that the graph setting could very well be embedded in GoI5.1 in the same way by considering the restriction to \emph{strict} projects, i.e. projects $\de{a}=(a,A)$ where $\unit{A}\neq 0$.

\subsection{Graphs and Hyperfinite GoI}

The results and definition relating operator and graph operations we have exposed up to this point only dealt with graphs and not sliced graphs. We now extend these results and definitions to sliced graphs.

\begin{definition}
Let $A$ be a sliced graph. The localised weight matrix of $A$, denoted by $\mat{A}$ is defined as $\bigoplus_{i\in I^{A}} \mat{{A_{i}}}$, i.e. the direct sum of the localised weight matrices corresponding to the slices of $A$.
\end{definition}

\begin{definition}
A \emph{sliced operator graph} is a sliced graph $A$ such that $\mat{A}$ is of norm at most $1$. It is said to be symmetric if $\mat{A}$ is moreover hermitian.

An operator project is a project $(a,A)$ where $A$ is a sliced operator graph.
\end{definition}

We now want to associate to each operator project $\de{a}=(a,A)$ a hyperfinite project $\phi(\de{a})=(p,b,\mathcal{A},\alpha,\bar{\phi}(A))$. The finite projection $p$ of the resulting project will of course be the projection associated to the set of vertices of $A$ and the wager of $\phi(\de{a})$ will be equal to the wager of $\de{a}$, i.e. $p=p_{V^{A}}$ and $a=b$. Moreover, the idiom $\mathcal{A}$ will be defined as $\bigoplus_{i\in I^{A}} \complexN$ so the image of a project will therefore have a \emph{commutative} idiom, while the coefficients $\alpha^{A}_{i}$ for $i\in I^{A}$ will define the pseudo-trace on $\mathcal{A}$.

We therefore first define the embedding $\bar{\phi}$ that maps an operator graph $G$ to the operator $\mat{G}\otimes 1_{\mathcal{R}}$ in $\mathcal{R}_{0,1}$. This defines a translation of graphs as operators in the hyperfinite type $\text{II}_{\infty}$ factor. To a sliced graph is associated a direct sum of matrices $\bigoplus_{i\in I}\mat{A_{i}}$, i.e. a matrix in the algebra $\B{\hil{H}}\otimes\bigoplus_{i\in I} \complexN$. We can then extend $\bar{\phi}$ to such a direct sums of matrices by considering $\bar{\phi}\otimes\text{Id}_{\bigoplus_{i^{I}}\complexN}$ which we will abusively denote by $\bar{\phi}$ in the following.

This embedding can then be extended to map operator projects to hyperfinite projects as follows. We will denote by $\bigoplus_{i\in I}\lambda_{i}$, where $\lambda_{i}$ are real numbers, the pseudo-trace on $\bigoplus_{i\in I}\complexN$ defined as:
$$(\bigoplus_{i\in I}\lambda_{i})(\bigoplus_{i\in I}x_{i})=\sum_{i\in I}\lambda_{i}x_{i}$$

\begin{definition}
To an operator project $\de{a}=(a,A)$ we associate the hyperfinite project $\phi(\de{a})=a'\cdot +\cdot \alpha^{A}+\bigoplus_{i\in I^{A}} \bar{\phi}(\mat{A_{i}})$, where:
\begin{itemize}
\item $a'=-\infty$ if $\unit{A}<0$ and $a=\infty$, and $a'=a$ otherwise;
\item $\alpha^{A}$ is the pseudo-trace defined as $\bigoplus_{i\in I^{A}} \alpha^{A}_{i}$
\end{itemize}
\end{definition}

\begin{lemma}
Let $\vn{A,B}$ be idioms and $\alpha,\beta$ be pseudo-traces on $\vn{A,B}$ respectively. Then for all operators $A\in\finhyp\otimes\vn{A}$ and $B\in\finhyp\otimes\vn{B}$, we will denote by $A\oplus B$ the operator in $\finhyp\otimes(\vn{A}\oplus\vn{B})$ defined by $\iota_{1}(A)+\iota_{2}(B)$ where $\iota_{1}$ (resp. $\iota_{2}$) is the embedding $\finhyp\otimes\vn{A}\rightarrow \finhyp\otimes(\vn{A\oplus B})$ (resp. $\finhyp\otimes\vn{B}\rightarrow \finhyp\otimes(\vn{A\oplus B})$). One can then define the pseudo-trace $\alpha\oplus\beta$ on $\vn{A}\oplus\vn{B}$, and we have:
\begin{equation*}
-\log(\det{}^{FK}_{tr\otimes(\alpha\oplus\beta)}(1-A\oplus B))=-\log(\det{}^{FK}_{tr\otimes\alpha}(1-A))-\log(\det{}^{FK}_{tr\otimes\beta}(1-B))
\end{equation*}
\end{lemma}

\begin{proof}
Since $\det{}^{FK}(AB)=\det{}^{FK}(A)\det{}^{FK}(B)$, one obtains:
$$-\log(\det{}^{FK}(A))-\log(\det{}^{FK}(B))=-\log(\det{}^{FK}(AB))$$ 
Then, using the fact that $\det{}^{FK}$ is multiplicative:
\begin{eqnarray*}
\lefteqn{-\log(\det{}^{FK}_{tr\otimes(\alpha\oplus\beta)}(1-A\oplus B))}\\
&=&-\log(\det{}^{FK}_{tr\otimes(\alpha\oplus\beta)}((1-A)\oplus (1-B)))\\
&=&-\log(\det{}^{FK}_{tr\otimes(\alpha\oplus\beta)}(1-A)\oplus 1)-\log(\det{}^{FK}_{tr\otimes(\alpha\oplus\beta)}(1\oplus (1-B)))
\end{eqnarray*}

Let us compute $-\log(\det{}^{FK}_{tr\otimes(\alpha\oplus\beta)}(1-A)\oplus 1)$:
\begin{eqnarray*}
-\log(\det{}^{FK}_{tr\otimes(\alpha\oplus\beta)}((1-A)\oplus 1))&=&tr\otimes(\alpha\oplus\beta)(\log(\abs{(1-A)\oplus 1}))\\
&=&tr\otimes(\alpha\oplus\beta)(\log(\abs{1-A}\oplus 1))\\
&=&tr\otimes(\alpha\oplus\beta)(\log(\abs{(1-A)})\oplus 0)\\
&=&tr\otimes\alpha(\log(\abs{(1-A)}))\\
&=&-\log(\det{}^{FK}_{tr\otimes\alpha}(1-A))
\end{eqnarray*}
Similarly, $-\log(\det{}^{FK}_{tr\otimes(\alpha\oplus\beta)}(1\oplus (1-B)))=-\log(\det{}^{FK}_{tr\otimes\beta}(1-B))$. We can now conclude.
\end{proof}

\begin{lemma}
Let $a\cdot +\cdot\alpha + A$ be a hyperfinite project such that $\vn{A}=\complexN$ (in which case $\alpha\in\realN$):
\begin{equation*}
-\log(\det{}^{FK}_{tr\otimes\alpha}(1-A))=-\alpha\times \log(\det{}^{FK}(1-A))
\end{equation*}
\end{lemma}

\begin{proof}
By definition:
\begin{eqnarray*}
\det{}^{FK}_{tr\otimes\alpha}(1-A)&=&\exp(-tr\otimes\alpha(\log(\abs{1-A})))\\
&=&\exp(\alpha\times(-\tr(\log(\abs{1-A}))))\\
&=&(\exp(-\tr(\log(\abs{1-A}))))^{\alpha}\\
&=&(\det{}^{FK}(1-A))^{\alpha}
\end{eqnarray*}
Thus $-\log(\det{}^{FK}_{tr\otimes\alpha}(1-A))=-\log((\det{}^{FK}(1-A))^{\alpha}=-\alpha \log(\det{}^{FK}(1-A))$.
\end{proof}

\begin{proposition}\label{measplongtranche}
Let $G,H$ be sliced graphs. Then: 
\begin{equation*}
\meas{G,H}=-\log(\det{}^{FK}_{tr\otimes \alpha^{G}\otimes\alpha^{H}}(1-\Phi(\mat{G}^{\dagger_{H}})\Phi(\mat{H}^{\ddagger_{G}})))
\end{equation*}
\end{proposition}

\begin{proof}
Notice that $\mat{G}^{\dagger_{H}}=\oplus_{i\in I^{G}}\oplus_{j\in I^{H}} \mat{G_{i}}$ and $\mat{H}^{\ddagger_{G}}=\oplus_{i\in I^{G}}\oplus_{j\in I^{H}} \mat{H_{j}}$. Thus:
\begin{eqnarray*}
\lefteqn{-\log(\det{}^{FK}_{tr\otimes \alpha^{G}\otimes\alpha^{H}}(1-\Phi(\mat{G}^{\dagger_{H}})\Phi(\mat{H}^{\ddagger_{G}})))}\\
&=&-\log(\det{}^{FK}_{tr\otimes \alpha^{G}\otimes\alpha^{H}}(1-\Phi(\mat{G}^{\dagger_{H}}\mat{H}^{\ddagger_{G}})))\\
&=&-\log(\det{}^{FK}_{tr\otimes (\oplus_{i\in I^{G}}\oplus_{j\in I^{H}})}(1-\oplus_{i\in I^{G}}\oplus_{j\in I^{H}}\bar{\phi}(\mat{G_{i}}\mat{H_{j}})))\\
&=&\sum_{i\in I^{G}}\sum_{j\in I^{H}} -\log(\det{}^{FK}_{tr\otimes\alpha^{G}_{i}\alpha^{H}_{j}}(1-\bar{\phi}(\mat{G_{i}})\bar{\phi}(\mat{H_{j}}))\\
&=&\sum_{i\in I^{G}}\sum_{j\in I^{H}} -\alpha^{G}_{i}\alpha^{H}_{j}\log(\det{}^{FK}(1-\bar{\phi}(\mat{G_{i}})\bar{\phi}(\mat{H_{j}})))
\end{eqnarray*}
Now, we use the fact that $-\log(\det{}^{FK}(1-\bar{\phi}(\mat{G_{i}})\bar{\phi}(\mat{H_{j}})))=\meas[-\log(1-x)]{G_{i},H_{j}}$, a result that was proven in our earlier paper \cite{seiller-goim}. We thus obtain:
\begin{equation*}
-\log(\det{}^{FK}_{tr\otimes \alpha^{G}\otimes\alpha^{H}}(1-\Phi(\mat{G}^{\dagger_{H}})\Phi(\mat{H}^{\ddagger_{G}})))=\sum_{i\in I^{G}}\sum_{j\in I^{H}} \alpha^{G}_{i}\alpha^{H}_{j}\meas{G_{i},H_{j}}
\end{equation*}
Finally, $-\log(\det{}_{tr\otimes \alpha^{G}\otimes\alpha^{H}}(1-\mat{G}^{\dagger_{H}}\mat{H}^{\ddagger_{G}}))=\meas{F,G}$.
\end{proof}

\begin{proposition}\label{execplongtranche}
Let $F,G$ be sliced graphs. Then $\mat{{F\plug G}}$ is the solution to the feedback equation involving $\mat{{F}}^{\dagger_{G}}$ and $\mat{{G}}^{\ddagger_{F}}$.
\end{proposition}

\begin{proof}
By definition $\what{F\plug G}=\sum_{i\in I^{F}}\sum_{j\in J^{G}}\alpha^{F}_{i}\alpha^{G}_{j}\what{F_{i},G_{j}}$. Thus:
\begin{eqnarray*}
\mat{{F\plug G}}&=&\mat{\sum_{i\in I^{F}}\sum_{j\in J^{G}}\alpha^{F}_{i}\alpha^{G}_{j}\what{F_{i},G_{j}}}\\
&=&\oplus_{i\in I^{F}}\oplus_{j\in J^{G}}\mat{{F_{i},G_{j}}}
\end{eqnarray*}
Proposition \ref{execution} ensures us that $\mat{{F_{i},G_{j}}}$ is the solution to the feedback equation involving $\mat{F_{i}}$ and $\mat{G_{j}}$. From this, one can deduce that $\oplus_{i\in I^{F}}\oplus_{j\in J^{G}}\mat{{F_{i},G_{j}}}$ is the solution to the feedback equation involving $\mat{{F}}^{\dagger_{G}}$ and $\mat{{G}}^{\ddagger_{F}}$. Finally, $\mat{{F\plug G}}=\ex{\mat{{F}}^{\dagger_{G}},\mat{{G}}^{\ddagger_{F}}}$.
\end{proof}

From the preceding propositions, one can then show that the interaction graphs setting, in the special case of $m(x)=-\log(1-x)$, is a combinatorial version of the hyperfinite geometry of interaction of Girard.

\begin{theorem}
Let $\de{a,b}$ be operator projects, and let $m$ be the map $x\mapsto -\log(1-x)$. We have the following:
\begin{eqnarray}
\sca{a}{b}&=&\mathopen{\ll}\phi(\de{a}),\phi(\de{b})\mathclose{\gg}\\
\phi(\de{a\deplug b})&=&\phi(\de{a})\deplug\phi(\de{b})\\
\phi(\de{a\otimes b})&=&\phi(\de{a})\otimes\phi(\de{b})\\
\phi(\de{a+\lambda b})&=&\phi(\de{a})\with\lambda\phi(\de{b})\\
\phi(\de{a})\poll\phi(\de{b})&\text{ iff }&\de{a}\poll\de{b}
\end{eqnarray}
\end{theorem}

\begin{proof}
The first equality is given by Proposition \ref{measplongtranche}, and it induces the last statement. The second equality is a consequence of Propositions \ref{measplongtranche} and \ref{execplongtranche}:
\begin{eqnarray*}
\phi(\de{a\deplug b})&=&\phi(\sca[-\log(1-x)]{a}{b},A\plug B)\\
&=&\sca[-\log(1-x)]{a}{b}\cdot+\cdot \bigoplus_{(i,j)\in I^{A}\times I^{B}}\alpha^{A}_{i}\alpha^{B}_{j}+\Phi(A\plug B)\\
&=&\sca{\Phi(A)}{\Phi(B)}\cdot+\cdot (\bigoplus_{i\in I^{A}}\alpha^{A}_{i})\otimes (\bigoplus_{j\in I^{B}}\alpha^{B}_{j})+\Phi(A\plug B)\\
&=&\sca{\Phi(A)}{\Phi(B)}\cdot+\cdot (\bigoplus_{i\in I^{A}}\alpha^{A}_{i})\otimes (\bigoplus_{j\in I^{B}}\alpha^{B}_{j})+\ex{A^{\dagger_{B}},B^{\ddagger_{A}}}
\end{eqnarray*}
The third equality is obtained by noticing that $\de{a\otimes b}=\de{a\deplug b}$ (once the locations of $\de{a}$ and $\de{b}$ are disjoint, the two constructions coincide). This leaves the equality $\phi(\de{a+\lambda b})=\phi(\de{a})\with\lambda\phi(\de{b})$. The proof of this, however is a direct consequence of the definition of $+$ and the definition of $\with$ in Girard's framework.
\end{proof}

\subsection{Orthogonality as Nilpotency}

\begin{definition}
Let $G$ be a graph. We define the \emph{localised connectivity matrix} $\matconn{G}$ of $G$ to be the operator of $p_{V^{G}}\mathcal{B}(\hil{H})p_{V^{G}}\subset\mathcal{B}(\hil{H})$ whose matrix in the base $\{e_{i}~|~i\in V^{G}\}$ is the connectivity matrix of $G$, that is the matrix $(a_{i,j})_{i,j\in V^{G}}$ such that:
$$a_{i,j}=\left\{\begin{array}{ll}
	0 & \text{ if $E^{G}=\emptyset$}\\
	1 & \text{ otherwise}
	\end{array}\right.$$

The \emph{connectivity matrix} of a sliced graph $G=\{G_{i}\}_{i\in I^{G}}$ is defined as the direct sum $\matconn{G}=\bigoplus_{i\in I^{G}}\matconn{G_{i}}$.

If $G,H$ are two sliced graph, we define: $$\matconn{G}\star\matconn{H}=\bigoplus_{(i,j)\in I^{G}\times I^{H}} \matconn{G_{i}}\matconn{H_{j}}$$
\end{definition}

\begin{remark}
The $\star$ operation is defined in this way in order to avoid dealing with the dialects and the operations $(\cdot)^{\dagger}$ and $(\cdot)^{\ddagger}$ (defined in the preceding subsections). However, one could very well define it as the connectivity matrix of the product $(\matconn{G})^{\dagger_{H}}(\matconn{H})^{\ddagger_{G}}$.
\end{remark}

\begin{proposition}\label{nilpotency}
Let  $\de{a,b}$ be two projects with carrier $V$, and $m(x)=\infty$ for $x\in]0,1]$. Then:
\begin{equation*}
\de{a}\poll\de{b}\Leftrightarrow\left\{\begin{array}{l}\matconn{A}\star\matconn{B}\text{ is nilpotent}\\ \unit{A}b+\unit{B}a\neq 0,\infty\end{array}\right.
\end{equation*}
In particular, if $A,B$ have only one slice, the product $\matconn{A}\matconn{B}$ is nilpotent.
\end{proposition}

\begin{proof}
We show the left-to-right implication: $\de{a}\poll\de{b}$ implies $a\unit{B}+b\unit{A}+\meas{A,B}\neq 0,\infty$. But, if there were a cycle in one of the $A_{i}\bicol B_{j}$, the last term $\meas{A,B}$ would be equal to $\infty$. Hence $\sca{a}{b}$ would be infinite, and the projects would not be orthogonal. Thus $\matconn{A}\star\matconn{B}$ is nilpotent and $\meas{A,B}=0$, which means that $\unit{B}a+b\unit{A}\neq 0,\infty$.

The converse is straightforward.
\end{proof}

\begin{remark}
This proposition is true because we are working with finite graphs. What we are really proving is that $\de{a}\poll\de{b}$ if and only if $ \unit{A}b+\unit{B}a\neq 0,\infty$ and for all $i,j$, \emph{no cycles appear in $A_{i}\bicol B_{j}$}. In the case of infinite graphs, this condition would imply weak nilpotency (hence more in the style of the second version of GoI \cite{goi2}).
\end{remark}

\begin{corollary}
Let $\de{a}=(a,A)$ be a project, and $\de{a'}=(a,A')$ be such that $\matconn{A'}=\matconn{A}$. Then $\de{a}\cong_{\cond{A}}\de{a'}$ for all conduct $\cond{A}$ containing $\de{a}$.
\end{corollary}

The model we obtain when taking $m(x)=\infty$ can therefore be reduced, by working up to observational equivalence, to working with simple (at most one edge between two points) non-weighted directed graphs.

Notice however the differences between the first versions of GoI and our framework. The addition of the wager is a quite important improvement: without it, we would have $\cond{1}=\cond{\bot}$. Moreover, the additive construction (the use of slices) allows us to define, as we have shown, a categorical model of MALL. Looking a little closer at this model, one can see, however, that it is not that exciting.

\begin{proposition}\label{trivial}
Let $m(x)=\infty$ and $\cond{A}$ be a behaviour. Then $\cond{A}$ is either empty or the orthogonal of an empty conduct.
\end{proposition}

\begin{proof}
Notice that a proper behaviour and its orthogonal contain only project $\de{a}=(a,A)$ with $a=0$. But two such projects cannot be orthogonal when $m(x)=\infty$. Hence there are no proper behaviours and using Proposition \ref{charac} we have the result.
\end{proof}

Hence, the categorical model we get with an orthogonality defined by nilpotency is nothing more than a truth-value model.

\subsection{Multiplicatives}

We will here show how one can obtain a refined version of the construction depicted by Girard in his \emph{multiplicatives} paper \cite{multiplicatives}. To obtain it, one needs not only to consider a particular choice of map $m$, but also to change slightly the notion of orthogonality. Before going into the details of the construction, let us first explain this change.

As in the classical realisability setting \cite{krivine}, we will define a notion of \emph{pole}. A pole will be a subset of $\realN$ that contains the values of the measurement representing a successful interaction.

\begin{definition}
A \emph{pole} is a subset $\pole\in\realN$.
\end{definition}

One can then define orthogonality with respect to the pole in the following way:
\begin{definition}
Two projects $\de{a,b}$ are \emph{orthogonal with respect to the pole $\pole$} when $\sca{a}{b}\in\pole$. This will be denoted by $\de{a}\poll\de{b}$.
\end{definition}

\begin{remark}
This definition extends the definitions given in Section \ref{gdisection}: one just have to choose the pole $\realN-\{0\}$.
\end{remark}

The reason why we chose not to present the whole paper in this way comes from the fact that the main properties of the additive construction use the fact\footnote{The construction of additive connectives can be performed with any pole not containing $0$. This obviously does not mean that no additive construction can be performed in the other cases. However, the answer to the question of defining additives in the case of an arbitrary pole is outside the scope of this paper.} that we work with the particular pole $\realN-\{0\}$. However, all the results concerning multiplicative connectives are independent of the choice of a pole. Among these multiplicative constructions, we will have a closer look at the one defined by:
\begin{itemize}
\item the pole $\realN-\{1\}$;
\item the map $m(x)=1$.
\end{itemize}

Let us define $G_{\sigma}$ to be the graph associated to a permutations $\sigma$ on the set $X$ by:
\begin{eqnarray*}
V^{G_{\sigma}}&=&X\\
E^{G_{\sigma}}&=&X\\
s^{G_{\sigma}}&=&x\mapsto x\\
t^{G_{\sigma}}&=&x\mapsto \sigma(x)\\
\omega^{G_{\sigma}}&=&x\mapsto 1
\end{eqnarray*}

Of course, the execution defined in this paper is a generalisation of the execution defined by Girard between permutations. In this case, as in the first GoI constructions, the difference is that our execution is always defined while Girard's was defined only when no cycles appeared between the two permutations. But in the case the execution as permutations is defined, the two notion coincide. This means that if $\sigma,\tau$ are permutations and $\Ex(\sigma,\tau)$ is defined, then:
$$G_{\Ex(\sigma,\tau)}=G_{\sigma}\plug G_{\tau}$$
Orthogonality was defined as $\sigma\poll\tau$ if and only if $\sigma\tau$ is a cyclic permutation. This is equivalent to saying that $\sigma\poll\tau$ if and only if $\sigma\tau$ decomposes in a single cycle. Which can be translated in terms of graphs as:
$$\sigma\poll\tau \Leftrightarrow \meas{G_{\sigma},G_{\tau}}=1$$
These results lead to the following theorem.

\begin{theorem}
If $\sigma, \tau$ are permutations and $\de{g}_{\sigma}=(0,G_{\sigma})$ and $\de{g}_{\tau}=(0,G_{\tau})$ are the associated projects:
\begin{eqnarray*}
\sigma\poll \tau&\Leftrightarrow& \de{g}_{\sigma}\poll \de{g}_{\tau}\\
\de{g}_{\sigma\plug\tau}&=&\de{g}_{\sigma}\deplug \de{g}_{\tau}\\
\de{g}_{\sigma\otimes\tau}&=&\de{g}_{\sigma}\otimes \de{g}_{\tau}
\end{eqnarray*}
\end{theorem}

\begin{proof}
The only thing we did not already prove is the fact that $\de{g}_{\sigma\otimes \tau}=\de{g}_{\sigma}\otimes \de{g}_{\tau}$. But this equality is a consequence of the previous one, taking into account the fact that the tensor product is a special case of execution (the case when carriers are disjoint).
\end{proof}

\section{An application of graphs for GoI5}

In this section, we will show how the combinatorial version of GoI5 obtained from the graphs of interaction can help understand important and fine aspects of the hyperfinite GoI construction. We will focus on the construction of additives, and more particularly on the fact that the set of projects $\cond{A+B}$ does not generate the behaviour $\cond{A\with B}$. As it turns out, the counter-example produced in the proof of Proposition \ref{counterexgraphs} can be used to obtain a similar result in the hyperfinite GoI. While this result could have been obtained directly in the hyperfinite setting, the finding of this proof came through the setting of graphs in which intuitions are easier to come by. As it turns out, one just has to mimic the proof of the graph case to obtain the same result with operators.

For this, we recall that in the theory of von Neumann algebras, one defines a notion of equivalence of projections: two projections $p,q$ are equivalent (in the sense of Murray and von Neumann) in a von Neumann algebra $\vn{M}$ when there exists an element\footnote{The equations satisfied by $u$ imply that $u$ is a \emph{partial isometry}.} $u\in\vn{M}$ such that $uu^{\ast}=q$ and $u^{\ast}u=p$. One can then quotient the partial order on projections defined by $p\leqslant q$ if and only if $pq=p$ by this equivalence relation to obtain a partial order denoted by $\precsim_{\vn{M}}$. It is a standard fact that if $\vn{M}$ is a factor, then the order $\precsim_{\vn{M}}$ is total. In particular, given two projections $p,q$, then one of them is equivalent to a sub-projection of the other.

The main idea of the proof of Proposition \ref{counterexgraphs} is to construct an element which is in $\cond{(A+ B)^{\pol}}$ but not in $\cond{(A\with B)^{\pol}}$ by considering a project of the form $\de{a+b}+\de{0}_{u}$, where $\de{a,b}$ are in $\cond{A,B}$ respectively and $\de{0}_{u}$ consists in a single edge from the carrier of $\cond{A}$ to the carrier of $\cond{B}$. The main intuition is that the edge in $\de{0}_{u}$ cannot be seen by tests in $\cond{A+B}$ — i.e. cannot be used to create cycles — while it can be used by tests in $\cond{A\with B}$. In the case of hyperfinite Geometry of Interaction, we will use the same proof. The only difficulty is to understand how to construct the hyperfinite project equivalent to $\de{0}_{u}$. We will use the total order $\precsim$ to obtain a partial isometry between the carrier of one of the two conducts $\cond{A,B}$ and a sub-projection of the other. This partial isometry will play the rôle of the edge used to define $\de{0}_{u}$. The rest of the proof simply consists in computations.

\subsection{GoI5.2 construction}

We will first need to show an important proposition which will be used extensively in the computations. This proposition (\autoref{detnil}) states that the Fuglede-Kadison determinant of $1-u$ is equal to $1$ when $u$ is nilpotent. We will use two properties of the Fuglede-Kadison determinant in the proof: the fact that the determinant of $A$ is less than or equal to the spectral radius $\specrad{A}$ of $A$, and the fact that the determinant is multiplicative and therefore $\det{}^{FK}(A^{-1})=(\det{}^{FK}(A))^{-1}$. These two properties are stated and proved in the Fuglede and Kadison paper \cite{FKdet}. 

\begin{lemma}\label{specnil}
If $A$ is nilpotent, then $\specrad{A}=0$. 
\end{lemma}

\begin{proof}
We know that $\specrad{A}=\lim_{n\rightarrow \infty} \norm{A^{n}}^{\frac{1}{n}}$. If $A$ is nilpotent of degree $k$, then $\norm{A^{n}}=0$ for all $n\geqslant k$. Thus $\specrad{A}=0$.
\end{proof}

\begin{lemma}\label{invnil}
Chose $k\in\naturalN$ and $\alpha_{1},\dots,\alpha_{k}$ in $\complexN$. If $A$ is nilpotent, then $P(A)=\sum_{i=1}^{k} \alpha_{k}A^{k}$ is nilpotent.
\end{lemma}

\begin{proof}
The minimal degree of $A$ in $P(A)^{i}$ is equal to $A^{i}$. Thus $P(A)^{i}=0$ for all $i\geqslant k$.
\end{proof}

\begin{proposition}\label{detnil}
If $A$ is nilpotent, then $\det{}^{FK}(1+A)=1$.
\end{proposition}

\begin{proof}
Let us denote by $k$ the nilpotency degree of $A$. Since $A$ is nilpotent, $\specrad{A}=0$ by Lemma \ref{specnil}. Chose $\lambda$ in the spectrum $\spec{1+A}$ of $1+A$. By definition, $\lambda-1-A$ is non-invertible, which means that $(\lambda-1)-A$ is non-invertible, i.e. $\lambda-1\in\spec{A}$. This implies that $\lambda=1$ since the spectrum of $A$ is reduced to $\{0\}$. Thus $\specrad{1+A}\leqslant 1$ and therefore $\det{}^{FK}(1+A)\leqslant 1$.

Moreover, $(1+A)^{-1}=\sum_{i=0}^{k-1} (-A)^{i}=1+\sum_{i=1}^{k-1}(-A)^{i}$. By Lemma \ref{invnil}, we know that $B=\sum_{i=1}^{k-1} (-A)^{i}$ is nilpotent, and thus $\det{}^{FK}(1+B)\leqslant 1$ applying the same arguments as before. Since $\det{}^{FK}(1+B)=\det{}^{FK}((1+A)^{-1})=(\det{}^{FK}(1+A))^{-1}$, we finally conclude that $\det{}^{FK}(1+A)=1$.
\end{proof}

\begin{proposition}\label{contreexgdi52}
Let $\cond{A,B}$ be proper dichologies. There exists a hyperfinite project $\de{f}\in(\cond{A^{\pol}+B^{\pol}})^{\pol}$ such that $\de{f}\not\in\cond{A\oplus B}$.
\end{proposition}

\begin{proof}
We can suppose, without loss of generality, that $p\precsim q$. As a consequence, there exists a projection $p'\leqslant q$ such that $p'\sim p$. Now, let $u$ be a partial isometry such that $uu^{\ast}=p$ and $u^{\ast}u=p'$. Let $\de{a}\in\cond{A}$ and $\de{c}=\de{a}_{p+q}+\de{0}_{u}$ where $\de{0}_{u}=0\cdot+\cdot 1_{\complexN}+(u+u^{\ast})$.

We will now show that the hyperfinite project $\de{c}$ satisfies the wanted properties: we will first show that $\de{c}\in\cond{(A^{\pol}+B^{\pol})}^{\pol}$, and then we will show that $\de{c}\not\in\cond{A\oplus B}$.

\begin{itemize}
\item[]\textbf{Showing that $\de{c}\in\cond{(A^{\pol}+B^{\pol})^{\pol}}$.}~\\
Let $\de{d}=\de{a'}_{p+q}+\de{b'}_{p+q}\in\cond{A^{\pol}+B^{\pol}}$. Then:
\begin{eqnarray*}
\sca[]{c}{d}&=&\sca[]{a_{\text{p+q}}+0_{\text{u}}}{a'_{\text{p+q}}+b'_{\text{p+q}}}\\
&=&\sca[]{a_{\text{p+q}}}{a'_{\text{p+q}}}+\sca[]{a_{\text{p+q}}}{b'_{\text{p+q}}}+\sca[]{0_{\text{u}}}{a'_{\text{p+q}}}+\sca[]{0_{\text{u}}}{b'_{\text{p+q}}}\\
&=&\sca[]{a}{a'}+\sca[]{0_{\text{u}}}{a'_{\text{p+q}}}+\sca[]{0_{\text{u}}}{b'_{\text{p+q}}}
\end{eqnarray*}
Since $u^{\ast}=u^{\ast}uu^{\ast}=p'u^{\ast}$, $A'=A'p$ and $pp'=0$, one has:
\begin{eqnarray*}
-\log(\det{}^{FK}(1-A'(u+u^{\ast})))&=&-\log(\det{}^{FK}(1-A'u-A'u^{\ast}))\\
&=&-\log(\det{}^{FK}(1-A'u-A'pp'u^{\ast}))\\
&=&-\log(\det{}^{FK}(1-A'u))
\end{eqnarray*}
Since $u=uu^{\ast}u=up'$ and $A'=pA'$, we have $A'uA'u=A'up'pA'u=0$, i.e. $A'u$ is nilpotent. By using Proposition \ref{detnil}, we then obtain:
\begin{eqnarray*}
-\log(\det{}^{FK}(1-A'u))&=&0
\end{eqnarray*}
Thus $\sca[]{0_{\text{u}}}{a'_{\text{p+q}}}=0$ since all wagers are equal to zero. We show in a similar way that $\sca[]{0_{\text{u}}}{b'_{\text{p+q}}}=0$ since $uB=puBq$ and thus $(uB)^{2}=0$:
\begin{eqnarray*}
-\log(\det{}^{FK}(1-(u+u^{\ast})B))&=&-\log(\det{}^{FK}(1-uB-u^{\ast}B))\\
&=&-\log(\det{}^{FK}(1-uB-u^{\ast}pqB))\\
&=&-\log(\det{}^{FK}(1-uB))\\
&=&-\log(1)
\end{eqnarray*}
As a consequence we have $\sca[]{c}{d}=\sca[]{a}{a'}$, i.e. $\de{c}\in\cond{(A^{\pol}+B^{\pol})}^{\pol}$.

\item[]\textbf{Showing that $\de{c}\not\in\cond{A\oplus B}$.}~\\
To show this, we will find an element $\de{t}$ in $(\cond{A}_{p+q})^{\pol}\cap(\cond{B}_{p+q})^{\pol}$ such that $\de{c}\not\poll\de{t}$. Chose $\de{b'}\in\cond{B}$, $\de{a'}\in\cond{A}$, $\lambda\in\realposN$ with $\abs{\lambda}<1$ and define $\de{0_{\text{$\lambda$u}}}=0\cdot+\cdot 1_{\complexN}+\lambda(u+u^{\ast})$. We will show that there exists a real number $\mu$ such that $\de{t}=\de{b'}_{p+q}+\de{a'_{\text{p+q}}}+\mu \de{0}_{\text{$\lambda$u}}$ satisfies $\de{t}\in(\cond{A}_{p+q})^{\pol}$, $\de{t}\in(\cond{B}_{p+q})^{\pol}$ and $\de{t}\not\poll\de{c}$. Let  us chose $\de{b}\in\cond{B}$. One can compute:
\begin{eqnarray*}
\sca[]{t}{b_{\text{p+q}}}&=&\sca[]{b'_{\text{p+q}}+a'_{\text{p+q}}+\mu 0_{\text{$\lambda$u}}}{b_{\text{p+q}}}\\
&=&\sca[]{b'_{\text{p+q}}}{b_{\text{p+q}}}+\sca[]{a'_{\text{p+q}}}{b_{\text{p+q}}}+\mu\sca[]{0_{\text{$\lambda$u}}}{b_{p+q}}\\
&=&\sca[]{b}{b'}+\sca[]{a'_{\text{p+q}}}{b_{\text{p+q}}}+\mu\sca[]{0_{\text{$\lambda$u}}}{b_{p+q}}
\end{eqnarray*}
Since $A'B=0$, we have $\sca[]{a'_{\text{p+q}}}{b_{\text{p+q}}}=0$.
Moreover we have, as in previous computations:
\begin{eqnarray*}
\sca[]{0_{\text{$\lambda$u}}}{b_{p+q}}&=&-\log(\det{}^{FK}(1-\lambda(u+u^{\ast})B))\\
&=&-\log(\det{}^{FK}(1-\lambda uB+\lambda u^{\ast}pB))\\
&=&-\log(\det{}^{FK}(1-\lambda uB))\\
&=&-\log(1)\\
&=&0
\end{eqnarray*}
This shows that $\sca[]{t}{b_{\text{p+q}}}=\sca[]{b}{b'}$, hence $\de{t}\in(\cond{B}_{p+q})^{\pol}$.

Let us now chose $\de{a}\in\cond{A}$. One can compute:
\begin{eqnarray*}
\sca[]{t}{a_{\text{p+q}}}&=&\sca[]{b'_{\text{p+q}}+a'_{\text{p+q}}+\mu 0_{\text{$\lambda$u}}}{a_{\text{p+q}}}\\
&=&\sca[]{b'_{\text{p+q}}}{a_{\text{p+q}}}+\sca[]{a'_{\text{p+q}}}{a_{\text{p+q}}}+\mu\sca[]{0_{\text{$\lambda$u}}}{a_{p+q}}\\
&=&\sca[]{a}{a'}+\sca[]{b'_{\text{p+q}}}{a_{\text{p+q}}}+\mu\sca[]{0_{\text{$\lambda$u}}}{a_{p+q}}
\end{eqnarray*}
Since $AB'=0$, we have $\sca[]{b'_{\text{p+q}}}{a_{\text{p+q}}}=0$.
Moreover:
\begin{eqnarray*}
\sca[]{0_{\text{$\lambda$u}}}{a_{p+q}}&=&-\log(\det{}^{FK}(1-\lambda(u+u^{\ast})A))\\
&=&-\log(\det{}^{FK}(1-\lambda up'A+\lambda u^{\ast}A))\\
&=&-\log(\det{}^{FK}(1-\lambda u^{\ast}A))\\
&=&-\log(1)\\
&=&0
\end{eqnarray*}
This shows that $\sca[]{t}{a_{\text{p+q}}}=\sca[]{a}{a'}$, hence $\de{t}\in(\cond{A}_{p+q})^{\pol}$.

We just showed that $\de{t}\in(\cond{A}_{p+q})^{\pol}\cap(\cond{B}_{p+q})^{\pol}$. This shows that $\de{t}\in\cond{(A\oplus B)^{\pol}}$:
\begin{eqnarray*}
\cond{(A\oplus B)^{\pol}}&=&\cond{(A_{\text{p+q}}\cup B_{\text{p+q}})^{\pol\pol\pol}}\\
&=&\cond{(A_{\text{p+q}}\cup B_{\text{p+q}})^{\pol}}\\
&=&\cond{(A_{\text{p+q}})^{\pol}\cap (B_{\text{p+q}})^{\pol}}
\end{eqnarray*}

Let us now compute $\sca[]{t}{c}$:
\begin{eqnarray*}
\sca[]{t}{c}&=&\sca[]{b'_{\text{p+q}}+a'_{\text{p+q}}+\mu0_{\text{$\lambda$u}}}{a_{\text{p+q}}+0_{\text{u}}}\\
&=&\sca[]{b'_{p+q}}{a_{\text{p+q}}}+\sca[]{b'_{\text{p+q}}}{0_{\text{u}}}+\sca[]{a'_{\text{p+q}}}{a_{\text{p+q}}}+\dots\\
&&~~~\dots+\sca[]{a'_{\text{p+q}}}{0_{\text{u}}}+\mu\sca[]{0_{\text{$\lambda$u}}}{a_{\text{p+q}}}+\mu\sca[]{0_{\text{$\lambda$u}}}{0_{\text{u}}}
\end{eqnarray*}
We have $\sca[]{b'_{\text{p+q}}}{a_{\text{p+q}}}=0$. Moreover, we know $\sca[]{b'_{\text{p+q}}}{0_{\text{u}}}$, $\sca[]{a'_{\text{p+q}}}{0_{\text{u}}}$ and $\sca[]{0_{\text{$\lambda$u}}}{a_{\text{p+q}}}$ are all equal to zero (this is once again the same reasoning: one of the two terms $u,u^{\ast}$ disappears, and we then use Proposition \ref{detnil}). This gives us:
\begin{eqnarray*}
\sca[]{t}{c}&=&\sca[]{a'_{\text{p+q}}}{a_{\text{p+q}}}+\mu\sca[]{0_{\text{$\lambda$u}}}{0_{\text{u}}}\\
&=&\sca[]{a'}{a}+\mu\sca[]{0_{\text{$\lambda$u}}}{0_{\text{u}}}
\end{eqnarray*}
But one can compute the second term, which is different from zero: 
\begin{eqnarray*}
\sca[]{0_{\text{$\lambda$u}}}{0_{\text{u}}}&=&-\log(\det{}^{FK}(1-\lambda(u+u^{\ast})(u+u^{\ast})))\\
&=&-2\tr(p)\log(1-\lambda)
\end{eqnarray*}
We can then define $\mu=\frac{\sca[]{a}{a'}}{2\tr(p)\log(1-\lambda)}$ and, in this case, $\sca[]{t}{c}=0$.
\end{itemize}
This shows that $\de{c}\in\cond{(A^{\pol}+B^{\pol})^{\pol}}$ and $\de{c}\not\in\cond{(A_{\text{p+q}}\cap B_{\text{p+q}})^{\pol\pol}}$ — i.e. $\de{c}\not\in\cond{A\oplus B}$.
\end{proof}

\subsection{GoI5.1}

Contradicting what we just showed (Proposition \ref{contreexgdi52}), Girard states (Proposition 16) in his paper that the set of elements of the form $\de{a+b}$ with $\de{a}\in\cond{A}$ and $\de{b}\in\cond{B}$ generates the conduct $\cond{A\with B}$ under bi-orthogonality. It is then natural to ask why the construction of additives in GoI5.2 satisfies weaker properties than the one defined by Girard. As it turns out, Proposition 16 in Girard's paper is not true. We will obtain the proof of this claim as an adaptation of the previous proof.

\begin{lemma}[(GoI5.1)]
Let $\cond{A}$ be a negative conduct, $\de{a}\in\cond{A}$ and $\de{b}\in\cond{A^{\pol}}$. Then we necessarily are in one of the two following cases:
\begin{itemize}
\item either the wager $b$ of $\de{b}$ is equal to zero;
\item or the interaction $-\log(\det{}^{FK}(1-A^{\dagger}B^{\ddagger}))$ is equal to zero.
\end{itemize}
\end{lemma}

\begin{proof}
Suppose that both the interaction $\lambda=-\log(\det{}^{FK}(1-A^{\dagger}B^{\ddagger}))$ and the wager $b$ of $\de{b}$ are different from zero. We have:
$$\sca[]{a}{b}=b\alpha(1_{\vn{A}})-\log(\det{}^{FK}(1-A^{\dagger}B^{\ddagger}))$$ 
Since $\alpha(1_{\vn{A}})\neq 0$, we can define the non-zero real number $\lambda/\alpha(1_{\vn{A}})$. Using a property of positive conducts, we can deduce that $\de{b'}=-\lambda/\alpha(1_{\vn{A}})\cdot +\cdot \beta+B$ is an element of $\cond{A^{\pol}}$. But $\sca[]{a}{b'}=0$, which is a contradiction. As a consequence, either $a=0$, or $\lambda=0$.
\end{proof}

\begin{lemma}[(GoI5.1)]
If $\cond{A,B}$ are negative conducts, one has $\cond{A+B}\subset\cond{A\with B}$.
\end{lemma}

\begin{proof}
Since $\cond{A\with B}=(\cond{A\oplus B})^{\pol}=(\cond{A}_{p+q}\cup\cond{B}_{p+q})^{\pol}$, we will show that any element of $\cond{A+B}$ is in the orthogonal of $\cond{A}_{p+q}\cup\cond{B}_{p+q}$. Let $\de{f+g}$ be an element of $\cond{A\with B}$, $\de{a}_{p+q}\in\cond{A}_{p+q}$ and $\de{b}\in\cond{B}_{p+q}$. Then:
\begin{eqnarray*}
\sca[]{f_{\text{p+q}}+g_{\text{p+q}}}{a_{\text{p+q}}}&=&\sca[]{f_{\text{p+q}}}{a_{\text{p+q}}}+\sca[]{g_{\text{p+q}}}{a_{\text{p+q}}}\\
&=&\sca[]{f}{a}+a\gamma(1_{\vn{G}})
\end{eqnarray*}
Using the preceding lemma, either $\sca[]{f}{a_{\text{p+q}}}=a\varphi(1_{\vn{F}})$ or $a=0$. In the first case, we obtain that $\sca[]{f_{\text{p+q}}+g_{\text{p+q}}}{a_{\text{p+q}}}=a(\varphi(1_{\vn{F}})+\gamma(1_{\vn{G}}))$ which is necessarily different from $0$ and $\infty$. In the second case, $\sca[]{f_{\text{p+q}}+g_{\text{p+q}}}{a_{\text{p+q}}}=\sca[]{f}{a}$ which is also different from $0$ and $\infty$.

We show in a similar way that $\sca[]{f_{\text{p+q}}+g_{\text{p+q}}}{b_{\text{p+q}}}\neq 0,\infty$.
\end{proof}

\begin{proposition}[Counter-examples in GoI5.1]\label{counterexgdi51}
Let $\cond{A,B}$ be two positive conducts with disjoint carriers. There exists a strict hyperfinite project in $\cond{(A^{\pol}+B^{\pol})^{\pol}}$ which is not an element of $\cond{A\oplus B}$.
\end{proposition}

\begin{proof}
This proof is an adaptation of the proof of Proposition \ref{contreexgdi52} that deals with the small differences between the GoI5.1 and GoI5.2 constructions. If we write $p$ and $q$ the respective (disjoint) carriers of $\cond{A}$ and $\cond{B}$, we can suppose without loss of generality that $p\precsim q$, and thus that there exists a projection $p'\leqslant q$ such that $p\sim p'$. We will write $u$ a partial isometry such that $uu^{\ast}=p$ and $u^{\ast}u=p'$. We chose an element $\de{a}=a\cdot+\cdot 1_{\complexN}+0$ in $\cond{A}$ and define $\de{c}=\de{a}_{p+q}+\de{0}_{u}-\de{0}$ where $\de{0}_{u}=0\cdot+\cdot 1_{\complexN}+(u+u^{\ast})$ and $\de{0}=0\cdot+\cdot 1_{\complexN}+0$.

\begin{itemize}
\item[]\textbf{The project $\de{c}$ is an element of $\cond{(A^{\pol}+B^{\pol})^{\pol}}$.}~\\
Let $\de{d}\in\cond{A^{\pol}+B^{\pol}}$. Then $\de{d}=\de{a'_{\text{p+q}}+b'_{\text{p+q}}}$ where $\de{a'}\in\cond{A^{\pol}}$ and $\de{b'}\in\cond{B^{\pol}}$. We can then compute $\sca[]{c}{d}$ (using the fact that $\de{a',b'}$ are elements of a negative conduct and therefore have a null wager, we notice that $\sca[]{0}{a'_{\text{p+q}}}=\sca[]{0}{b'_{\text{p+q}}}=0$):
\begin{eqnarray*}
\sca[]{c}{d}&=&\sca[]{a_{\text{p+q}}+0_{\text{u}}+0}{a'_{\text{p+q}}+b'_{\text{p+q}}}\\
&=&\sca[]{a_{\text{p+q}}}{a'_{\text{p+q}}+b'_{\text{p+q}}}+\sca[]{0_{\text{u}}}{a'_{\text{p+q}}}+\dots\\
&&~~~~\dots+\sca[]{0}{a'_{\text{p+q}}}+\sca[]{0_{\text{u}}}{b'_{\text{p+q}}}+\sca[]{0}{b'_{\text{p+q}}}\\
&=&\sca[]{a_{\text{p+q}}}{a'_{\text{p+q}}+b'_{\text{p+q}}}+\sca[]{0_{\text{u}}}{a'_{\text{p+q}}}+\sca[]{0_{\text{u}}}{b'_{\text{p+q}}}
\end{eqnarray*}
As in the proof of Proposition \ref{contreexgdi52}, we show that both $\sca[]{0_{\text{u}}}{a'_{\text{p+q}}}$ and $\sca[]{0_{\text{u}}}{b'_{\text{p+q}}}$ are equal to zero since the wagers of $\de{a',b'}$ are equal to zero. Thus:
\begin{eqnarray*}
\sca[]{c}{d}&=&\sca[]{a_{\text{p+q}}}{a'_{\text{p+q}}}
\end{eqnarray*}
We can conclude, using the preceding lemma, that $\de{c}\poll\de{d}$.

\item[]\textbf{The project $\de{c}$ is not in $\cond{A\oplus B}$.}~\\
To show this, we find an element $\de{t}$ in $(\cond{A}_{p+q})^{\pol}\cap(\cond{B}_{p+q})^{\pol}$ such that $\de{c}\not\poll\de{t}$. Let us chose $\de{b'}\in\cond{B}^{\pol}$, $\de{a'}\in\cond{A}^{\pol}$, $\lambda\in\realposN$ with $\abs{\lambda}<1$ and $\de{0_{\text{$\lambda$u}}}=0\cdot+\cdot 1_{\complexN}+\lambda(u+u^{\ast})$. We will show that there exists a real number $\mu$ such that $\de{t}=\de{b'}_{p+q}+\de{a'_{\text{p+q}}}+\mu \de{0}_{\text{$\lambda$u}}-\mu\de{0}$ satisfies $\de{t}\in(\cond{A}_{p+q})^{\pol}$, $\de{t}\in(\cond{B}_{p+q})^{\pol}$ and $\de{t}\not\poll\de{c}$. Let us chose $\de{b}\in\cond{B}$; we can compute:
\begin{eqnarray*}
\sca[]{t}{b_{\text{p+q}}}&=&\sca[]{b'_{\text{p+q}}+a'_{\text{p+q}}+\mu 0_{\text{$\lambda$u}}-\mu 0}{b_{\text{p+q}}}\\
&=&\sca[]{b'_{\text{p+q}}+a'_{\text{p+q}}}{b_{\text{p+q}}}+\mu\sca[]{0_{\text{$\lambda$u}}}{b_{p+q}}-\mu\sca[]{0}{b_{p+q}}\\
&=&\sca[]{b'_{\text{p+q}}+a'_{\text{p+q}}}{b_{\text{p+q}}}+\mu\sca[]{0_{\text{$\lambda$u}}}{b_{p+q}}-\mu\sca[]{0}{b_{p+q}}
\end{eqnarray*}
Using the fact that $A'B=0$, we have $\sca[]{a'_{\text{p+q}}}{b_{\text{p+q}}}=\alpha'(1_{\vn{A}'})b$. Moreover, $\sca[]{0}{b_{p+q}}=b$. Thus:
\begin{eqnarray*}
\sca[]{0_{\text{$\lambda$u}}}{b_{p+q}}&=&b-\log(\det{}^{FK}(1-\lambda(u+u^{\ast})B))\\
&=&b-\log(\det{}^{FK}(1-\lambda uB+\lambda u^{\ast}pB))\\
&=&b-\log(\det{}^{FK}(1-\lambda uB))\\
&=&b-\log(1)\\
&=&b
\end{eqnarray*}
As a consequence, $\sca[]{t}{b_{\text{p+q}}}=\sca[]{b_{\text{p+q}}}{b'_{\text{p+q}}+a'_{\text{p+q}}}$, which shows (using the preceding lemma) that $\de{t}\in(\cond{B}_{p+q})^{\pol}$.

Let us now chose $\de{a}\in\cond{A}$. We can compute:
\begin{eqnarray*}
\sca[]{t}{a_{\text{p+q}}}&=&\sca[]{b'_{\text{p+q}}+a'_{\text{p+q}}+\mu 0_{\text{$\lambda$u}}-(1+\mu)0}{a_{\text{p+q}}}\\
&=&\sca[]{b'_{\text{p+q}}}{a_{\text{p+q}}}+\sca[]{a'_{\text{p+q}}}{a_{\text{p+q}}}+\dots\\
&&~~~~~\dots+\mu\sca[]{0_{\text{$\lambda$u}}}{a_{p+q}}-(1+\mu)\sca[]{0}{a_{p+q}}\\
&=&\sca[]{a}{a'}+\sca[]{b'_{\text{p+q}}}{a_{\text{p+q}}}+\dots\\
&&~~~~~\dots+\mu\sca[]{0_{\text{$\lambda$u}}}{a_{p+q}}-(1+\mu)\sca[]{0}{a_{p+q}}
\end{eqnarray*}
Since $AB'=0$, $\sca[]{b'_{\text{p+q}}}{a_{\text{p+q}}}=a$; similarly $\sca[]{0}{b_{p+q}}=a$.
Moreover:
\begin{eqnarray*}
\sca[]{0_{\text{$\lambda$u}}}{a_{p+q}}&=&a-\log(\det{}^{FK}(1-\lambda(u+u^{\ast})A))\\
&=&a-\log(\det{}^{FK}(1-\lambda up'A+\lambda u^{\ast}A))\\
&=&a-\log(\det{}^{FK}(1-\lambda u^{\ast}A))\\
&=&a-\log(1)\\
&=&a
\end{eqnarray*}
Thus $\sca[]{t}{a_{\text{p+q}}}=\sca[]{a}{a'}+a+\mu a-(1+\mu)a=\sca[]{a}{a'}$, which shows that $\de{t}\in(\cond{A}_{p+q})^{\pol}$.

We just showed that $\de{t}\in(\cond{A}_{p+q})^{\pol}\cap(\cond{B}_{p+q})^{\pol}$, which means that $\de{t}\in\cond{(A\oplus B)^{\pol}}$:
\begin{eqnarray*}
\cond{(A\oplus B)^{\pol}}&=&\cond{(A_{\text{p+q}}\cup B_{\text{p+q}})^{\pol\pol\pol}}\\
&=&\cond{(A_{\text{p+q}}\cup B_{\text{p+q}})^{\pol}}\\
&=&\cond{(A_{\text{p+q}})^{\pol}\cap (B_{\text{p+q}})^{\pol}}
\end{eqnarray*}

We can now compute $\sca[]{t}{c}$:
\begin{eqnarray*}
\sca[]{t}{c}&=&\sca[]{b'_{\text{p+q}}+a'_{\text{p+q}}+\mu0_{\text{$\lambda$u}}-\mu 0}{a_{\text{p+q}}+0_{\text{u}}-0}\\
&=&\sca[]{b'_{\text{p+q}}+a'_{\text{p+q}}}{a_{\text{p+q}}}+\sca[]{\mu0_{\text{$\lambda$u}}-\mu 0}{0_{\text{u}}-0}\\
&=&\sca[]{b'_{\text{p+q}}+a'_{\text{p+q}}}{a_{\text{p+q}}}+\mu\sca[]{0_{\text{$\lambda$u}}}{0_{\text{u}}}
\end{eqnarray*}

The computation of the third term gives us:
\begin{eqnarray*}
\sca[]{0_{\text{$\lambda$u}}}{0_{\text{u}}}&=&-\log(\det{}^{FK}(1-\lambda(u+u^{\ast})(u+u^{\ast})))\\
&=&-2\tr(p)\log(1-\lambda)
\end{eqnarray*}
As a consequence:
\begin{eqnarray*}
\sca[]{t}{c}&=&\sca[]{b'_{\text{p+q}}+a'_{\text{p+q}}}{a_{\text{p+q}}}-2\mu \tr(p)\log(1-\lambda)
\end{eqnarray*}
By chosing $\mu=\frac{\sca[]{b'_{\text{p+q}}+a'_{\text{p+q}}}{a_{\text{p+q}}}}{2\tr(p)\log(1-\lambda)}$ we have $\sca[]{t}{c}=0$.
\end{itemize}
Finally, we just showed that $\de{c}\in\cond{(A^{\pol}+B^{\pol})^{\pol}}$ and $\de{c}\not\in\cond{(A_{\text{p+q}}\cap B_{\text{p+q}})^{\pol\pol}}$ — i.e.  $\de{c}\not\in\cond{A\oplus B}$.
\end{proof}

\subsection{Discussion about the preceding results}

It is now necessary to discuss the obtained result. Indeed, one must understand to what extent this counter-example to Girard's proposition 16 changes the results of his paper.

The construction of additive connectives described in Section \ref{gdisection} corresponds to the alternative construction of additives in the hyperfinite setting (GoI5.2). These constructions, and the proofs of their properties, use only the adjunction for graphs and some combinatorial properties of the formal sum defined on graphs and extended to projects. It turns out that hyperfinite projects satisfy all those properties. Hence, the proofs of these Propositions yield proofs in Girard's setting by simply understanding the projects as hyperfinite projects. We can conclude that the construction described in Section \ref{gdisection} can be mimicked step by step with hyperfinite projects to yield an additive construction having the exact same properties. Notice however that one should prove formally that the trefoil property holds in this case in order to ensure that the results of Section \ref{denotsection} also hold in this case.

The case of the initial construction of additives (GoI5.1) described by Girard in his paper \cite{goi5}, is more complicated. Once again, the adjunction is satisfied by strict hyperfinite projects, but one should be careful when applying the combinatorial properties of the formal sum which may yield non-strict hyperfinite projects. We believe that with some care one should be able to show that the proofs of the propositions obtained in Section \ref{gdisection} can be adapted to this setting. We will however not attempt to do so in this paper, believing that such results extend beyond the scope of the paper.

\section{Conclusion}

\subsection{Results}

Generalising the first GoI model introduced by Girard \cite{multiplicatives}, we were able to define a graph-theoretic geometry of interaction in which one can interpret the Multiplicative Additive fragment of Linear Logic. Contrary to what happens in the two other versions of GoI dealing with additives \cite{goi3,goi5}, proofs of MALL are interpreted in our framework by \emph{finite objects}. 

Moreover, we were able to define an internal notion of observational equivalence with which we were able to solve the usual issue when dealing with additives in GoI: the connective $\with$ is not a product. We were then able to show that one can obtain, from our constructions, categorical models of MALL (with additive units) where no connectives and units are equal and in which neither the mix rule nor the weakening rule are satisfied. These models are moreover obtained as subcategories of $\ast$-autonomous categories, i.e. models of MLL with units. All these results moreover rely on a single geometric property we called the \emph{trefoil property}.

All these constructions being parametrised by a choice of a "measuring map" from $]0,1]\rightarrow\mathbb{R}_{\geqslant 0}\cup\{\infty\}$, we looked more closely at the construction in two particular cases. It can be shown that a first choice ($m(x)=-\log(1-x)$) defines a combinatorial version of the Multiplicative Additive fragment of Girard's GoI5 \cite{goi5}. It therefore gives insights on the notion of orthogonality used by Girard, and his use of \emph{idioms} — which corresponds in our setting to slices — and \emph{pseudo-trace} — which corresponds in our setting to the weights associated to the slices. On the other hand, a second choice of map ($m(x)=\infty$) defines a refined version of the Multiplicative fragment of the first versions of GoI where orthogonality was defined as nilpotency. However, this choice of parameter yields a trivial model of the additives. Nonetheless, our construction makes a bridge between "old-style" geometry of interaction and Girard's most recent work \cite{goi5}.

These results unveil the trefoil property as a fundamental identity upon which all Girard's GoI constructions are founded. Indeed, the adjunctions obtained from the determinant, from nilpotency or from the cyclicity in case of permutations are all special cases of the trefoil property.

\subsection{Perspectives}
This work opens a number of perspectives. Among these, we think the following two are of great interest and importance: extending the results to deal with exponential connectives, and obtain a better understanding of the trefoil property.

The extension to exponentials is the object of an upcoming paper. Although we know how to extend the setting of graphs to deal with exponentials \cite{seiller-phd}, some work is still needed to understand the relations between the obtained construction and Girard's numerous GoIs. This work seems of great interest, in particular when it comes to the study of computational complexity. We should be able to obtain characterisations of complexity classes as particular sets of graphs in the same way we recently obtained characterisations of the classes \textbf{co-NL} and \textbf{L} as sets of operators \cite{seiller-conl,seiller-lsp}. Moreover, this new framework can help gain intuitions about the restrictions on exponential connectives in systems like Elementay Linear Logic or Light Linear Logic \cite{lll} since we naturally obtain a restrained system as in Girard's GoI5.

Better understanding of the trefoil property may be used to help  obtain new GoI constructions. As a matter of fact, we believe that this property should be satisfied in some (if not all) model categories. We would therefore like to characterise the kinds of situations in which one can perform  GoI constructions, in the same way model categories characterise the kinds of situations where one can define and study homotopy. The interest of such an axiomatization would be to obtain new GoI constructions, in various domains of mathematics such as algebraic geometry, without much effort. As the operator-theoretic construction of the crossed product of an algebra and a group acting on it was used to obtain computational complexity results \cite{normativity,seiller-conl,seiller-lsp}, we believe one could use constructions and invariants specific to other mathematical fields to obtain new results and/or gain a better understanding of computation.

\bibliography{goiadd-revised}
\bibliographystyle{elsarticle-harv}

\end{document}